\documentclass[draftcls,onecolumn,12pt]{IEEEtran}

\usepackage{amsmath, amsthm,amsfonts, amssymb, float, enumerate, bm, graphicx, mathtools, mathrsfs,hyperref, latexsym}
\usepackage{nicefrac}

\DeclarePairedDelimiter{\floor}{\lfloor}{\rfloor}

\DeclareMathOperator*{\argmin}{arg\,min}
\DeclareMathOperator*{\argmax}{arg\,max}

\usepackage{tikz}
\usetikzlibrary{arrows, backgrounds,matrix,patterns,calc}
\usepackage{verbatim}
\tikzstyle{int}=[draw, fill=white!20, minimum size=2em]
\tikzstyle{init} = [pin edge={to-,thin,black}]
\usetikzlibrary{datavisualization.formats.functions}

\usepackage{pgfplots}
\pgfplotsset{compat=newest}
\usepgfplotslibrary{fillbetween}

\restylefloat{table}
\newcommand{\bd}{{\sf d}}

\newcommand{\khat}{\widehat{k}}

\newcommand{\Typep}{T_{p,\epsilon}^{n}}
\newcommand{\entrX}{H(\prbX)}
\newcommand{\entrP}{H(p)}
\newcommand{\prbX}{ P_{\sf X}}
\newcommand{\rateD}{R({\sf d})}
\newcommand{\Typepx}{T_{\prbX,\epsilon}^{n}}
\newcommand{\Pcode}{P_0}
\newcommand{\Gtilde}{\widetilde{G}}
\newcommand{\Vtilde}{\widetilde{V}}
\newcommand{\Etilde}{\widetilde{E}}
\newcommand{\Gsdir}{\Gtilde_{{\sf s},\delta}}

\newcommand{\types}{\Pscr_n(\Xscr)}
\newcommand{\distr}{\Pscr(\Xscr)}
\newcommand{\xvec}{\boldsymbol{x}}%{\textbf{\textit{x}}}
\newcommand{\yvec}{\boldsymbol{y}} %{\textbf{\textit{y}}}
\newcommand{\zvec}{\boldsymbol{z}} %{\textbf{\textit{z}}}

\newcommand{\Xvec}{\boldsymbol{X}} %{\textbf{\textit{X}}}
\newcommand{\Yvec}{\boldsymbol{Y}} %{\textbf{\textit{Y}}}
\newcommand{\Zvec}{\boldsymbol{Z}} %{\textbf{\textit{Z}}}

\newcounter{eg}[section]
\renewcommand{\theeg}{\arabic{section}.\arabic{eg}}
\newenvironment{examp}[1][]{\refstepcounter{eg}
\par\medskip \noindent
   \textit{Example~\theeg. #1} \rmfamily}{\hfill $\square$   \hspace{-4.5pt} \vspace{6pt}}

 \newcounter{rem}[section]

\usepackage{amsmath, amssymb, bbm, xspace}
\usepackage{epsfig}
\usepackage{longtable}
\usepackage{color}
\usepackage{mathrsfs}
\usepackage{comment}
\usepackage{ifthen}
\newboolean{showcomments}
\setboolean{showcomments}{true}
\usepackage{courier}
\usepackage{xcolor}
\usepackage{todonotes}
\usepackage[framemethod=tikz]{mdframed}
\usepackage{lineno}
\newmdenv[leftmargin=\dimexpr-0.4em, innerleftmargin=0.5em,
rightmargin=\dimexpr-0.4em, innerrightmargin=0.5em,
linewidth=2pt,linecolor=red, topline=false, bottomline=false,
innertopmargin=0pt,innerbottommargin=0pt,skipbelow=0pt,skipabove=0pt,%
]{notex}

\newenvironment{note}%
{\vskip\dimexpr\dp\strutbox-\prevdepth\relax\notex\strut\ignorespaces}%
{\xdef\notetpd{\the\prevdepth}\endnotex\vskip-\notetpd\relax}

\let\oldtodo\todo

\makeatletter%
\DeclareDocumentCommand{\todo}{ O{} +g +d<> }{%
		\setlength{\marginparwidth}{1.5cm}%
	\IfNoValueTF{#2}{\relax}{%
		\oldtodo[caption={#2},size=\scriptsize,#1]{\renewcommand{\baselinestretch}{0.8}\selectfont\sffamily#2\par}%
	}%
	\IfNoValueTF{#3}{\relax}{%
		\IfNoValueTF{#2}{% when parmark but no todo
			\begin{note}%
				\begin{internallinenumbers}%
					\indent%
					#3%
				\end{internallinenumbers}%
			\end{note}%
		}{% when parmark and todo
			\vspace{-0\baselineskip}%
			\begin{note}%
				\begin{internallinenumbers}%
					\indent%
					#3%
				\end{internallinenumbers}%
			\end{note}%
		}%
	}%
}%
\makeatother
\usepackage{soul}

\newcommand{\hlc}[2][yellow]{{%
		\colorlet{foo}{#1}%
		\sethlcolor{foo}\hl{#2}}%
}

\newcommand{\removetodo}[2]{\todo[color=pink]{\textbf{delete:} ``#1'' #2}\hlc[pink]{#1}}
\newcommand{\inserttodo}[1]{\todo[color=green!40]{\textbf{insert:} #1}}

\newcommand{\hltodoy}[2]{\todo[color=yellow!40]{#2}\hl{#1} }

\newcommand{\hltodoc}[3]{\todo[color=#3!40]{#2}\hlc[#3]{#1} }

\newcommand{\hltodo}[2]{\todo[color=orange!40]{#2}\hlc[orange!40]{#1} }
\newcommand{\replacetodo}[2]{\todo[color=pink!40]{\textbf{replace with:}``#2'' }\hl{#1} }

\usepackage{marginnote}
\newcommand{\todol}[1]{{%
		\let\marginpar\marginnote
		\reversemarginpar
		\renewcommand{\baselinestretch}{0.8}%
		\todo{#1}}}

\newcommand{\inserttodol}[1]{{%
		\let\marginpar\marginnote
		\reversemarginpar
		\renewcommand{\baselinestretch}{0.8}%
		\inserttodo{#1}}}

\newcommand{\removetodol}[2]{{%
		\let\marginpar\marginnote
		\reversemarginpar
		\renewcommand{\baselinestretch}{0.8}%
		\removetodo{#1}{#2}}}

\newcommand{\hltodol}[2]{{%
		\let\marginpar\marginnote
		\reversemarginpar
		\renewcommand{\baselinestretch}{0.8}%
		\hltodo{#1}{#2}}}

\newcommand{\replacetodol}[2]{{%
		\let\marginpar\marginnote
		\reversemarginpar
		\renewcommand{\baselinestretch}{0.8}%
		\replacetodo{#1}{#2}}}

\newcommand{\hltodoyl}[2]{{%
		\let\marginpar\marginnote
		\reversemarginpar
		\renewcommand{\baselinestretch}{0.8}%
		\hltodoy{#1}{#2}}}

\newcommand{\hltodocl}[3]{{		\let\marginpar\marginnote
		\reversemarginpar
		\renewcommand{\baselinestretch}{0.8}%
		\hltodoc{#1}{#2}{#3}}}

\newtheorem{theorem}{Theorem}[section]

\newtheorem{lemma}[theorem]{Lemma}
\newtheorem{claim}[theorem]{Claim}
\newtheorem{proposition}[theorem]{Proposition}

\newtheorem{definition}{Definition}[section]

\def\bkE{{\rm I\kern-.17em E}}
\def\bk1{{\rm 1\kern-.17em l}}
\def\bkD{{\rm I\kern-.17em D}}
\def\bkR{{\rm I\kern-.17em R}}
\def\bkP{{\rm I\kern-.17em P}}

\def\bkZ{{\bf{Z}}}

\def\bkE{{\rm I\kern-.17em E}}
\def\bk1{{\rm 1\kern-.17em l}}
\def\bkD{{\rm I\kern-.17em D}}
\def\bkR{{\rm I\kern-.17em R}}
\def\bkP{{\rm I\kern-.17em P}}

\makeatletter
\newcommand{\pushright}[1]{\ifmeasuring@#1\else\omit\hfill$\displaystyle#1$\fi\ignorespaces}
\newcommand{\pushleft}[1]{\ifmeasuring@#1\else\omit$\displaystyle#1$\hfill\fi\ignorespaces}
\makeatother

\def\bkZ{{\bf{Z}}}
\def\b12{(\beta_1,\beta_2)}

\newenvironment{example}{{\noindent \bf Example}}{\hfill $\square$\hspace{-4.5pt}\vspace{6pt}}
\newcounter{example}
\renewcommand{\theexample}{\thesection.\arabic{example}}

\newcounter{remark}
\renewcommand{\theremark}{\thesection.\arabic{remark}}

\def\Bscr{\mathscr{B}}
\def\Xscr{\mathcal{X}}
\def\Yscr{\mathcal{Y}}

\def\Ebb{\mathbb{E}}
\newlength{\noteWidth}
\long\def\notes#1{\ifinner
{\tiny #1}
\else
\marginpar{\parbox[t]{\noteWidth}{\raggedright\tiny #1}}
\fi\typeout{#1}}

 \def\notes#1{\typeout{read notes: #1}} %uncomment for final version

\newcommand{\wi}[1]{\widehat{#1}}

\DeclareMathOperator{\image}{Im}

\newcommand{\round}[1]{\left(#1\right)}

\newcommand{\Gs}{G_{\mathsf{s}}}

\newcommand{\ut}{\mathscr{U}}

\newcommand{\best}{\mathscr{B}}%{\Bscr}
\newcommand{\pbest}{\mathscr{S}}%{\Sscr}

\newcommand{\Ihat}{\widehat{I}}

\newcommand{\ie}{i.e.\@\xspace} %%% i.e.,
\newcommand{\Real}{\ensuremath{\mathbb{R}}}

\newcommand{\maximize}[1]{\displaystyle\maxim_{#1}}
\newcommand{\maxim}{\mathop{\hbox{\rm max}}}

\def\OPT{{\rm OPT}}

\def\rarr{\rightarrow}

\def\Ebb{\mathbb{E}}

\def\Pbb{{\mathbb{P}}}
\def\Nbb{{\mathbb{N}}}

\def\spose#1{\hbox to 0pt{#1\hss}}
\def\sub#1{^{\null}_{#1}}
\def\text #1{\hbox{\quad#1\quad}}

\def\Escr{\mathcal{E}}

\def\nthinsp{\mskip -2   mu}

\def\superstar{^{\raise 0.5pt\hbox{$\nthinsp *$}}}
\def\SUPERSTAR{^{\raise 0.5pt\hbox{$*$}}}
\def\lamstarT {\lambda^{\raise 0.5pt\hbox{$\nthinsp *$}T}}

\def\Ascr{{\cal A}}
\def\Bscr{{\cal B}}
\def\Fscr{{\cal F}}
\def\Dscr{{\cal D}}
\def\Iscr{{\cal I}}

\def\Mscr{{\cal M}}

\def\Oscr{{\cal O}}
\def\Pscr{{\cal P}}
\def\Qscr{{\cal Q}}

\def\Wscr{{\cal W}}
\def\Mscr{{\cal M}}

\def\Rscr{{\cal R}}

\def\Cscr{{\cal C}}
\def\Zscr{{\cal Z}}
\def\Xscr{{\cal X}}
\def\Yscr{{\cal Y}}
\def\Kscr{{\cal K}}

\def\dhat{\widehat d}

\def\Dbar{\bar D}

\def\gbar{\skew{4.3}\bar g}
\def\ghat{\skew{4.3}\widehat g}
\def\gtilde{\skew{4.5}\widetilde g}

\def\mhat{\widehat m}

\def\pbar{\skew2\bar p}
\def\phat{\skew2\widehat p}

\def\Pbar{\skew5\bar P}

\def\Phat{\widehat P}
\def\Ptilde{\skew5\widetilde P}

\def\Rhat{\widehat R}
\def\Rtilde{\widetilde R}
\def\sbar{\bar s}
\def\shat{\widehat s}
\def\stilde{\widetilde s}

\def\Uhat{\widehat U}

\def\Vtilde{\widetilde V}

\def\Xhat{\widehat X}

\def\supp{{\rm supp}}

\def\non{\nonumber}

\let\forallnew\forall
\renewcommand{\forall}{\forallnew\ }
\let\forall\forallnew

		\def\bkE{{\rm I\kern-.17em E}}
		\def\bk1{{\rm 1\kern-.17em l}}
		\def\bkD{{\rm I\kern-.17em D}}
		\def\bkR{{\rm I\kern-.17em R}}
		\def\bkP{{\rm I\kern-.17em P}}
		\def\bkY{{\bf \kern-.17em Y}}
		\def\bkZ{{\bf \kern-.17em Z}}
		\def\bkC{{\bf  \kern-.17em C}}
{\begin{list}{}%
         {\setlength{\leftmargin}{#1}}%
         \item[]%
}
{\end{list}}

		\def\bsp{\begin{split}}
		\def\beq{\begin{eqnarray}}
		\def\bal{\begin{align*}}
		\def\bc{\begin{center}}
		\def\be{\begin{enumerate}}
		\def\bi{\begin{itemize}}
		\def\bs{\begin{small}}
		\def\bS{\begin{slide}}
		\def\ec{\end{center}}
		\def\ee{\end{enumerate}}
		\def\ei{\end{itemize}}
		\def\es{\end{small}}
		\def\eS{\end{slide}}
		\def\eeq{\end{eqnarray}}
		\def\eal{\end{align*}}
		\def\esp{\end{split}}
		\def\qed{ \vrule height7.5pt width7.5pt depth0pt}  %width4.17pt depth0pt}

	\def\maxproblemsmall#1#2#3#4{\fbox
		 {\begin{tabular*}{0.47\textwidth}
			{@{}l@{\extracolsep{\fill}}l@{\extracolsep{6pt}}l@{\extracolsep{\fill}}c@{}}
				#1 & $\maximize{#2}$ & $#3$ & $ $ \\[5pt]
					 & $\subject\ $    & $#4$ & $ $
			\end{tabular*}}
			}

	\def\cp2problem#1#2#3#4{\fbox
		 {\begin{tabular*}{0.9\textwidth}
			{@{}l@{\extracolsep{\fill}}l@{\extracolsep{6pt}}l@{\extracolsep{\fill}}c@{}}
				#1 & & $#4 $
			\end{tabular*}}}

		\def\bkE{{\rm I\kern-.17em E}}
		\def\bk1{{\rm 1\kern-.17em l}}
		\def\bkD{{\rm I\kern-.17em D}}
		\def\bkR{{\rm I\kern-.17em R}}
		\def\bkP{{\rm I\kern-.17em P}}

		\def\bkZ{{\bf{Z}}}

\newcommand {\beeq}[1]{\begin{equation}\label{#1}}
\newcommand {\eeeq}{\end{equation}}
\newcommand {\bea}{\begin{eqnarray}}
\newcommand {\eea}{\end{eqnarray}}

\def\texitem#1{\par\smallskip\noindent\hangindent 25pt
               \hbox to 25pt {\hss #1 ~}\ignorespaces}

\def\bsp{\begin{split}}
		\def\beq{\begin{eqnarray}}
		\def\bal{\begin{align*}}
		\def\bc{\begin{center}}
		\def\be{\begin{enumerate}}
		\def\bi{\begin{itemize}}
		\def\bs{\begin{small}}
		\def\bS{\begin{slide}}
		\def\ec{\end{center}}
		\def\ee{\end{enumerate}}
		\def\ei{\end{itemize}}
		\def\es{\end{small}}
		\def\eS{\end{slide}}
		\def\eeq{\end{eqnarray}}
		\def\eal{\end{align*}}
		\def\esp{\end{split}}
		\def\qed{ \vrule height7.5pt width7.5pt depth0pt}  %width4.17pt depth0pt}

\usepackage{amsmath, amssymb, xspace}
\usepackage{epsfig}
\usepackage{longtable}
\usepackage{color}
\usepackage{mathrsfs}
\usepackage{subfig}
\def\Cscr{{\cal C}}

\def\sub{\hbox{\rm s.t}}

	\def\maxproblemsmall#1#2#3#4{\fbox
		 {\begin{tabular*}{0.47\textwidth}
			{@{}l@{\extracolsep{\fill}}l@{\extracolsep{-4pt}}l@{\extracolsep{\fill}}c@{}}
				#1 &  & $\maximize {#2}$ $#3$ & $ $ \\[4pt]
					 $\sub \ $  &   & $#4$ & $ $
			\end{tabular*}}
                    }

\author{\hspace{0.5cm} Anuj S. Vora, \and Ankur A. Kulkarni\thanks {A. S. Vora is with the Delft Center for Systems and Control group at Delft University of Technology, 2628CD Delft. A. A. Kulkarni is with the Systems and Control Engineering group at the Indian Institute of Technology Bombay, Mumbai, 400076. This work was done when A. S. Vora was a doctoral student at the Systems and Control Engineering group. They can be reached at \texttt{a.vora@tudelft.nl},  \texttt{kulkarni.ankur@iitb.ac.in}. Preliminary results of this paper were presented at the NCC 2020\cite{vora2020communicating} and ISIT 2020~\cite{vora2020achievable}.} }

\title{Achievable Rates for Information Extraction from a Strategic Sender}

\begin{document}

\maketitle

\begin{abstract}

   We consider a setting of non-cooperative communication where a receiver wants to recover randomly generated sequences of symbols that are observed by a strategic sender. The sender aims to maximize an average utility that may not align with the recovery criterion of the receiver, whereby the signals it sends may not be truthful. The rate of communication is defined as the number of reconstructions  corresponding to the sequences recovered correctly while communicating with the sender. We pose this problem as a sequential game between the sender and the receiver with the receiver as the leader and determine strategies for the receiver that attain vanishing probability of error and compute the rates of such strategies. We show the existence of such strategies under a condition on the utility of the sender. For the case of the binary alphabet, this condition is also necessary, in the absence of which, the probability of error goes to one for all choices of strategies of the receiver.  We show that for reliable recovery, the receiver chooses to correctly decode only a \textit{subset} of messages received from the sender and deliberately makes an error on messages outside this subset.
  Despite a clean channel, our setting exhibits a non-trivial \textit{maximum} rate of communication, which is in general strictly less than the capacity of the channel. This implies the impossibility of strategies that correctly decode sequences of rate greater than the maximum rate while also achieving reliable communication. This is a key point of departure from the usual setting of cooperative communication.

\end{abstract}

%Keywords: Information theory, game theory, Stackelberg game, non-cooperative communication,

%limitations arise due to strategic nature
%differences from mismatched coding  highlighted?

\section{Introduction}

Consider a setting with a sender and a receiver, where the sender observes a sequence of randomly generated symbols that is to be recovered subject to a fidelity criterion by the receiver. % The sender sends messages to the receiver via a clean channel.
However, unlike the usual setting of communication, the sender has a tendency to \textit{misreport} its information. %The receiver recovers a sequence based on the messages sent by the sender and tries to minimize the error in communication.
The sender signals the receiver with an aim of maximizing a \textit{utility function} that depends on the observed sequence and the sequence recovered by the receiver. Since maximizing this utility may not align with the interests of the receiver, the sender's signals may not be truthful. How should the receiver then communicate with the sender? We call this setting as the problem of \textit{information extraction} from a strategic sender.  In this paper, we ascertain the rates achievable for reliable communication in this setting.% given the sender's tendency to misreport its information.%Our primary investigation in this paper is centered on characterizing the rate regime for reliable communication.%of communication.% We develop a theory that establishes achievable rates for reliable communication.

Our setting is motivated by networked multi-agent systems such as cyber-physical systems. These systems often comprise diverse entities like sensors, controllers, and smart devices, that are remotely connected via communication channels. For example the sender could be a fusion center, that collects information from the sensors and transmits it  via a channel to a receiver, such as a controller. However, such a system could be attacked by an adversary which may cause the sender to act maliciously and misrepresent its information. Since these networked systems form the backbone of critical infrastructure, it is increasingly relevant to study the strategic setting.

%This unconventional setting between a strategic sender and the receiver is increasingly relevant due to the proliferation of networked-based multi-agent  systems such as cyber-physical systems. These systems comprise of multiple entities like sensors, controllers and smart devices, that are remotely connected via communication channels. In a typical scenario, a sender such as a fusion center collects information observed by the sensors and transmits it via a communication channel to a controller. The controller then takes appropriate actions to achieve a certain objective. However, such networked systems are vulnerable to adversarial attacks which may cause the sender to act maliciously and misrepresent its information. These networked systems form the backbone of critical infrastructures whereby, it is important to study this problem and determine ways by which the receiver can ensure information extraction from the sender

%Fundamentally, the setting considered is a setting of mismatch between the objectives of the sender and the receiver. This has received a considerable amount of attention in the literature (see~\cite{scarlett2020information} for a survey)

From the outset, one may relate this setting to the problem of coding in the presence of mismatched criteria studied extensively in the information theory literature (see~\cite{scarlett2020information} for a survey). However, our setting is  different in the following manner.  {\color{black}   For the problem with mismatched criteria, given a codebook, the encoding and decoding functions are chosen to minimize the respective distortions. The objective is to devise a codebook that ensures reliable communication despite this mismatched criteria.
    Importantly, the encoder  is not chosen as a \textit{best response} to the decoding function of the receiver. In contrast, the receiver in our setting has to choose its decoding function while considering the  best responses of the sender to the receiver's strategy}. %\todo{check if this ok}

Our setting is also different from the setting where a sender reports arbitrarily corrupted information. Such a \textit{defunct} sender is indifferent to the eventual outcome and hence its signals to the receiver may be uncorrelated to the observed information. The sender in our setting seeks to maximize its utility function; it chooses to report certain sequences incorrectly when that suits this goal, whereas for other sequences it may choose to report truthfully. The messages of the strategic sender thus have a pattern dictated by its utility  that can be exploited by the receiver.

%Our setting also differs from the setting of a  \textit{defunct} sender that reports arbitrarily corrupted information. A defunct sender is indifferent and has no preference over the sequences recovered by the receiver and hence its output may be uncorrelated to the observed sequence. The strategic sender in our setting is driven by its utility function; it actively chooses to report certain sequences incorrectly, whereas for other sequences, it may choose to report truthfully. The messages transmitted by the strategic sender thus have a pattern that can be exploited by the receiver.

In this paper, we ask: what can the receiver achieve in the above strategic setting? The objective of the receiver is minimize the  probability of excess Hamming distortion averaged over the blocklength.
%In the case of the lossy recovery, the distortion level is a fixed positive number and for the lossless recovery, the distortion level is zero.
 We find that for both lossy and lossless recovery, by using the right amount of communication resources the receiver can ensure asymptotically vanishing error as the blocklength grows large. Our goal is to describe the communication resources used by the receiver in achieving this performance.

 \subsection{Main results}
 The channel between the sender and receiver is noiseless, but due to the strategic behaviour of the sender, our problem demands a new vantage point for analysis. We formulate this setting as a leader-follower game with the receiver as the leader and the sender as the follower. Thus, the receiver declares its decoding strategy and the sender follows with a best response. We determine decoding strategies for the receiver that achieve reliable communication, \ie, for any best response of the sender, yield asymptotically vanishing probability of error. %We also define a notion

%The setting of our paper requires a new vantage point for analysis since it is distinct from the usual coding problem; due to the non-cooperative nature of the sender its signals are not codewords, and the receiver cannot do a usual decoding. Instead, the receiver must follow strategic considerations, taking into account possible responses of the sender, which in turn are motivated by its utility. We formulate this setting as a leader-follower game with the receiver as the leader and the sender as the follower. In this setting, the receiver declares its decoding strategy and the sender follows with a best response signalling strategy. We determine the \textit{achievable sequence of strategies} for the receiver. The achievable sequence of strategies are defined as a sequence of decoding strategies that, along with the best response signalling strategy, achieve asymptotically vanishing probability of error. %We also define a notion of the \textit{achievable rates}.

{\color{black}Our results centre around characterizing the  \textit{rate of communication}. Our rate is the number of reconstructions  corresponding to the sequences recovered correctly while communicating with the sender. A rate $R$ is said to be achievable if there exist receiver strategies achieving reliable communication and with rate approaching $R$. %and the number of useful reconstructions approach the rate $R$.
	%have rate $R$ of correctly decoded sequences.
	Since we assume a clean channel, in case of cooperative communication, the sender and receiver can communicate reliably using channel inputs of rate greater than the rate-distortion function of the source and less than the channel capacity (\ie, logarithm of the alphabet). Consequently,  in the cooperative setting, the only quantity of interest is the \textit{minimum} achievable rate.
%A rate $ R $ is called an achievable rate if there exists an achievable sequence of strategies with the limit of rate of communication equal to $ R $.
On the contrary, in our setting, in addition to the minimum achievable rate, there is also a \textit{maximum} (technically, supremum) achievable rate which is relevant, and it is  in general strictly less than the channel capacity. Importantly, there do not exist receiver strategies that ensure reliable communication and utilize reconstructions of rate greater than this maximum achievable rate for correctly decoding sequences.}

We give a sufficient condition on the utility function of the sender for the achievable rate region to be nonempty. This condition is in terms of a single-letter optimization problem that can also be cast as a linear program. We then determine a subset of achievable rates under this sufficient condition. As mentioned above, we find that the set of achievable rates may be a \textit{strict subset} of the rate region for cooperative communication.
%We show that the achievable rate region is a convex set, whereby it suffices to determine the minimum and the maximum achievable rate.

For lossless recovery with binary alphabet, we show that the above condition on the utility is also necessary for  the existence of an achievable rate. Furthermore, in this case, we also have a \textit{strong converse}: if this condition does not hold, then the probability of error tends to one for all choices of strategies of the receiver. The rate region for binary alphabet depends on the structure of the utility as follows: when the sender prefers truthful reporting, it collapses to the Shannon region, but it shrinks as the degree of misreporting of the sender increases. %in this case demonstrates a dual nature, depending on the structure of the utility, it is either a strict subset of the Shannon rate region or it is equal to the Shannon rate region
For lossless recovery, the minimum achievable rate is the entropy of the source. For the maximum achievable rate we give an upper bound which, for a certain class of utilities, is strictly less than unity. For the case of lossy recovery with the binary alphabet, we show that the minimum achievable rate is the rate-distortion function  and determine a lower bound for the maximum achievable rate.

For lossless recovery with general alphabet too, this linear programming-based condition is necessary, except for a degenerate case.  As in the binary case, the entropy is the minimum achievable rate. We further show that this condition is  sufficient for lossy recovery and we only determine an interval of achievable rates around the entropy. This also yields an upper bound for the minimum achievable rate. It remains open to ascertain any (upper or lower) estimates on the maximum achievable rate for a general alphabet.

Throughout, the optimal strategy for the receiver is to choose a decoding function that allows the correct recovery of  only a \textit{subset} of the sequences and deliberately induces an error on the rest of the sequences. The dilemma for the receiver is to choose a decoding function  that ensures that the recovered set is large enough in terms of probability measure, and at the same time does not  give too much room to the sender to lie as a best response. {\color{black} This strategy restricts the choices of the sender and forces it to  report truthfully for a set of sequences  having high probability measure and thereby achieving vanishing probability of error. In fact, we show that it is sufficient for the receiver to use a strategy that identically decodes the sequences having the type closest to the probability of the source. For the sequences outside this type class, it decodes to an arbitrary sequence inside the type class. Moreover, if the receiver attempts to also recover sequences having a sufficiently distant type, then this gives the sender more freedom to lie about its information and increases the probability of error}. As a result, despite the communication medium being a clean channel,
%An interesting aspect of our setting is that we assume that the channel input and output spaces are both equal to the space of source sequences. In the cooperative setting, in the noiseless case, this would trivially lead to recovery of all the source sequences. However, we show that the same does not hold in our setting since
the receiver chooses to use only a subset of channel resources by the way of a selective decoding strategy. This dilemma of the receiver of choosing the right subset of sequences manifests as the minimum and maximum achievable rates.
%The strategic setting thus gives rise to the idea of \textit{capacity utilization} that is distinct from the traditional notion of the capacity of a channel.

%In a sense, this marks a shift from the traditional concept of the capacity of a channel to that of \textit{capacity utilization}, as something more relevant for the non-cooperative setting. We believe these are glimpses into the fascinating world of results that lie in this non-cooperative communication theoretic regime and we hope to explore these in future work.
%The concept of maximum rate is akin to the notion of channel capacity. However, it is left to see whether there is strong converse type result.

\subsection{Related work}
\label{sec:related_work}

Numerous studies have explored the setting of strategic communication in different contexts. To the best of our knowledge ours is the first formal information-theoretic analysis of information extraction. The seminal work of Crawford and Sobel \cite{crawford1982strategic} introduced the setting of strategic information transmission between a sender and a receiver with misaligned objectives. They formulated the problem as a simultaneous move game and showed that any equilibrium involves the sender resorting to a \textit{quantization} strategy, where the sender reports only the interval in which its information lies. The variants and generalizations of these work were subsequently studied in (Battaglini \cite{battaglini2002multiple}, Sarita{\c{s}} \textit{et al.} \cite{saritacs2015multi}, Kaz{\i}kl{\i} \textit{et al.} \cite{kazikli2021signaling}). These works considered the Nash equilibrium solution of the game. The works by Farokhi \textit{et al.} in \cite{farokhi2016estimation} and Sayin \textit{et al.} in \cite{sayin2019hierarchical} have explored the setting of strategic communication in the context of control theory. In \cite{farokhi2016estimation} the authors studied a problem of static and dynamic estimation as a game between the sensors and a receiver, where the sensors are strategic. The authors in \cite{sayin2019hierarchical} considered a dynamic signalling game between a strategic sender and a receiver. An information-theoretic perspective of strategic communication was studied by Akyol \textit{et al.} in \cite{akyol2015privacy, akyol2016information} where they considered a sender-receiver game and characterized the equilibria that achieve certain rate and distortion levels. In the economics literature, the setting of the Bayesian persuasion problem studied by Kamenica and Gentzkow in \cite{kamenica2011bayesian} and Bergemann and Morris \cite{bergemann2019information}, is a strategic communication game, where the informed sender shapes the information of the receiver to influence its actions. Le Treust and Tomala in \cite{letruest2019persuasion} studied a Bayesian persuasion problem where they studied the payoffs of the sender while communicating with a receiver via a noisy channel. Rouphael and Le Treust in \cite{rouphael2023strategic} and \cite{rouphael2022strategic} study problems of information-theoretic nature in the context of a strategic setting. In \cite{rouphael2023strategic}, they study the setting of coding for successive refinement in the strategic setting with one sender and two receivers. In  \cite{rouphael2022strategic}, they study the problem of strategic communication between an encoder and two decoders communicating over a Gray-Wyner network.  Deori and Kulkarni in \cite{deori22information} study a setting of \textit{information revelation} where the authors characterize the  minimum number of source symbols that can be recovered by a receiver from a strategic sender in an equilibrium. The works \cite{farokhi2016estimation}-\cite{deori22information} have formulated the game with the sender as the leader.

 Our work differs from the above models as follows. We study the problem from the perspective of the receiver and hence we formulate the game with the receiver as the leader. We consider a model where the malicious behaviour of the sender is explicitly governed by a utility function. Moreover, we characterize the necessary and sufficient conditions on the utility that allows the receiver to communicate reliably with the sender.%the set of achievable rates

 %Our main contribution is the notion of information extraction capacity and the lower and upper bounds on this quantity. %We show that under certain cases, the maximum number of perfectly recoverable sequences in the strategic setting can be strictly smaller than the Shannon limit.

Related to our problem is the problem discussed by Jackson and Sonnenschein in \cite{jackson2007overcoming}. They design a mechanism to implement a function where they \textit{link} independent copies of a decision problem. They show that as $ n $ grows large, the mechanism overcomes the incentive constraints and implements the function asymptotically. This linking is akin to the block structure of our setting and implementation of the function is analogous to information recovery by the receiver. However, they do not study information-theoretic aspects such as rate of communication.

Problems with mismatched criteria are studied extensively in the information theory literature \cite{scarlett2020information}. The problem relevant to our setting is the problem of mismatched distortion studied by Lapidoth \cite{lapidoth1997role}. In this setting, the receiver and the sender have mismatched distortion criteria, and the receiver aims to construct a codebook to achieve reliable communication taking into account this mismatch. The author determines an upper bound on the distortion as a function of the rate.  % the minimum achievable rate to achieve asymptotically vanishing probability of error.
 %Although this setting seems to be similar to our setting, there are certain key differences.
  The sender in the mismatched problem is using a distortion function to encode its observation and this distortion measure is assumed to be a non-negative function. Whereas in our case, the sender uses a utility function that, in general, is real-valued function. Although the encoder in the mismatched problem is assumed to be strategic, its objective is not affected by the sequence chosen by the receiver. Moreover, the author aims to only determine the least possible rate that achieves reliable communication at a certain level of distortion. In our work, in addition to the least achievable rate, we also have a notion of the maximum achievable rate for reliable communication.

 {\color{black} Our work is a significant extension to the results of the conference version \cite{vora2020achievable} where we discussed the case of lossy recovery with the binary alphabet.  In this work, we only proved a sufficient condition on the utility for the existence of an achievable rate. %This condition is identical to the condition described in our current paper, but without the extensive implications.
 	We have also studied related problems of communication with a strategic sender in \cite{ vora2020communicating, vora2020zero, vora2020information}. In \cite{vora2020communicating}, we studied a rate limited setting, where the receiver wished to compute a function of the source. %We determined sufficient and necessary conditions for reliable communication and determined conditions under which the information-theoretic source code is an equilibrium of the game.
 	In \cite{vora2020zero, vora2020information,vora2024shannon} we studied a setting of information extraction where the receiver aimed for zero-error recovery of the source.
 	% Our point of investigation in these works is a notion that we define as the \textit{information extraction capacity}, which determines the maximum rate of the growth of perfectly recovered sequences from the strategic sender. A crucial concept in these settings is a sender graph\footnote{This graph is also useful for some of our results and is discussed in Section~\ref{append:lossl_delta=0_empty_R}.} which is a graph induced by the utility of the sender.
 	  In \cite{vora2020zero}, we consider a special structure of the utility
 	  %and show that that this notion of capacity is bounded above by the Shannon capacity of a certain graph. In
 	  and in \cite{vora2020information}, we study the zero-error setting in the presence of a noisy channel.
 	  % and characterize the information extraction capacity.
 	  We extend these results in \cite{vora2024shannon} %, where we show that the results presented in \cite{vora2020zero} holds
 	   for any general utility function. %We further show that this notion of capacity is a  generalization of the Shannon capacity of the graph and determine lower and upper bounds. In certain cases, the exact characterization is also obtained.
 	  This zero-error setting was extended to the case of multiple senders in \cite{vora2021optimal} and preliminary bounds on the rate of information extraction were obtained.

 	  In contrast to the zero-error recovery in  \cite{vora2020zero}-\cite{vora2021optimal}, we allow for asymptotically vanishing probability of error and even consider recovery with a finite distortion.
 	   A crucial concept in these settings is a sender graph %\footnote{ and is discussed in Section~\ref{append:lossl_delta=0_empty_R}.}
 	    which is a graph induced by the utility of the sender. Some of the results in our setting also use the structure of this graph.  These works characterize the maximum rate of the growth of   sequences recovered with zero-error while we seek a characterization of the achievable rates.
% 	     with a variation that we seek the largest set of symbols where this condition holds.
A condition similar to the sufficient condition on the utility in our paper  is used to characterize this growth.
}

 The paper is  organized as follows. We introduce the problem in Section~\ref{sec:prob_form} and main results are presented in  Section~\ref{sec:main} and \ref{sec:main_general_alph}. %, we present the main results of the paper. %Few examples discussing the implications of the results are discussed in Section~\ref{sec:exam}.
 Section~\ref{sec:concl} concludes the paper.

\section{Problem Formulation}
\label{sec:prob_form}
\subsection{Notation}

Random variables are denoted with upper case letters $\Xvec, \Yvec, \Zvec$ and their instances are denoted as lower case letters $\xvec, \yvec, \zvec$. The space of scalar random variables is denoted by calligraphic letters $\Xscr$ and the space of $n$-length vector random variables is denoted as $\Xscr^n$. The bold face upper case letters $\Xvec, \Yvec, \Zvec$ and their instances $ \xvec, \yvec, \zvec $ are vector valued and non-bold faced letters $X,Y,Z$ and $x,y,z$ are scalar random variables.  The probability of an event or a set $ A $ is denoted as $ \Pbb(A) $. The expectation of a random variable $X$ is denoted as $ \Ebb[X]$. The set of probability distributions on a space $ \Xscr $ is denoted as $\Pscr(\Xscr)$.  The type of a sequence $\xvec \in \Xscr^n$ is denoted by $ P_{\xvec} $ and the joint type of sequences $ \xvec,\yvec $ is denoted as $ P_{\xvec,\yvec} $. The set of all types induced by sequences in $ \Xscr^n $ is denoted by $ \Pscr_n(\Xscr) $. For a fixed $ n $ and $ \epsilon > 0 $, we define the typical set around a distribution $ P \in \Pscr(\Xscr) $ as $T_{P,\epsilon}^n  = \{ \xvec \in \Xscr^{n} :  P(i) - \epsilon < P_{\xvec}(i) < P(i) + \epsilon \; \forall \; i \in \Xscr \}$. For a type $P \in \Pscr_n(\Xscr)$, we define the type class   as
%\begin{align}
$U_{P}^n = \{\xvec \in \Xscr^n : P_{\xvec}(i) = P(i) \;\forall\;i\in \Xscr \}$. Matrices are denoted by uppercase letters $ V, W $. For a sequence $ \xvec $, the Hamming ball of sequences around $ \xvec $ with radius $ \delta$ is denoted as $ B_\delta(\xvec) $.  For a function or a random variable, we denote $ \supp(\cdot) $ as its support set. For a function $f$, $\image(f)$ denotes its image. For an optimization problem `$ \cdot $', we denote $\OPT(\cdot)$ as its optimal value.  For random variables $ (X,Y) $ we denote the mutual information  as $ I(X;Y) $. Entropy of a random variable $ X $ distributed as $ \prbX $ is denoted as $ H(X) $ or $ H(\prbX) $. For a binary random variable distributed according to Bern($p$), we write the entropy as  $ H(p) $. {\color{black} The rate distortion function for a distortion  $ \bd $ and a source $ X $ distributed as $ \prbX $ is given as $ R(\bd) =  \inf_{P_{Y|X}: \Ebb[d(X,Y)] \leq \bd} I(X;Y)$, where  $ (X,Y) \sim \prbX P_{Y|X} $.} A graph with vertices $ V $ and edges $ E $ is denoted as $ G = (V,E) $. For a pair of vertices  $ \xvec,\yvec \in  V $, we denote the adjacency by $ (\xvec,\yvec) \in E $ or $ \xvec \sim \yvec $. For a directed graph $ \Gtilde = (\Vtilde,\Etilde) $, we denote a directed edge from $ \xvec $ to $ \yvec $ as $ \xvec \leadsto \yvec $. We denote the interval between $ a, b \in \Real $ as $ [a,b]$. All the logarithms are  with base $2$.%%For any sequence $ x \in \Xscr^n $ and $ \bd $, we denote the $\bd-$radius ball around $ x $ as $ B_{\bd}(x) \{ x \in \Xscr^n : \}$
%\todo{add type class to this section}
% \begin{itemize}
%   \item $ B_\delta(\cdot)  $
% \end{itemize}

\subsection{Model formulation}

{\color{black}
\begin{figure}[h]
	\begin{center}
		\begin{tikzpicture}[node distance=2.5cm,auto,>=latex']
			\node [int] (enc2)  {$s_n$};
			\node [int] (cha)  [right of=enc2, node distance=2.5cm] {$g_n$};
			\node (start) [left of=enc2,node distance=2cm, coordinate] {};
			\node (end) [right of=cha, node distance=2cm] {};
			\node [below of=enc2, node distance = 0.8cm] { Sender};
			\node [below of=cha, node distance=0.8cm] { Receiver};
			\path[->] (start) edge node {$\Xvec$} (enc2);
			\path[->] (enc2) edge node {$\Yvec $} (cha);
			\path[->] (cha) edge node {$\wi{\Xvec}  $} (end);
		\end{tikzpicture}
	\end{center}
	\caption{Communication setup between the strategic sender and the receiver}
			\label{fig:comm_setup}
\end{figure}
}
Let $\Xscr = \{0,1,\hdots,q-1\}$ be a source alphabet, where $q \in \Nbb$ is the alphabet size. The sender observes a sequence of source symbols $\Xvec = (X_1,\hdots, X_n) \in \Xscr^n $, where $ X_k $ are i.i.d. according to a distribution $ \prbX \in \Pscr(\Xscr) $. The distribution of $ \Xvec $ is thus given as $ P_{\Xvec}(\xvec) = \prod_{i=1}^n \prbX(x_i) $.
 The sender transmits a message  $s_n(\Xvec) = \Yvec \in \Xscr^n$, where $s_n : \Xscr^n \rarr \Xscr^n$, as input to a channel. The channel input and output spaces are both $ \Xscr^n $. We assume that the channel is noiseless and hence the message is relayed perfectly to the receiver. The receiver decodes the message as $g_n(\Yvec) = \wi{\Xvec}$, where $g_n : \Xscr^n \rarr \Xscr^n $. {\color{black}The functions $s_n$ and $g_n$ are deterministic functions}.% \cup \Delta$. Here $\Delta$ is an error symbol and it is introduced for convenience; we explain its meaning subsequently.
%\todo{Do you not require $ \Delta $?}\tododone{Currently I have written $g_n$ without it, we can see if this way it is better, no change in the results though}

Let $d_n : \Xscr^n \times \Xscr^n \rarr \Real$ be the mean Hamming distance between $\xvec, \wi{\xvec} \in \Xscr^n $ given as
\begin{align}
  d_n(\wi{\xvec},\xvec) = \frac{1}{n} \sum_{k = 1}^{n} \big|\{ k : \wi{x}_k \neq x_k\}\big|.  \label{eq:dist-func}
\end{align}
For $\bd \in [0,1]$, define the set of correctly recovered sequences when the receiver plays $g_n $ and the sender plays $s_n$ as
\begin{align}
	\Dscr_{\bd}(g_n,s_n) =   \Big\{\xvec \in \Xscr^n : d_n(g_n \circ s_n(\xvec),\xvec) \leq \bd \Big\} \label{eq:seq-rec-within-d}
	%	\Dscr_{\bd}(\varphi_n,f_n)  = \Big\{\xvec \in \Xscr^n : d_n(\varphi_n \circ f_n(\xvec),\xvec) \leq \bd \Big\} \label{eq:seq-rec-within-d}
\end{align}
The probability of error for strategies $ g_n $ and $s_n $ is
\begin{align}
\Escr_{\bd}(g_n,s_n) = \Pbb \big( d_n(g_n \circ s_n (\Xvec),\Xvec) > \bd \big) = 1-\Pbb(	\Dscr_{\bd}(g_n,s_n) ),
\end{align}
that is the probability of exceeding the distortion level $\bd$.% when the receiver plays $g_n$ and the sender plays $s_n$.

The sender has an $ n $-block utility function given by  $\ut_n : \Xscr^n  \times \Xscr^n \rarr \Real$, where   %$\ut_n : (\Xscr^n \cup \Delta) \times \Xscr^n \rarr \Real$, is given as
\begin{align}
  \ut_n(\wi{\xvec},\xvec) = \frac{1}{n} \sum_{i=1}^{n} \ut(\wi{x}_i,x_i) \qquad \forall \;\xvec, \wi{\xvec} \in \Xscr^n, \label{eq:avg-util-defn}
\end{align}
and $\ut : \Xscr \times \Xscr \rarr \Real$ is the sender's single-letter utility function. %, and  $\ut_n(\Delta,x) = -\infty$ for all $x \in \Xscr^n$ and $n \in \Nbb$.
In our formulation, $\xvec$ is the source sequence observed by the sender and $\wi{\xvec}$ is the sequence recovered by the receiver. Thus, the utility of the sender is a function of the true sequence and the sequence recovered by the receiver. We also assume that $ \ut(i,i)=0 $ for all $ i\in \Xscr $; this is without loss of generality as we explain after Definition~\ref{defn:best_resp}.% below.

{\color{black} The receiver chooses strategies $ g_n $ to achieve vanishing probability of error. %  $\Escr_{\bd}(g_n,s_n)$ by choosing an .
The sender, on the other hand, maximizes the utility $ \ut_n(g_n \circ s_n(\xvec),\xvec) $ for all $ \xvec \in \Xscr^n $ by choosing an appropriate strategy $ s_n $.  We formulate this problem as a  leader-follower game, also called a \textit{Stackelberg} game, with the receiver as the leader and the sender as the follower.
The receiver, being the leader, announces its strategy before the sender and for a given strategy of the receiver, the sender chooses a response that maximizes its utility. The receiver anticipates this response of the sender and accordingly chooses an optimal strategy that minimizes its objective.

The set of best responses for the sender are defined as follows.
}

\begin{definition} [Best response strategy set]\label{defn:best_resp}
 For a strategy $g_n$ of the receiver, the set of best responses of the sender, denoted by $\best(g_n)$, is defined as
\begin{align}
  \best(g_n) = {\Big\{} &s_n : \Xscr^n \rarr \Xscr^n \;|\; \ut_n(g_n \circ s_n (\xvec),\xvec) \geq \ut_n(g_n \circ s_n' (\xvec),\xvec) \quad \forall \; \xvec \in \Xscr^n, \forall \; s_n' {\Big\}}. \label{eq:sen-opt-stra-game-noiseless}
\end{align}
% for 2 coln
% \begin{align}
%   \best(g_n) = {\Big\{} &s_n : \Xscr^n \rarr \Xscr^n \;|\; \non \\
%   &\ut_n(g_n \circ s_n (x),x) \geq \ut_n(g_n \circ s_n' (x),x) \non \\
%   &\hspace{3cm} \forall \; x \in \Xscr^n, \forall \; s_n' {\Big\}}. \label{eq:sen-opt-stra-game-noiseless}
% \end{align}
\end{definition}
%It is easy to see that the set of best responses $ \best(g_n) $ is the same if $ \ut_n(\xvec,\xvec) $ is subtracted on both sides of the inequality in \eqref{eq:sen-opt-stra-game-noiseless}.
% In this case, it would imply that the sender chooses its best response by computing how far is $ \ut_n(g_n \circ s_n (x),x) $ from $ \ut_n(x,x) $.
{\color{black} %For any $ \xvec $, the sender  maximizes its utility by choosing a sequence $ \yvec \in \image(g_n)$.
	Without loss of generality,  we assume  $ \ut(i,i) = 0 $ for all $ i \in \Xscr $. This is because, for any utility $\ut'$ we can define another utility $ \ut(j,i) = \ut'(j,i) - \ut'(i,i)$ which gives %for all $ i,j \in \Xscr $, so that
    \begin{align}
\argmax_{s_n}	\; \ut_n(g_n \circ s_n(\xvec),\xvec) &= \argmax_{s_n} \sum_{i} \ut(g_n \circ s_n(\xvec)_i,x_i)/n \non \\
	&= \argmax_{s_n} \sum_{i} (\ut'((g_n \circ s_n(\xvec))_i,x_i)-\ut'(x_i,x_i))/n \non \\
	&= \argmax_{s_n} \; \ut_n'(g_n \circ s_n(\xvec),\xvec). \non
\end{align}

 %\todo{not sure why this has been added... plus it is incorrect.}

We will look for strategies $ \{g_n\}_{n \geq 1 } $ such that for all $ \epsilon, \delta > 0 $, there exists an $ N \in \Nbb $  and
\begin{align}
\max_{s_n \in  \best(g_n) }	\Escr_{\bd + \delta}(g_n,s_n) < \epsilon \quad \forall \; n \geq N. \label{eq:receiver_objective}
\end{align}
In \eqref{eq:receiver_objective}, we assume that the sender chooses the worst-case best response from the set of best responses $ \best(g_n) $. This \textit{pessimistic} formulation is commonly adopted for the concept of Stackelberg equilibrium in game theory \cite{basar99dynamic} and the pessimistic viewpoint is also referred to as a \textit{weak Stackelberg equilibrium} \cite{breton1988sequential}. The set of worst-case best response strategies is %of the sender is defined as
\begin{align}
	\pbest_{\bd}(g_n) = \argmax_{s_n \in \best(g_n)}\; \Escr_{\bd}(g_n,s_n). \label{eq:worst_case_br}
\end{align}

%Thus, the objective of the receiver is to construct a sequence of strategies for which the worst-case probability of error over the best response strategies of the sender is arbitrarily small. % vanishes asymptotically. %Thus, an achievable sequence of strategies help the receiver to \textit{achieve} vanishing probability of error.

In \eqref{eq:receiver_objective}, we follow the notion of recovery used by Lapidoth in \cite{lapidoth1997role}. In the information theory literature (see \cite{csiszar2011information, cover2012elements}) the convention is to choose an encoder-decoder as $ \{f_n,\varphi_n\}_{n \geq 1}$, with $ f_n : \Xscr^n \rarr \Yscr^n, \varphi_n : \Yscr^n \rarr \Xscr^n  $, $ \Yscr $ being an appropriate alphabet space, such that for all $ \epsilon > 0  $, we have $  \Escr_{\bd}(\varphi_n,f_n) < \epsilon $ for large enough $n$. This is equivalent to taking $ \delta= 0 $ in \eqref{eq:receiver_objective} %instead of a limiting sequence $ \delta_n \rarr 0  $ in the Definition~\ref{defn:achiev-rate}
and it yields the same results pertaining to the rate region when $\delta > 0 $. However, the two notions do not yield the same results in our setting and for lossless recovery, there may not exist a sequence of strategies satisfying \eqref{eq:receiver_objective} when $ \delta = 0 $.
\begin{theorem}\label{thm:lossl_delta=0_empty_R}
	Let $ \ut $ be such that $ \ut(i,j) \geq 0 $ for some $ i,j \in \Xscr, i \neq j $. Then, for all sequences of strategies $ \{g_n\}_{n \geq 1} $ we have
	\begin{align}
		\lim_{n \rarr \infty} \max_{s_n \in \best(g_n)} \Pbb( g_n \circ s_n(\Xvec) \neq \Xvec) > 0. \non
	\end{align}
\end{theorem}
\begin{proof}
	Proof is in     Appendix~\ref{append:lossl_delta=0_empty_R}.
\end{proof}
The above theorem demonstrates the ``discontinuity'' between the two notions of achievability, one defined with $ \delta =0 $ and other with  $ \delta \rarr 0 $. In the strategic setting it is imperative to take a $ \delta \rarr 0 $ in \eqref{eq:receiver_objective}.
\subsection{Rate of communication}
\label{sec:rate_of_comm}

%Our central object of inquiry in this paper is the  \textit{set of achievable strategies and achievable rates},  defined as follows.
%\begin{definition}[Achievable sequence of strategies for lossy recovery] \label{defn:achiev-str-lossless}
%  A sequence of strategies $\{g_n\}_{n \geq 1}$ is achievable if for all $ \epsilon, \delta > 0 $ there exists a large enough $ N \in \Nbb $ such that %there exists a sequence of strategies $ \{s_n\}_{n \geq 1} $ with $ s_n \in \best(g_n) $,  such that
%  \begin{align}
	%    \max_{s_n \in \best(g_n)} \Escr_{\bd + \delta}(g_n,s_n) =  \max_{s_n \in \best(g_n)}\Pbb \big( d_n(g_n \circ s_n (\Xvec),\Xvec) > \bd + \delta \big)< \epsilon   \quad \forall \; n \geq N. \label{eq:asymp_vanish_error}
	%  \end{align}
%\end{definition}
%\begin{itemize}
%
%\end{itemize}

%We now motivate the notion of the rate of communication for our setting. We then define the achievable rate and show how it relates to the information-theoretic notion of the rate.
%\textbf{motivation and example}

%As in information theory, our interest is in quantifying the amount of channel resources that are required by the receiver to communicate with the sender.

We now define the rate of communication.
%We motivate our definition of rate with the following interpretation of the information-theoretic rate. Consider the canonical setting of source coding, where an encoder and a decoder jointly construct a code to represent a source from $ \Xscr^n $ over $ \Cscr  $ codewords. The encoder is defined as $ f_n: \Xscr^n \to \Cscr $ and the decoder is defined as $ \varphi_n :  \Cscr \to \Xscr^n$.
In our setting, we do not define the strategies $ s_n $ and $ g_n $ of the sender and receiver with reference to a codebook since we do not  restrict the image of the sender. Thereby, we adopt an indirect approach to measuring the communication resources employed by the strategies $s_n,g_n$.

%Observe that from the perspective of the receiver, the amount of communication resources used is equivalently given by the number of unique reconstructions used to represent the output of the channel.
%We later show that without loss of generality, we can restrict the strategy of the receiver as
%\begin{align}
%g_n(\xvec) = 	\begin{cases}
%		\xvec & \mbox{ if } \xvec \in I^n \\
%		\xvec_0 & \mbox{ else }
%	\end{cases},
%\end{align}
%where $ I^n \subseteq \Xscr^n $, and $ \xvec_0 \in I^n $ is arbitrary (cf. Lemma~\ref{lem:simplified_gn}).
When the sender and receiver employ strategies $s_n,g_n$, the number of reconstructed sequences is given by $|\image(g_n \circ s_n)|$.
The number of reconstructions corresponding to \textit{correctly} recovered sequences is given as
%The set \textit{all} resources used is given by  $\image(\varphi_n \circ f_n)$ and
\begin{align}
    \Ascr^n_{\bd}(g_n,s_n)	= \Big\{ \wi{\xvec} \in \image(g_n) : \exists \xvec \in \Dscr_{\bd}(g_n,s_n)  \mbox{ such that } g_n \circ s_n(\xvec) = \wi{\xvec}   \;\Big\}. \label{eq:defn-Ascr_d}
    %\Ascr^n(\varphi_n,f_n)	= \Big\{ \wi{\xvec} \in \image(\varphi_n) : \varphi_n \circ f_n(\xvec) = \wi{\xvec}   \mbox{ for some }\;\xvec \in \Dscr_{\bd}(\varphi_n,f_n) \Big\}. \label{eq:image_dec_enc_comp}
\end{align}
The set $\Ascr_{\bd}^n(g_n,s_n)$ is hence the image of the map $ g_n \circ s_n : \Dscr_{\bd}(g_n, s_n ) \rarr \Xscr^n $. We define the rate using the set  $\Ascr^n_{\bd}(g_n,s_n)$.
 \begin{definition}[Rate of communication] \label{defn:rate-of-comm}
    The rate of communication when the receiver plays the strategy $g_n$ and the sender plays the strategy $s_n$ is defined as %$R_{\bd}(g_n,s_n)$ is defined as
    \begin{align}
        R_{\bd}(g_n,s_n) &=  \frac{1}{n}\log \big| \Ascr_{\bd}^n(g_n,s_n) \big|. \label{eq:defn-rate-Rdn}
    \end{align}
    %\hltodo{For $\bd = 0$, the rate is denoted as $R(g_n,s_n)$}{do you need this? It clashes with Def II.5}.
\end{definition}
Thus, $R_{\bd}(g_n,s_n)$ computes the rate of growth of the number of reconstructions used for correct recovery. An achievable rate is defined as follows.
\begin{definition}[Achievable rate] \label{defn:achiev-rate}
    A rate $ R $ is achievable if there exists a  sequence of strategies $ \{g_n\}_{n \geq 1 }$, $s_n \in \pbest_{\bd+\delta_n}(g_n)$ and $ \epsilon_n, \delta_n \rarr 0  $ such that $  \Escr_{\bd + \delta_n}(g_n,s_n) < \epsilon_n   $ and
    \begin{align}
        \lim_{n \rarr \infty}  R_{\bd + \delta_n}(g_n,s_n) = R. \non
    \end{align}
    In this case we say the sequence of strategies $ \{g_n\}_{n \geq 1} $ achieves the rate $ R $.  For $ \bd > 0 $ the rate region is the set of achievable rates for lossy recovery and is denoted by $\Rscr_{\bd}$. For the lossless case, the achievable rate is defined by taking $ \bd = 0 $ in the above expressions and the rate region is denoted as $ \Rscr $.

    The infimum and supremum of achievable rates are defined as $ \Rscr_{\bd}^{\inf} = \inf \{ R : R \in  \Rscr_{\bd}\} $ and $ \Rscr_{\bd}^{\sup} = \sup \{R : R \in  \Rscr_{\bd} \} $ respectively.
\end{definition}

In the cooperative setting when the encoder and decoder are chosen jointly to minimize probability of error, this definition reduces to the traditional way of defining the rate (\ie, the size of the codebook). We term the rate region for the cooperative case as the Shannon rate region and it is given by $ [\entrX,\log q]$ for $\bd = 0 $ and $ [R(\bd),\log q]$ for $d > 0 $.

The key part of the definition is the restriction to \textit{correctly} decoded sequences; we now explain why this restriction is necessary.
%In the cooperative case and one-to-one decoder, this is also the size of the codebook used by the encoder and decoder. In the our setting too, if $g_n$ is one-to-one, $|\image(g_n \circ s_n)|$ is  the number of channel inputs used during the interaction of the sender and the receiver.
 %Here $R(\bd)$ is the rate-distortion function given as
%\begin{align}
%	R(\bd) = \inf_{\Ebb[d(X,Y)] \leq \bd} I(X;Y). \label{eq:rate_dist_func}
%\end{align}
Consider the quantity  	\begin{align}
    \Rhat_{\bd}(g_n,s_n) = \frac{1}{n} \log |\image(g_n \circ s_n)|. \non
\end{align}
Define the set $\wi{\Rscr}_{\bd}$ as the set of $\Rhat$ for which there exists a  sequence of strategies $ \{g_n\}_{n \geq 1 }$, $s_n \in \pbest_{\bd+\delta_n}(g_n)$ and $ \epsilon_n, \delta_n \rarr 0  $ such that $  \Escr_{\bd + \delta_n}(g_n,s_n) < \epsilon_n   $ and
\begin{align}
	\lim_{n \rarr \infty} \Rhat_{\bd+\delta_n}(g_n,s_n)  = 	\Rhat. \label{eq:rate_defn_with_image}
\end{align}
In general, $\Ascr_{\bd}^n(g_n,s_n) \subseteq  \image(g_n \circ s_n) $. But in the cooperative setting, for an optimal pair of encoder and decoder that minimizes the error, it must be that $	  \Ascr_{\bd}^n(g_n,s_n) =  \image(g_n \circ s_n) $, which makes $\wi{\Rscr}_\bd=\Rscr_{ \bd}.$
%Thus, the two notions of rate are equivalent in the cooperative setting
%% This is because, trivially, we can assume that $ \varphi_n \circ f_n(\xvec) = \xvec  $ for all $ \xvec \in \Cscr $. %\image(\varphi_n \circ f_n) $.
% and $\image(\varphi_n \circ f_n) $ is also the correct measure of required resources.% used for recovery of sequences.% and the rate is defined as $\log |\image(\varphi_n \circ f_n)|/n$.
%Since our channel is noiseless, source coding is the analogous problem for our setting. In source coding, the amount of this resource is determined by the number of sequences available for the encoder to represent its information. In our setting, we find that this manner of counting may be an overestimation of the required resources.  Thereby, we adopt a somewhat non-standard manner of quantifying the rate, but which also generalizes the information-theoretic notion of rate. % discounts any \textit{unused} sequences and also generalizes the information-theoretic notion of rate.
%We present this below.
In the strategic setting, measuring the required communication resources using $\image(g_n \circ s_n) $ results in an overestimation.
In Example~\ref{eg:rate_gap_eg} in Section~\ref{sec:illustrative_examples}, we show that there exists a sequence of strategies $ \{g_n\}_{n \geq 1 }$, $ s_n \in \pbest_{\bd+\delta}(g_n)$ achieving vanishing error  such that %for which we have % and for the corresponding sequence $ s_n \in \pbest_{\bd+\delta}(g_n)$, we have
\begin{align}
	\lim_{n \rarr \infty} R_{\bd+\delta}(g_n,s_n) \leq \Rscr_{\bd}^{\sup} < 	\lim_{n \rarr \infty} \Rhat_{\bd+\delta}(g_n,s_n). \non
\end{align}
 Thus, in our setting $\wi{\Rscr}_\bd \supseteq \Rscr_{ \bd}$, where the inclusion is in general strict.

Observe, that $\inf \wi{\Rscr}_\bd = \Rscr_\bd^{\inf}$ whereby the two notions agree on the infimum achievable rate. This is because, for any $g_n $, we  can define a new strategy $ \gbar_n $ by truncating the image of $ g_n $ such that $ \image(\gbar_n \circ \sbar_n) = \Ascr_{\bd}^n(\gbar_n,\sbar_n) $ for all $\sbar_n \in \pbest_{\bd+\delta}(\gbar_n) $.
However, the distinction arises for the supremum rate and as shown above  % in Example~\ref{eg:rate_gap_eg} the rate according to Definition~\ref{defn:infor_theoretic_rate} may be strictly greater and hence
we may have $ \sup \{ R : R \in  \wi{\Rscr}_{\bd}\} > \sup \{ R : R \in  \Rscr_{\bd}\} $.

This is crucial because the sequences in $ \image(g_n \circ s_n)  \setminus \Ascr_{\bd}^n(g_n,s_n)  $ are ``unused sequences'' which are only mapped to by  sequences \textit{outside} $ \Dscr_{\bd}(g_n,s_n)$.
 These sequences do not aid in the correct recovery of any sequence and  must not be counted  as required communication resources.

 Finally, we note that in the cooperative setting, it is possible to achieve a smaller error probability by increasing the rate of communication. However, this may not hold in our setting. Our setting exhibits a gap between $\Rscr_{ \bd}^{\sup}$ and the  channel capacity whereby the receiver cannot construct a strategy with a rate beyond a certain limit while also achieving vanishing probability of error (cf. Theorem~\ref{thm:lossl_rate_charac} part b) and Example~\ref{eg:rate_gap_eg}).

 We conclude the section with the following result that shows that the rate region is  convex.
 \begin{theorem}[Convexity of rate region]\label{thm:ach-rate-convex}
 	For any $\bd \in [0,1]$, the achievable rate region $\Rscr_{\bd}$ is convex. Moreover, $ \Rscr_{\bd}^{\inf} $ is a convex function of $ \bd  $ and $ \Rscr_{\bd}^{\sup} $ is a concave function of $ \bd $.
 \end{theorem}
 \begin{proof}
 	Proof is in  Appendix~\ref{appen:ach-rate-convex}.
 \end{proof}
 Thus, for a complete characterization of the rate region, it suffices to determine $ \Rscr_{\bd}^{\inf} $ and $ \Rscr_{\bd}^{\sup} $. %\hltodol{
 	%Additionally, the above theorem also implies that $  \Rscr_{\bd}^{\inf} $ and $ \Rscr_{\bd}^{\sup} $ are convex and concave functions of $ \bd $, respectively.%}{Do you mean $ \Rscr_{ \bd}^{\min} $ as a function of $ \bd $?}%\todo{This section feels too short. Can we add more results here? Or may include figures/example. Definitely needs some discussion on convexity and strong converse -- why its not obvious and why it is different from Shannon theory.}\tododone{Have added explanations and some results from the appendix}
 %The following lemma shows that $ \Rscr_{\bd}^{\inf} $ is a convex curve.
 %  \textbf{check again}
 %  \begin{corollary} \label{lem:Rdmax_Rdmin}
 	%  $ \Rscr_{\bd}^{\inf} $ is a convex function of $ \bd  $ and $ \Rscr_{\bd}^{\sup} $ is a concave function of $ \bd $.
 	%\end{corollary}
 	%\begin{proof}
 	%  Proof is in  Appendix~\ref{appen:Rdmax_Rdmin}.
 	%\end{proof}
It may not be entirely evident from the formulation that the rate region is a convex set. This is because the rate is also determined by the choice of strategy of the sender which is chosen according to its utility.  Nevertheless, convexity holds and this is shown by appropriately choosing the strategy of the receiver and using the standard time-sharing argument from information theory.}

\section{Main Results : Binary Alphabet}
\label{sec:main}

We now present our results for the binary alphabet.
%{\color{blue} We then discuss the differences between our setting and the cooperative setting of communication. We also discuss the similarities between our formulation and the mechanism design problem from game theory}.
%
First, we discuss the lossless case where we determine $ \Rscr^{\inf} $ and also derive a bound on $\Rscr^{\sup}$. We write $ \prbX(0) = p $ for the binary case.

\subsection{Lossless recovery}

% Define
% \begin{align}
%   K(n,p,q,\ut) &= \sum_{i=0}^{\kbar}\left(\binom{n(1-p)}{\khat+i}/\binom{np}{i}\right), \non \\
%   M(n,p,q,\ut) &= \sum_{i=0}^{\kbar}\left(\binom{nq)}{\khat+i}/\binom{n(1-q)}{i}\right), \non
% \end{align}
% where $ \khat = n(q-p) $ and $ \kbar = n(q-p)a/(b-a) $. Finally define
% \begin{align}
%   \Delta(p,q,\ut) = \lim_{n\to \infty} \max_{q \in \Pscr(\Xscr)} \frac{1}/{n}\log   \frac{K(n,p,q,\ut) }{M(n,p,q,\ut) }. \non
% \end{align}
% We state the rate region in terms of the above expressions.
\begin{theorem}[Lossless rate region]\label{thm:lossl_rate_charac}
  Let $ X \sim {\rm Bern}(p) $  and $\bd = 0 $. %If , then and hence  $ \Rscr = \emptyset $.  If %
  Then, $ \Rscr \neq \emptyset $ if and only if
  \begin{align}
    \ut(0,1) + \ut(1,0)  < 0. \label{eq:u01_u10_cond}
  \end{align}
  %then $ \Rscr \neq \emptyset $.
  Further, % \textbf{strong converse?}
  \begin{itemize}
    \item [a)] if $ \ut(0,1) + \ut(1,0)  \geq 0 $, then for small enough $ \delta $, $ \max_{s_n \in \best(g_n)}\Escr_{\delta}(g_n,s_n) \to 1 $ for all $ \{g_n\}_{n \geq 1} $.
  \item[b)] if $  \ut(0,1) + \ut(1,0)  < 0 $ and either $ \ut(0,1) \geq 0 $ or $ \ut(1,0) \geq  0 $, then
    \begin{align}
\Rscr^{\inf} = \entrP, \quad \Rscr^{\sup} \leq  H\left( \min\left\{\frac{b}{a} p , \frac{1}{2}\right\}\right)  \label{eq:rate_reg_loss_bin} %\subseteq \left[ \entrP,\;  H\left( \min\left\{\frac{b}{a} p , \frac{1}{2}\right\}\right) \right], \label{eq:rate_reg_loss_bin}
    \end{align}
    where $ a = \min\big\{|\ut(0,1)|,\;|\ut(1,0)| \big\}, \;b = \max \big\{|\ut(0,1)|,\;|\ut(1,0)| \big\} $. Moreover, if $ \ut(1,0) = -b $ and $ \ut(0,1) = a $, then for a $ g_n $ such that $ \lim_{n \to \infty} \log |\image(g_n)|/n > H\left( \min\left\{\frac{b}{a} p , \frac{1}{2}\right\}\right) $, and for small enough $ \delta$, we have
    \begin{align}
      \lim_{n \to \infty} \max_{s_n \in \best(g_n)}\Escr_\delta(g_n,s_n) = 1. \non
    \end{align}
  \item[c)] if $ \ut(0,1), \ut(1,0) < 0 $, then  $$ \Rscr = [\entrP,1]. $$
%  \item If $ \ut(0,1) + \ut(1,0) \geq 0 $, then $ \Rscr = \emptyset $
  \end{itemize}
\end{theorem}
\begin{proof}
  We briefly discuss the idea of the proof. The complete arguments are in Appendix~\ref{sec:bin_alpha_lossl}.
%  Specifically, the proofs of necessity of \eqref{eq:u01_u10_cond} and part a) are given in \ref{appen:proof_necc_cond_lossl}. The proofs of sufficiency of \eqref{eq:u01_u10_cond} and part b) and c) are given in \ref{appen:suff_cond}. %\ref{appen:proof_necc_cond} and Appendix~\ref{appen:suff_cond}.

\subsubsection{Sketch of proof of part a) }
  Let  $ \delta > 0 $ and suppose \eqref{eq:u01_u10_cond} does not hold. Then, we show that from any type class, only a  $2\delta$-ball of sequences can be recovered within distortion $ \delta $. This is because for any correctly recovered sequence $ \xvec \in \Xscr^n $, there exists a sequence $ \yvec $ in the image of receiver's strategy such that $ \yvec \in B_\delta(\xvec) $. Consequently, due to the structure of the utility,  for all sequences $ \widehat{\xvec} $ sufficiently far from $ \xvec $, \ie, $ d_n(\xvec,\widehat{\xvec}) > 2 \delta$, the sender prefers to map $ \widehat{\xvec} $ to the sequence $ \yvec $ over any other sequence in  $ B_\delta(\widehat{\xvec}) $.   Thus, for small enough $ \delta $  the probability of error tends to one and there does not exist any sequence of strategies that achieve arbitrarily small probability of error and hence the rate region is empty.

   \subsubsection{Sketch of proof of part b) }
   For the part b), we show that there exists a sequence of  strategies  %\ach
     achieving the  rate $ \entrP $ using a result due to Lapidoth in \cite{lapidoth1997role}. For $ bp/a \geq 1/2 $, we get the trivial upper bound on the achievable rate. When $ bp/a < 1/2 $, we show that the rate region is bounded by $ H(bp/a) $. We prove this by considering two cases. For the case when $ \ut(1,0) \geq 0 $, we show that the sequences having type in the set $ [bp/a,1] $ do not contribute to the rate of any sequence of strategies. %\ach
    When $ \ut(0,1) \geq 0 $, we show that if the image of any strategy of the receiver includes sequences having type in the set $ [bp/a,1]$, then only a fraction of the high probability sequences are recovered correctly. Thus, the  probability of error cannot be arbitrarily small. In either case, the achievable rate is bounded by $ H(bp/a) $.  %Since we do not explicitly construct strategies that achieve rates higher than $ H(p) $, we only get the rate region in terms of an inclusion.% included in $ $ We only get %we use  Theorem~\ref{thm:dist_rate_bound} to show that $ H(p) $ is achievable. To show $ \Rscr^{\sup} = \entrP $
\subsubsection{Sketch of proof of part c) }
  In the part c) where both utility terms are negative, the sender is truthful about its information which corresponds to the cooperative communication case and hence the rate region is same as the Shannon rate region.
\end{proof}

The condition \eqref{eq:u01_u10_cond} is a non-asymptotic single-letter condition on the utility and indicates a sharp threshold for the existence of a non-empty rate region.
% and does not depend on the blocklength $ n $.represents a sharp boundary between the empty rate region and non-empty rate region.
Furthermore, the non-empty rate region also demonstrates a dual nature, either it is a strict subset of the Shannon rate region or it is the complete rate region. %single point or it is equal to the Shannon rate region.  %Also, the condition \eqref{eq:u01_u10_cond} is a non-asymptotic single-letter condition on the utility and does not depend on the blocklength $ n $.
If $ \ut(1,0) \geq 0 $, %\ie, the sender prefers to report the symbol $ 0 $ as $ 1 $,
then under  \eqref{eq:u01_u10_cond}, $ \ut(0,1) $ will be negative and greater than $ \ut(1,0) $ in magnitude. Intuitively, this implies that the gain derived by the sender by misrepresenting $ 0 $ by $ 1 $ is lesser than the \textit{penalty} for misreporting $ 1 $ by $ 0 $. This paves the way for a non-zero achievable rate.
%Suppose $ bp/a < 1 $. Then, in part b), we observe that all achievable rates lie below $ H(bp/a) $. This is because including sequences having type in $[bp/a,1]$ leads to the following -- either no sequence contributes to the rate as computed by \eqref{eq:defn-rate-Rdn} or it leads to the receiver not recovering sufficiently many high probability sequences. In the former, such sequences are redundant and need not be included in the image. Whereas in the latter, the receiver necessarily has to leave out such sequences else the probability of error will not be arbitrarily small. Thus only rates below $ H(bp/a) $ are achievable.%The latter happens because every sequence from type classes outside this typical set, has an edge with some sequence from the typical set. Thus, communicating at a rate higher than the entropy includes too many sequences from other type classes which gives the sender more freedom to lie about its information. This in turn leads to
When $ \ut(0,1), \ut(1,0) < 0 $, the sender is truthful about its information since its interests are aligned with the objective of the receiver.
%; it corresponds to the notion of \textit{incentive compatibility} as will be discussed later in Section~\ref{sec:inc_comp}.%from the theory of mechanism design \cite{myerson1997game}.

%\todo{There are some interesting intermediate results in the appendix. I am not sure if they should be left there. Think how many of them can brought forward.}

%Finally, observe that the part a) of the above theorem states that $  \Rscr^{\inf} = \Rscr^{\sup} = \entrP $. Thus, to achieve asymptotically vanishing probability of error, the receiver should communicate at a rate exactly equal to the entropy of the source. \hltodo{This is akin to the notion of the capacity of a channel which states for reliable communication over the channel, the rate should necessarily be less than the capacity.}{Not sure I agree with this analogy}\tododone{ what if we say in context of Lemma~\ref{lem:strong_conv_type} from Sec 3}

  \begin{figure}
    \begin{center}
      \begin{tikzpicture}[scale=1]
      % Draw axes

      \matrix (S) [matrix of math nodes, nodes={anchor=center,outer sep=0pt, minimum width=2cm, minimum height=2cm},nodes in empty cells]
      {
        & \\
        & \\
        % | [fill=green!50!black!50]| & \\
      };

      % \begin{scope}[on background layer]
      %   \path [pattern=vertical lines, pattern color=black] (S-1-1.north west) |- (S-1-1.south east) -- cycle;
      %   \path [pattern=vertical lines, pattern color=black] (S-2-2.north west) |- (S-2-2.south east) -- cycle;

      %   \path [pattern= north west lines, pattern color=red] (S-1-2.north west) |- (S-1-2.south east) -- cycle;
      %   \path [pattern= north west lines, pattern color=red] (S-1-2.north west) -| (S-1-2.south east) -- cycle;

      %   \path [pattern= north west lines, pattern color=red] (S-1-1.north west) -| (S-1-1.south east) -- cycle;
      %   \path [pattern= north west lines, pattern color=red] (S-2-2.north west) -| (S-2-2.south east) -- cycle;
      % \end{scope}
      \begin{scope}[on background layer]

        \path [pattern=dots, pattern color=black] (S-1-1.north west)  |- (S-1-1.south east) -- cycle;
        \path [pattern=dots, pattern color=black] (S-2-2.north west) |- (S-2-2.south east) -- cycle;
        \path [pattern= north west lines, pattern color=black] (S-2-1.north west) |- (S-2-1.south east) -- cycle;
        \path [pattern= north west lines, pattern color=black] (S-2-1.north west) -| (S-2-1.south east) -- cycle;

      \end{scope}
      \draw [<->,thick] (0,2) node (yaxisp) [right] {$ \ut(1,0) $}
      |- (2,0) node (xaxisp) [above] {$ \ut(0,1) $};
      \draw [<->,thick] (0,-2)  node (yaxisn) [below] { }
      |-  (-2,0) node (xaxisn) [left] {};
      % Draw two intersecting lines
      \draw [->] (0,0) coordinate (a_1) -- (2,-2) coordinate (a_2);
      \draw [->] (0,0) coordinate (a_3) -- (-2,2) coordinate (a_4);

      \path [fill=white]
      ($(S-2-1.center) + (-8mm,-3mm) $) rectangle ($(S-2-1.center) + (8mm,3mm)$);
        \draw  ($(S-2-1.center) + (0mm,0mm)$)  node {$ [H(p),1] $};

      \path [fill=white]% (2.5,2.5) rectangle (3,3);
      ($(S-1-1.center) + (-6mm,-5.5mm)$) rectangle ($(S-1-1.center) + (-1mm,-1mm)$);

      \draw  ($(S-1-1.center) + (-3.5mm,-2.5mm)$) node {$ \tilde{\Rscr} $};

      \path [fill=white]% (2.5,2.5) rectangle (3,3);
      ($(S-2-2.center) + (-6mm,-5.5mm)$) rectangle ($(S-2-2.center) + (-1mm,-1mm)$);
      \draw  ($(S-2-2.center) + (-3.5mm,-3mm)$) node {$  \tilde{\Rscr} $};

         \draw  ($(S-1-2.center) + (-2mm,-2mm)$) node {$ \emptyset $};
    \end{tikzpicture}
    \caption{Rate region (lossless recovery, binary alphabet) : a) the rate region is empty whenever $ \ut(0,1)+ \ut(1,0) \geq 0$, b) for $ \ut(0,1)+ \ut(1,0) < 0 $, $ \ut(0,1)\ut(1,0)\leq 0 $, the rate region is  $  \tilde{\Rscr} \subseteq  [H(p), R] , R \leq H(\min\{bp/a,1/2\})$, c) the rate region is same as Shannon rate region for $  \ut(0,1),\ut(1,0) <  0$}
    \end{center}
  \end{figure}

\subsection{Lossy recovery}

In this section, we present a characterization of the rate region for a positive distortion level $ \bd $. The following theorem shows that when $ \bd \in (0,\min\{p,1/2\}) $, the condition \eqref{eq:u01_u10_cond} is necessary and sufficient for the lossy rate region to be non-empty.
\begin{theorem}[Lossy rate region]\label{thm:lossy_rate_charac}
 Let $ X \sim {\rm Bern}(p) $ and  $ \bd \in (0,p) $ where $ p \leq 1/2 $. Then $ \Rscr_{\bd} \neq \emptyset $ if and only if
   \begin{align}
\ut(0,1) + \ut(1,0)  < 0. \label{eq:u01_u10_cond_lossy}
\end{align}
 \begin{itemize}
  \item[a)]  If \eqref{eq:u01_u10_cond_lossy} holds, then %\todo{Unlike the earlier result, here Rsup has a lower bound. Make this clear in the discussion. Same for results that follow.}
%  {\color{red} check}
 \begin{align}
   % \ut(0,1) + \ut(1,0)  < 0, \label{eq:u01_u10_cond_lossy}
   \Rscr_{\bd}^{\inf} &=  \rateD, \non \\
     \Rscr_{\bd}^{\sup} &\in \left\{\begin{array}{c l}
       [H\left( p+\bd\right),1] & \mbox{ if } p+\bd < \frac{1}{2} \\
     =	1 & \mbox{ else }%  p+\bd \geq 1/2
     \end{array}\right.. \non
 \end{align} % 	\geq H\left( p+\bd\right) & \mbox{ if } p+\bd < 1/2
 %and $ \Rscr_{\bd}^{\sup} = 1 $ if $ p+\bd \geq 1/2 $.
% then
%$$  \left[\rateD,\; H\left(\min\left\{ p+\bd,\frac{1}{2}\right\}\right)\right] \subseteq \Rscr_{\bd}.  $$
% where $ a = \min\{|\ut(0,1)|,|\ut(1,0)|\}, \;b = \max\{|\ut(0,1)|,|\ut(1,0)|\} $.% $ \Rscr_{\bd} \neq \emptyset $. %Further,
%If \eqref{eq:u01_u10_cond_lossy} holds, then% if  $ \ut(0,1) \geq 0 $ or $ \ut(1,0) \geq  0 $, then

  \item[b)] If $ \ut(0,1), \ut(1,0) < 0 $,   then
    $$ \Rscr_{\bd} = [\rateD,1]. $$

  Finally, if $ \bd \geq p $, then $ \Rscr_{\bd}^{\inf} = 0 $.
%  \item If $ \ut(0,1) + \ut(1,0) \geq 0 $, then $ \Rscr = \emptyset $
  \end{itemize}
\end{theorem}
%\todo{can you express case a) as $\Rscr^{\sup} \in [H(p+d),1]$?}

  \begin{proof}
  The proof of necessity of \eqref{eq:u01_u10_cond_lossy} is given in Appendix~\ref{appen:proof_necc_cond}. The proofs of sufficiency of \eqref{eq:u01_u10_cond_lossy} and part a) and b) are given in Appendix~\ref{appen:suff_parta_partb}. %Detailed proofs are in  Appendix~\ref{appen:bin_alph_lossy}.

  We present a sketch of the proof here.  If \eqref{eq:u01_u10_cond_lossy} does not hold, then as shown in Theorem~\ref{thm:lossl_rate_charac} we have that from any type class, only a $2(\bd+\delta)$-ball of sequences can be recovered within distortion $ \bd+\delta $, where $ \delta > 0 $. Thus, as long as $2(\bd+\delta) < 1 $, then the probability of error will not tend to zero. Thus, there does not exist any sequence of strategies that achieve arbitrarily small probability of error and hence the rate region is empty.

   For the Case b), we show that for any sequence $ \xvec $ and given a type class $ U_P^n, P \in \types $, the sender prefers a sequence from $ U_P^n $ that is at the least distance from $ \xvec $. Thus, the receiver can choose a strategy where the image  consists sequences only from a single, appropriately chosen type class.  By choosing sequences from the type class corresponding to the distribution closest to the rate-distortion achieving distribution, we get that $ \rateD $ is achievable. By choosing the image as the entire type class with type closest to $ p+\bd $ we get that $ H(p+\bd) $ is achievable.
  \end{proof}
The above theorem, unlike Theorem~\ref{thm:lossl_rate_charac}, only gives a \textit{lower bound}  on $ \Rscr_{\bd}^{\sup}$. %\hltodo{Observe that our definition of the lossless recovery is defined as a limiting case of the lossy recovery case. For the lossless case, a rate is achievable if the probability of excess distortion vanishes asymptotically for arbitrarily small distortion.}{not sure what is implied by this}%  \hltodo{$ \bd = \delta $}{confusing.. let's drop this}. %In contrast, the lossy case requires that for an achievable rate, the probability of excess distortion goes to zero for a fixed distortion. However, in both the cases, we see that $ \ut(0,1) + \ut(1,0) < 0 $ is necessary and sufficient for existence of an achievable rate.

\subsection{Distinction from Shannon theory}
In the view of the above results, we highlight a  few observations that set our setting apart from the setting of cooperative communication. As mentioned before, one key characteristic is that the rate of communication may also have an upper limit that is less than the channel capacity.
%Notice that the objective of the receiver in our setting is similar to the setting of source coding which is to design a \textit{codebook} that helps achieve vanishing probability of error. The canonical setting of communication deals only with the minimum rate of communication, \ie, the least number of distinct sequences in the codebook that achieve the said objective.
In the cooperative setting, better error performance can be achieved when more channel resources are utilised since one can design an encoding that is complementary to the decoding function. However, this freedom is lost in our setting since the receiver can only choose the decoder, considering that the encoding will be chosen by the sender according to its strategic intentions.

Having more number of sequences implies that the sender has a greater freedom to represent, and thereby misreport, its information and obtain a higher utility. The above results quantify this limit on the number of sequences used for representation. It is worthwhile to note that this limit arises despite the clean channel. The upper limit on the rate is entirely a characteristic of the sender  and can be strictly smaller than the capacity of the channel.

%The utility of the sender determines its tendency to misreport its information and this also crucially determines the limits on the rate of communication. Further, this also dictates how the receiver should strategize to recover information from the sender.
Consider, for instance, the condition \eqref{eq:u01_u10_cond} and suppose $\ut(1,0) > 0 $ and $\ut(0,1) <0$ such that the sum of the two terms is negative. In this case, the receiver recovers information from the sender by choosing a strategy (concretely, the set $ I^n $ in \eqref{eq:general_form_g_n} in the Appendix) such that misreporting forces the sender to trade-off the incentive derived by reporting $ 0'$s as $ 1'$s with the penalty of $1'$s recovered as $0'$s, and over a $n$-block the penalty dominates. The challenge for the receiver is to devise a strategy such that forces the sender to be truthful (see proof of Theorem~\ref{thm:lossl_rate_charac}) and also drives the error to zero.  However, if \eqref{eq:u01_u10_cond} does not hold then it is impossible for the receiver to recover any information from the sender.

{\color{black}
Finally, note that when the sender also minimizes the Hamming distortion, then for all $ x, y \in \Xscr $ we have  $ \ut(x,y) = -d(x,y)$  . In this case, we have that  $\Rscr_{ \bd}=\wi{\Rscr}_\bd$, where $\wi{\Rscr}_\bd$ is defined in \eqref{eq:rate_defn_with_image}. Thus, the rate region coincides with Shannon rate region. To see this, fix a strategy $ g_n $ for the receiver.	 Then, for any sequence $ \xvec \in \Xscr^n $, the utility maximizing sequence in the image of $ g_n $ also  minimizes the distortion function. Thus for all  $ s_n \in \best(g_n)$, it holds that $ d_n(g_n \circ s_n (\xvec),\xvec) \leq d_n(\yvec,\xvec)  $ for all $\yvec \in \image(g_n).$ % $ \ut(i,j) = -d(i,j) $ for all $ i,j \in \Xscr $.
	Specifically, for all $ \xvec \in \image(g_n)$, we have $ g_n \circ s_n (\xvec) = \xvec $ and hence $ \image(g_n) \subseteq \Dscr_{\bd}(g_n,s_n)$. Since $ g_n \circ s_n(\xvec) = \xvec  $ for all $ \xvec \in \image(g_n) \bigcap \Dscr_{\bd}(g_n,s_n) $, from  \eqref{eq:defn-Ascr_d} we get that $ \Ascr_{\bd}(g_n,s_n)  = \image(g_n) $ which implies that $\Rscr_{ \bd}=\wi{\Rscr}_\bd.$

 In fact, we show a stronger result where $ \ut(0,1) < 0 $ and $ \ut(1,0) < 0 $ also leads to the same rate region as the Shannon rate region (Theorem~\ref{thm:lossl_rate_charac} c) and Theorem~\ref{thm:lossy_rate_charac} b)).}%, the interests of the sender align with the objective of the receiver and the rate region coincides with the Shannon rate region.

\section{Main Results : General Alphabet}
\label{sec:main_general_alph}
We now present a characterization of an achievable rate for the case of the general alphabet.

Let $\Qscr = \{Q^{(0)},\hdots,Q^{(|\Qscr|)}\}$ be the set of all $|\Xscr| \times |\Xscr|$ permutation matrices and let $Q^{(0)} = \mathbf{I} $ be the identity matrix. Consider the optimization problem $ \Oscr(\ut) $ defined as
  \begin{align}
\maxproblemsmall{$\Oscr(\ut) :$}
{Q}
{\displaystyle  \sum_{i,j \in \Xscr} Q(i,j)\ut(i,j) }
{\begin{array}{r@{\ }c@{\ }l}
    \qquad Q &\in& \Qscr \setminus \{\mathbf{I}\}.
	\end{array}}         \non
\end{align}
 Observe that $ \Oscr(\ut) $ depends on the single letter utility $ \ut $. Defining $ \Gamma(\ut)  = \OPT(\Oscr(\ut)) $, we  give a characterization of achievable rates based on the optimal value $ \Gamma(\ut) $.

\subsection{Lossless recovery}

\begin{theorem}[Lossless rate region]\label{thm:lossl_rate_charac_gen}
Let $ X \sim \prbX $ and $ \bd = 0 $. If $ \Rscr \neq \emptyset $, then we have $\Gamma(\ut) \leq 0 $. Moreover, if $  \Gamma(\ut) < 0 $, then $ \Rscr \neq \emptyset $. Further,
    \begin{itemize}
      % \item[a)] if  $  \Gamma(\ut) < 0 $, then $$  \Rscr = [\entrX, \; \log \;\Xi(\ut,P_X)]. $$
      \item[a)] if  $  \Gamma(\ut) < 0 $, then $$ \Rscr^{\inf} = \entrX. $$%  = [\entrX, \; \log \;\Xi(\ut,P_X)]. $$
  \item[b)] if $ \ut(i,j) < 0  $ for all $  i,j \in \Xscr $, then $$ \Rscr = [\entrX,\log q]. $$
%  \item If $ \ut(0,1) + \ut(1,0) \geq 0 $, then $ \Rscr = \emptyset $
  \end{itemize}
\end{theorem}
\begin{proof}
  We present a short overview of the proofs. Detailed proofs are given in Appendix~\ref{append:gen_alph_lossl} with some preliminaries and the results are proved in  Appendix~\ref{appen:nece_of_gamma} and Appendix~\ref{appen:thm_loss_rate_charac_gen}.

  When $ \Gamma(\ut) > 0 $,  under some conditions on certain types $ P_1, P_2 \in  \types $, if the image of the receiver's strategy has more than  $ (1-\beta)$  fraction of sequences from the type class  $ U_{P_1}^n $, then  no more than $ \beta $ fraction of sequences from the type class $ U_{P_2}^n $ can be recovered within the prescribed distortion. For small enough distortion, we can show that all the type classes close to the distribution $ \prbX $ will satisfy this condition. We will use this fact to show that the probability of error is always bounded away from zero for any sequence of strategies of the receiver  and hence the rate region is empty.
  % Thus, the probability of error is always bounded away from zero.
  For the case a) and b), we use ideas similar to the case of binary alphabet to determine the achievable rates.
\end{proof}
The condition $  \Gamma(\ut) < 0 $ is in fact a generalization of the condition \eqref{eq:u01_u10_cond} discussed in the previous section. For the case of binary alphabet, the feasible space of $ \Oscr(\ut) $ contains a single off-diagonal permutation matrix, and the condition $ \Gamma(\ut) < 0  $ reduces to $ \ut(0,1) + \ut(1,0) < 0 $.

Note that $ \Gamma(\ut) < 0 $ is only a sufficient condition for the rate region to be non-empty. We do not yet know if there exists an achievable rate when  $ \Gamma(\ut) = 0 $. However, this is not a generic case of utility, since a small perturbation to the utility can give $ \Gamma(\ut) < 0$ or $ \Gamma(\ut) > 0 $. % necessary and sufficient for the general alphabet also. %However, in addition to the entropy, rates less than and including $ \log \;\Xi(\ut,P_X) $ are also achievable. This is in contrast to the case of binary alphabet where only the entropy was achievable. One may wonder whether the weighted capacity of the sender is equal to the entropy and the rate region same as the binary case. However,

\subsection{Lossy recovery}

% Define
% \begin{align}
%   \Pscr' = \curly{p' \in \distr : |p - p'|_{\infty} \leq \frac{\bd}{q(q-1)}}.
% \end{align}
% Using the convexity of the rate region, we have the following corollary.
% Let $ J \subseteq [n], |J| \geq 2 $ and define $$ N(q,|J|) = q\cdot(q-1)\hdots(q-(|J|-2)). $$ Define a set $ \Pscr_{J} $ as
% \begin{align}
%   \Pscr_{J} = \Big\{ P \in \distr &: |P(i) - \prbX(i)| \leq \frac{\bd}{N(q,|J|)} \non \\
%                                          &\;\;\;|P(j) - \prbX(j)| \leq  \frac{\bd}{N(q,|J|)} \non \\
%     &\;\;\;P(k) = \prbX(k) \;\;\forall \; k \in [n]\setminus J  \Big\}.
% \end{align}
% For a particular choice of $ J $, the set $ \Pscr_J $ consists of all those distributions $ P $ that differ from $ \prbX $ only at the symbols in $ J $. Further, the difference is at most $ \bd/N(q,|J|) $. For symbols not in $ J $, the distribution $ P $ is identical to $ \prbX $. Finally, define $ \Pscr' = \bigcup_{J \in [n]} \Pscr_J$.
Let $ j,k \in \Xscr $ be some distinct symbols and define  $ \Pscr'_{jk} $  as
\begin{align}
  \Pscr'_{jk} = \Big\{ P \in \distr : P(i) &=  \prbX(i) \;\;\forall \;i \neq j,k, 
  \big| P(i) - \prbX(i)\big| \leq \frac{\bd}{q-1} \mbox{ if } i = j,k \Big\}. \label{eq:type_class_within_delta}
\end{align}
For a given pair of distinct symbols $ j, k\in \Xscr $, the set $ \Pscr_{jk}' $ is defined as the set of distributions that coincide with $ \prbX$ on all symbols except $ j $ and $k $. Further, the distribution at $ j, k $ differ from $\prbX$  by $ \bd/(q-1)$.

Unlike the case of binary alphabet, we only give a partial characterization of the rate for the case of general alphabet. Using the sets $ \Pscr_{jk}' $, we find an interval of achievable rates around the entropy of the source. We also use this set to determine an upper bound on the infimum rate.
%We give a subset of achievable rates as a function of the set $ \Pscr' $ defined as
%\begin{align}
%  \Pscr' &= \bigcup_{j,k \in \Xscr} \Pscr'_{jk}. \label{eq:union_of_types}
%\end{align}
%\todo{$j,k\in \Xscr$?}

\begin{theorem}\label{thm:achie-rate-lossy_gen_alph}
Let  $ X \sim \prbX $, $ \bd >  0 $ and define $ \Pscr' = \bigcup_{j,k \in \Xscr} \Pscr'_{jk}$.
\begin{itemize}
\item[a)] If $\ut$ is such that $\Gamma(\ut) < 0 $, then  %\tododone{removed the $  \Rscr_{\bd}^{\sup}$ result, we would need to show an achievable rate for this, like the  Claim~\ref{clm:H(P+d)_achievable}  }
  \begin{align}
    \Rscr_{\bd}^{\inf} \leq \min_{ P' \in \Pscr'} H(P'). \non%, \quad \Rscr_{\bd}^{\sup} \geq \max_{ P'\in \Pscr'} H(P'). \non
    %\Big[ \min_{ X' \sim P' \in \Pscr'} H(X'),\; \max_{  X' \sim P'\in \Pscr'} H(X')\Big]\subseteq \Rscr_{\bd}. \non%   \Rscr_{\bd}^{\inf} &\in \left[R(\bd),\min_{p' \in \Pscr' } H(p')\right], \non \\
%    \Rscr_{\bd}^{\sup} &\in \left[\max_{p' \in \Pscr'} H(p'),1\right].
  \end{align}
\item[b)] Further, if $ \ut(i,j) < 0 $, $ \ut(i,j) = -c $, $ c > 0 $ for all distinct $ i,j \in \Xscr $, then $$ \Rscr = [\rateD,\log q]. $$
\end{itemize}
\end{theorem}
\begin{proof}
  Detailed proofs are in the  Appendix~\ref{append:gen_alph_lossy}.

  When $ \Gamma(\ut) < 0 $, we show that the receiver can construct a strategy with the image as a typical set around a distribution in the set $ \Pscr' $ defined above. % in \eqref{eq:type_class_within_delta} and \eqref{eq:union_of_types}. 
  The sequences in the typical set around the distribution $ \prbX $ will then be mapped by the sender to a sequence in the image within a distortion that is arbitrarily close to $ \bd $. Varying the distribution over the set $ \Pscr' $ gives the upper bound on $ \Rscr^{\inf}_\bd $. %a set of achievable rate region.
  For the case b), the objectives of the sender and the receiver coincide and hence the setting becomes equivalent to the case of cooperative communication.
\end{proof}
Unlike the earlier results, we only have an upper bound on  $ \Rscr^{\inf}_\bd $. %the infimum of the achievable rates.
In the following proposition, we present a sufficient condition on the utility $ \ut $ for $ \Gamma(\ut) < 0 $ to hold.

	\begin{proposition} \label{prop:suff-cond-for-pos-LP}
	Let the utility $\ut$ be such that for all distinct symbols $ i_0,\hdots,i_{K-1} \in \Xscr  $, there exists a pair $ i_{m} $ and $ i_{(m+1){\mathtt{mod}K}} $ such that $ \ut( i_{m}, i_{(m+1){\mathtt{mod}} K}) < 0 $.
%	 a symbol $ i \in \Xscr $, there exists a symbol $ j \in \Xscr $ for which $ \ut(i,j) < 0$.
        If
	\begin{align}
	&\min_{i,j: \ut(i,j) < 0} |\ut(i,j)| > (q-1)  \max_{i,j: \ut(i,j) \geq 0} \ut(i,j),\label{eq:least-pos-gre-max-neg}
	\end{align}
        then $ \Gamma(\ut) < 0 $.% $ \Big[ \min_{ X' \sim P' \in \Pscr'} H(X'),\; \max_{  X' \sim P'\in \Pscr'} H(X')\Big]\subseteq \Rscr_{\bd} $, where $ \Pscr' $ is as defined in the last section.%$ \Gamma(\ut) < 0 $.%Then, the set $[\min_{P \in \Pscr'}H(X'), \max_{P \in \Pscr'}H(X')]$ is contained in the achievable rate region.
        % achievable for all $p' \in \distr$ where $p'$ satisfies $|p-p'|_\infty \le \bd/q(q-1)$.
      \end{proposition}
      \begin{proof}
%  Recall the problem $\Oscr(\ut)$. % The feasible space of the linear program is the space of doubly stochastic matrices.
%   From Birkhoff-von Neumann theorem \cite{schrijver2003combinatorial}, we have that the set of doubly stochastic matrices is the convex hull of the set of $q \times q$ permutation matrices. Let $\Qscr = \{\Qscr_1,\hdots,\Qscr_L\}$ be the set of all $q \times q$ permutation matrices. Then, we can write the linear program as
%   \begin{align}
% \problemsmallawb{$\Oscr(\ut) :$}
% {\alpha  \in  [0,1]^{L} }
% {\displaystyle \sum_{i,j \in \Xscr} \sum_k\alpha_k \Qscr_i[i,j) V[i,j)  }
% {\begin{array}{r@{\ }c@{\ }l}
%    \sum_k \alpha_k  &=& 1  \\
%       \Qscr_i &\in& \Qscr
% 	\end{array}}         \non
% \end{align}
% %Clearly, the optimal occurs at some $\alpha^*$ where $\alpha_i^* = 1$ for some $i \in [L]$. This implies that the optimal value of $\Oscr(\ut)$ occurs at some permutation matrix $\Wscr^* \in \Qscr$.
% We now use this form to show that $\Gamma(\ut) = 0$.
 Let $\bar{Q} \in \Qscr$ be induced by a cyclic permutation $ \pi: \Xscr \to \Xscr $, then
\begin{align}
  \sum_{i,j} \bar{Q}(i,j)\ut(i,j) &= 	\ut(i_0,\pi(i_0)) + \ut(\pi(i_0),\pi^2(i_0)) + \hdots +  \ut(\pi^{q-1}(i_0),\pi^q(i_0)). \non
\end{align}
From the structure of utility, we have that $\ut(\pi^{j-1}(i_0),\pi^j(i_0)) $ for some $ j \in \{1,\hdots,q \}$.  Using \eqref{eq:least-pos-gre-max-neg}, we get that $ \sum_{i,j} \bar{Q}(i,j)\ut(i,j) < 0$ for all cyclic permutations.

Now any permutation can be written as a composition of finite cyclic permutations denoted as $ \pi_1,\hdots,\pi_A $. Thereby, $ \sum_{i,j} \bar{Q}(i,j)\ut(i,j)  $ can be written as a finite sum of $\ut(i_0,\pi_a(i_0)) + \hdots +  \ut(\pi_a^{l-1}(i_0),\pi_{a}^l(i_0))$, with $ a \leq A, l \leq q $. Using \eqref{eq:least-pos-gre-max-neg} for each of the summations, we get that $ \sum_{i,j} \bar{Q}(i,j)\ut(i,j) < 0$. Using this in $ \Oscr(\ut)$, we get that $ \Gamma(\ut) < 0 $.
\end{proof}

\subsection{Illustrative examples}
\label{sec:illustrative_examples}

The following example  uses  Proposition~\ref{prop:suff-cond-for-pos-LP} to show that the rate region is non-empty even when the utility of the sender incentivizes it to misreport almost all of its information.%\textit{pathological} liar. %graph $ \Gs $ is a complete graph. {\color{red} sender graph mentioned here, modify}
\begin{examp}
Let $\ut : \{0,1,2\} \times \{0,1,2\} \to \Real $ and consider the  following form of $ \ut $,
	\begin{align}
	\ut = \round{   \begin{array}{c c c }
		0 & 1 & 1   \\
		-4 &  0 & 1  \\
		-4 &  -4 &  0  \\
		\end{array}}. \non
	\end{align}
	From the utility, it can be observed that the sender has an incentive to lie about the symbols $ \{0,1\} $. This follows since $\ut(0,1), \ut(0,2), \ut(1,2) > 0 $.
%	\begin{align}
%	\ut(0,1), \ut(0,2), \ut(1,2) > 0. \non
%	\end{align}
%	Thus, $ \Gs $ is a $3$-cycle graph and hence is a complete graph.
	It can be easily observed that $ \ut $ satisfies the hypothesis of Proposition~\ref{prop:suff-cond-for-pos-LP} and hence $\Gamma(\ut) < 0 $. Thus, from Theorem~\ref{thm:lossl_rate_charac_gen} and \ref{thm:achie-rate-lossy_gen_alph}, there exists an achievable rate at which the receiver can communicate with the sender. % \todo{incomplete... what does this say about rates? finish the story}
\end{examp}

This example seems paradoxical -- even though the sender misreports most of its information, the receiver can ensure asymptotically vanishing probability of error. However, the main lesson to be drawn from it is that the \textit{magnitude} of the gains or losses from truth-telling or lying determine the extent of reliable communication, not the symbols alone.

We now discuss an example that shows that naively constructing strategies with more sequences in the image can increase the probability of error. % demonstrate the notion discussed above that utilizing more sequences than the supremum achievable rate can lead to a higher error probability.
\begin{examp} \label{eg:numerical_neary_type_classes}
		Consider the binary alphabet $ \Xscr = \{0,1\}$ and a utility defined as $ \ut(0,1) = 1$ and $ \ut(1,0) = -2 $.  Fix $ \delta = 0.2 $ and let $X \sim \prbX$  taking $ \prbX(0) = 0.3 $.  We define four different strategies for the receiver with increasing size of the image and show that the probability of error increases with the size of the image.

		Fix $ n \in \Nbb $ and let $p_1 \in \Pscr_n(\Xscr)  $ be the type closest to $ \prbX(0) $ such that $ p_1 \leq \prbX(0) $. Further, let $ p_2, p_3, p_4 \in \Pscr_n(\Xscr) $ be defined as $np_2 = np_1 + 1, np_3 = np_1 + 2 $ and $ 	np_4 = np_1 + 3 $.  Let $ U_{p_k}^n $ be the type class corresponding to $ p_k $ and for an arbitrary  $ \xvec_0 \in \bigcup_{k \leq i}U_{p_k}^n $, define the strategy $ g_n^i $ as
		\begin{align}
			g_n^i(\xvec) = \left\{ \begin{array}{c l}
				\xvec & \mbox{ if }\xvec \in \bigcup_{k \leq i}U_{p_k}^n \non \\
				\xvec_0 & \mbox{ else }
			\end{array} \right.,
		\end{align}
		%where is arbitrary.	Thus, the image of the strategy $ g_n^i $ is the union of all the type classes $ \{U_{p_1}^n \hdots,U_{p_i}^n\} $. %We define these strategies for all $ n $.

		For these strategies we compute the probability of recovered sequences $ \Dscr_\delta(g_n^i,s_n^i)$ denoted as $ \Pbb_{\Dscr}(g_n^i,s_n^i)$ for some $ s_n^i \in \best(g_n^i) $ and compare them with the probability of  recovered sequences  in the cooperative setting. In the cooperative setting all sequences in the $ \delta$-radius around $ \image(g_n^i) $ are recovered correctly. We denote this as $ \Bscr_\delta(\image(g_n^i))$.
		We compute these probabilities for $ n \leq 10 $ and plot them in Figure~\ref{fig:numerical_recovered_set}. The red colour signifies the strategic setting and corresponds to $ \Pbb_{\Dscr}(g_n^i,s_n^i)$  for some $ s_n^i \in \best(g_n^i)$. The blue colour corresponds to the cooperative setting.

		\begin{figure}[H]
			\centering
			\includegraphics[width=0.55\textwidth]{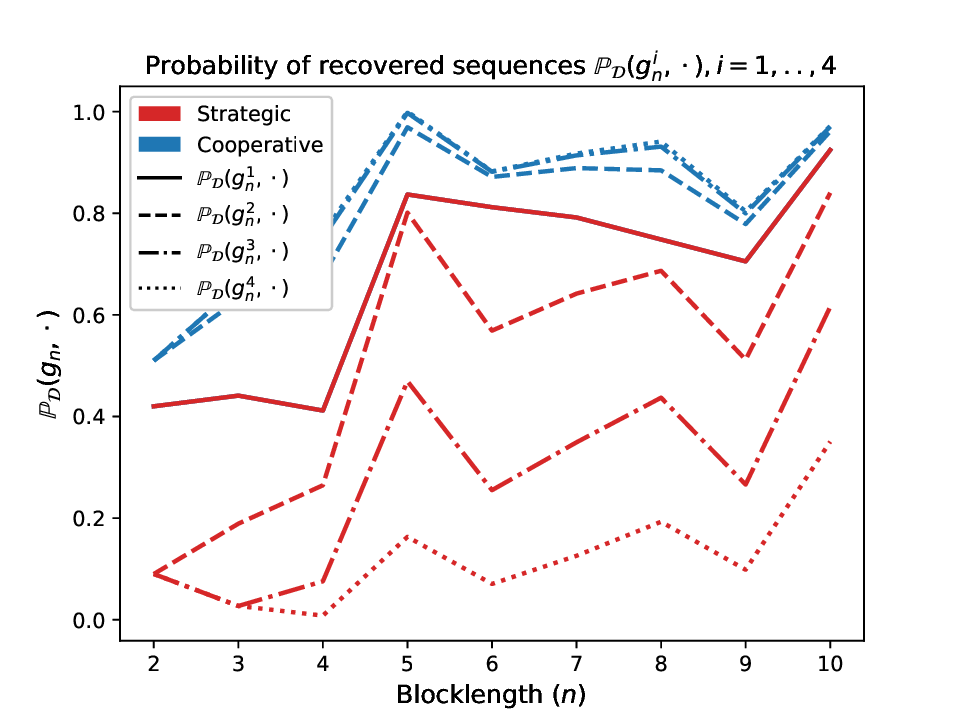}
			\caption{Comparison of  $ \Pbb(\Dscr_\delta(g_n^i,s_n^i))$ for some $ s_n^i \in \best(g_n^i) $ with the cooperative setting}
			\label{fig:numerical_recovered_set}
		\end{figure}

		For $ i = 1 $, the sets $ \Dscr_\delta(g_n^i,s_n^i) $ and  $ \Bscr_\delta(\image(g_n^i))$ are identical for all $ n $. This is shown by the red solid line in the figure.
		For $ i \geq 2$, the image of $ g_n^i$ includes sequences from neighbouring type classes. As the set $\Bscr_\delta(\image(g_n^i)) $ includes more sequences compared to $\Bscr_\delta(\image(g_n^1)) $, the probability of the set of correctly recovered sequences in the cooperative setting increases which is shown by the blue hashed and dotted curves  above the red solid line. In contrast,  we observe the opposite phenomenon in the strategic setting. Adding sequences in the image has a detrimental effect on the information recovery and the probability of the set of correctly recovered sequences decreases. This is shown by the red hashed and dotted curves below the red solid line.

		This occurs because when the number of sequences increase, the sender has more room to misreport its information. We briefly discuss this here. Since $\ut(0,1) + \ut(1,0) < 0 $, we have that for any distinct pair of sequences $ \xvec, \xvec' \in U_{p_1}^n$, $ \ut_n(\xvec',\xvec) < 0$. This is because $P_{\xvec',\xvec}(0,1) = P_{\xvec',\xvec}(1,0)$, which gives		$	\ut_n(\xvec',\xvec) = P_{\xvec',\xvec}(0,1)(\ut(0,1) + \ut(1,0)) < 0 $.
		Further, for a sequence $ \xvec \in U_{p_1}^n $  there exists a sequence $ \yvec \in U_{p_2}^n $, such that $ P_{\yvec,\xvec}(0,1) = p_2-p_1 $ and $P_{\yvec,\xvec}(1,0) = 0 $ which gives $\ut_n(\yvec,\xvec) 	= (p_2-p_1)a  > 0 $.
%		\begin{align}
%			\ut_n(\yvec,\xvec) 		&= (P_{\yvec}(0)-P_{\xvec}(0))\ut(0,1) +  P_{\yvec,\xvec}(1,0) \ut(1,0)  \non \\
%			&= (p_2-p_1)a - P_{\yvec,\xvec}(1,0)b \non
%		\end{align}
%		There exists a sequence $ \yvec $ such that $P_{\yvec,\xvec}(1,0) = 0 $ which gives  $\ut_n(\yvec,\xvec) > 0$.
 It implies that for the sequence $ \xvec $, the best response strategy maps $ \xvec $ to $ \yvec $, \ie, $ s_n^i(\xvec) = \yvec $. Moreover, $ d_n(\yvec,\xvec)  = p_2-p_1 $ and if $p_2-p_1 > \delta $ then  $ \xvec $ is not recovered correctly. It follows that no sequence in  $ U_{p_1}^n $  will be recovered correctly. Similarly, for $ g_n^3$, if $ p_3-p_2 > \delta$, then no sequence in $ U_{p_1}^n \bigcup U_{p_2}^n$ can be recovered correctly.

{\color{black}There is also a jump in the probability from $n = 4 $ to $ n = 5 $. To see this, consider $i = 1 $. For $n \geq 5$ we have $ 1/n \leq \delta =0.2 $ and hence the strategy $g_n^1$, unlike $g_1^1,\hdots, g_4^1$,  also recovers neighbouring type class of $U_{p_1}^n$. Moreover, since $ p_1 \leq \prbX(0) $, we get $ p_1 = 1/5 $ for $n=5$ and hence, the $\delta$-ball around $ U_{p_1}^5$ also includes the sequence with all ones, which has a high probability. This explains the jump at $n=5$ since more sequences with high probability are recovered and hence the higher share of $\Pbb(\Dscr_\delta(g_n^1,s_n^1))$. For $ 5 < n \leq 9 $, this $\Pbb(\Dscr_\delta(g_n^1,s_n^1))$ decreases a bit since the probability of the set $U_{p_1}^n $ and the neighbouring type class decreases as compared to $U_{p_1}^5$. However, this probability later approaches $ 1 $ for $n > 9 $ since most of the probability begins to concentrate in and around the type class.} % This pattern holds for all strategies $g_n^i$. 

		This shows that if the receiver naively adds sequences in the image of its strategy, then the probability of error may increase. Thus, the receiver has to deliberately restrict its image in order to recover the maximum possible information from the sender.
	\end{examp}

{\color{black}The following example provides a justification for the choice of rate of communication defined in Section~\ref{sec:rate_of_comm} and shows that measuring the rate using $ \image(g_n \circ s_n)$ is an overestimation.
\begin{examp}
	\label{eg:rate_gap_eg}
	Let $ \Xscr = \{0, 1, 2, 3\} $ and consider a sender with a utility $ \ut : \Xscr \rarr \Real $ defined as
	\begin{align}
		\ut = \round{   \begin{array}{c c c c}
				0 &  -2 & -13	& -18  \\
				1 &    0 	& -13  	& -18\\
				12 &  1 		&  0 		& -18 \\
				5 & 5 		&  5		&0 \\
		\end{array}}. \non %\label{eq:util_rate_examp}
	\end{align}
	Let $ \prbX$ be such that $ (\prbX(0),\prbX(1),\prbX(2),\prbX(3)) = (\frac{1}{36},\frac{1}{36},\frac{1}{36},\frac{11}{12})$ and $ \bd = 0 $.  For this utility, we have that $ \Gamma(\ut) < 0 $. This can be proved by taking subsets $\Yscr \subseteq \Xscr $   and checking the value of 
$ \ut(i_0,i_1) +  \ut(i_1,i_2)+ \hdots +  \ut(i_{|\Yscr|-1},i_0) $ for $ i_0,\hdots,i_{|\Yscr|-1} \in \Yscr $ (cf. proof of  Proposition~\ref{prop:suff-cond-for-pos-LP}). From Theorem~\ref{thm:lossl_rate_charac_gen}, we get that $ \entrX$ is achievable. 
We take $ \delta $ small enough such that
	\begin{align}
		\frac{1}{6}    - 46\delta  > 12\delta. \label{eq:bd_delta_defn}
	\end{align}

	\begin{claim} \label{clm:R_sup_in_eg}
	For the above utility,  $\Rscr^{\sup} \leq \log 2.73 $.
	\end{claim}
	\begin{proof}
	Let  $ \epsilon \in (0,\delta) $ and let $\Typepx$ be the typical set around $ \prbX$ and $\Bscr_{\delta}(\Typepx)$ be the sequences within an average Hamming distance of $\delta$ from   $ \Typepx $.	Any sequence of strategies that achieves vanishing error must have a sequence from $\Bscr_{\delta}(\Typepx)$ in its image. We prove that for all $ \xvec \in \Xscr^n $ with $P_{\xvec}(3) \leq 2/3 $, we have  $s_n(\xvec) \notin \Bscr_{\delta}(\xvec) $ and $s_n(\xvec) \in  \Bscr_{\delta}(\Typepx) $ which gives $ \Bscr_{\delta}(\xvec) \notin \Ascr_{\delta}^n(g_n,s_n)  $. This holds for all $ \delta $ satisfying \eqref{eq:bd_delta_defn}.

		Let $ \xvec \in \Xscr^n $ with $ P_{\xvec}(3) \leq 2/3 $. The best utility over $\Bscr_{\delta}(\xvec)$ is bounded as
	%	\begin{align}
		$	\max_{\yvec \in \Bscr_{\delta}(\xvec)} \ut_n(\yvec,\xvec) \leq \delta\max_{i\in \Xscr, j \in \{1,2\}} \ut(i,j)  = 12\delta$.  %\non%\frac{12}{3}. \non% \max\{\ut(3,0),\ut(3,1),\ut(3,2)\} = \frac{1}{24} 3.
	%	\end{align}
		% For any $ \xvec \in \Bscr_{\delta}(U_{1/3}^n) $, the best utility obtained over the set $ U_{1/3}^n $ is bounded as
		%\begin{align}
		%	\max_{\yvec \in U_{1/3}^n} \ut_n(\yvec,\xvec) \leq (\delta)\max_{i\in \Xscr, j \in \{1,2\}} \ut(i,j)  = \frac{1}{3}. \non% \max\{\ut(3,0),\ut(3,1),\ut(3,2)\} = \frac{1}{24} 3.
		%\end{align}
		%\begin{align}
		%	12a &< (\frac{1}{4}-a)5 - \frac{13}{12} \\
		%	17a &< \frac{5}{4} - \frac{13}{12} \\
		%	a &< \frac{1}{6*17}
		%\end{align}
		Moreover, the least utility obtained over $ \Bscr_{\delta}(\Typepx)$ is given as
		\begin{align}
			\min_{\yvec \in \Bscr_{\delta}(\Typepx)} \ut_n(\yvec,\xvec) &= \min_{\yvec \in \Bscr_{\delta}(\Typepx)} \Big(\sum_{j}P_{\yvec,\xvec}(3,j)\ut(3,j)  + \sum_{i\neq 3,j}P_{\yvec,\xvec}(i,j)\ut(i,j)\Big)\non \\
			&\geq \left(\frac{1}{3}-2\delta\right)  \min_{i}\ut(3,i) -  \left(\frac{1}{12} + 2\delta\right) \max_{i,j}|\ut(i,j)| \non \\
			&\geq  \left(\frac{5}{3} -\frac{3}{2}\right)    - 46\delta  > 12\delta, \non %= \frac{13.5}{36} > \frac{1}{3}.\non
			%\sum_{j \in \{1,2\}} \min_{i}P_{\yvec,\xvec}(i,j)  \ut(i,j) - \frac{1}{12}\min_{i}|\ut(i,1)| \non \\
			%&\geq \left(\frac{4}{12}  - \frac{1}{36}\right)\min_{i}\ut(3,i) - \frac{11}{24} - \frac{13}{24} \non \\
			%&\geq \frac{49.5}{36}- \frac{13}{12}  = \frac{13.5}{36} > \frac{1}{3}.\non
		\end{align}
		which follows since  for all $\yvec \in \Bscr_{\delta}(\Typepx)$, we have $\sum_{j}P_{\yvec,\xvec}(3,j)=P_{\yvec}(3) \geq 1/3 -2\delta $ and \newline $\sum_{i\neq 3,j}P_{\yvec,\xvec}(i,j) \leq (1/12+2\delta)$. % $ \min_{i}\ut(3,i)/3$ is the least positive utility that can be obtained and  $-(1/12 + 2\delta)18$ is the most negative utility that can be obtained which occurs when $\sum_{i\in\{0,1,2\}} P_{\yvec,\xvec}(i,3) =1/12 + 2\delta$.  %$P_{\yvec,\xvec}(3,3) \leq 2/3 $ and   $P_{\yvec,\xvec}(3,2) =  P_{\yvec,\xvec}(1,2) \leq  1/24 $ and
		The last inequality follows from  %the choice of $\delta$  in
		  \eqref{eq:bd_delta_defn}.
		This gives that  \newline  $\max_{\yvec \in \Bscr_{\delta}(\xvec)} \ut_n(\yvec,\xvec)  < 	\min_{\yvec \in \Bscr_{\delta}(\Typepx)} \ut_n(\yvec,\xvec)   $ and hence   $ s_n(\xvec) \in \Bscr_{\delta}(\Typepx) $. 
		%The largest rate of growth of such sequences is given by the type $ P $ defined in \eqref{eq:P_type_defn_eg}.
		Thus,  $\Ascr_{\delta}^n(g_n,s_n)$ contains only those sequences with type  $P$ where $ P(3) \geq 2/3 $. The largest rate occurs when the image includes sequences with the type
		\begin{align}
			P(i) = \frac{1}{9} \quad \forall \;i\in \{0,1,2\}, \;P(3) = \frac{2}{3} \non %\label{eq:P_type_defn_eg}
		\end{align}
		and is equal to  $ H(P) $ where $ H(P) \leq \log 2.73 $. This gives that $\Rscr^{\sup} \leq \log 2.73 $.
	\end{proof}

We now show the existence of a sequence of strategy $ \{g_n\}_{n \geq 1}$  with rate $ \entrX $  and  \newline $\lim_{n \rarr \infty} \log |\image(g_n \circ s_n)|/n > \Rscr^{\sup}. $

Let $ n = 12K, K \in \Nbb$ and define type classes $ U_{\prbX}^n$ and $ U_{\frac{1}{3}}^n $ as
	\begin{align}
		U_{\prbX}^n &= \{\xvec \in \Xscr^n : P_{\xvec}(i) = \prbX(i) \quad \forall \;i \in \Xscr \}, \non \\
		U_{\frac{1}{3}}^n &= \{\xvec \in \Xscr^n : P_{\xvec}(0) = 0, P_{\xvec}(i) = \frac{1}{3} \quad \forall \;i \in \{1,2,3\} \}.\non
	\end{align}
\begin{claim}
		Let  $ \xvec_0 \in U_{\prbX}^n \bigcup U_{\frac{1}{3}}^n $ be an arbitrary sequence and define a strategy $ g_n $ as
	\begin{align}
		g_n(\xvec) = \begin{cases}
			\xvec & \mbox{ if } x \in U_{\prbX}^n \bigcup U_{\frac{1}{3}}^n  \\
			\xvec_0 & \mbox{ else }
		\end{cases}. \non
	\end{align}
	Then, $ \{g_n\}_{n \geq 1}$  achieves the rate $ \entrX $  and  $\lim_{n \rarr \infty} \log |\image(g_n \circ s_n)|/n > \Rscr^{\sup}. $
\end{claim}
	%We show that this strategy achieves the rate $ \entrX $  and  $\lim_{n \rarr \infty} \log |\image(g_n \circ s_n)|/n > \Rscr^{\sup}. $ We proceed in the following steps. %For this strategy, we show that $ \lim_{ n \rarr \infty} \max_{s_n \in \best(g_n)} \Escr_{\delta}(g_n,s_n) \rarr 0 $ and

	%From Theorem~\ref{thm:lossl_rate_charac_gen} and its proof (to be proved in Section~\ref{sec:main_general_alph}), it follows that \newline $ \lim_{ n \rarr \infty} \max_{s_n \in \best(g_n)} \Escr_{\delta}(g_n,s_n) \rarr 0 $ for all $ \delta > 0 $. According to Definition~\ref{defn:achiev-rate}, the achievable rate  is $ \entrX $. We then show that $\lim_{n \rarr \infty} \log |\image(g_n \circ s_n)|/n = \log (3)  $ and that $ \Rscr^{\sup} \leq 1.003 $. This  implies that computing the rate using the image of the composition is an overestimation. %.
	%
	%
	%We show that for all $ s_n \in \pbest(g_n) $, we have $\Ascr_{\bd}(g_n,s_n) \subset \image(g_n \circ s_n)$. We proceed
	%This is shown through the following steps that will be proved later in Section~\ref{sec:illustrative_examples}. %
	\begin{proof}
		We proceed through following steps.
	\begin{enumerate}
		\item We first show that the above defined strategy achieves the rate $ \entrX $ %  $ \xvec \in U_{\prbX}^n $, $ s_n(\xvec) = \xvec $ and hence $ U_{\prbX}^n \subseteq \Ascr_{\delta}^n(g_n,s_n) $.
		\item Then,  we show that there exists a type class $ \Uhat^n $ such that $ s_n(\xvec) \in U_{\frac{1}{3}}^n $ for all $ \xvec \in \Uhat^n $ and which implies $U_{\frac{1}{3}}^n \subseteq \image(g_n \circ s_n)  $. The set $\Uhat^n$ grows at a rate $ \log (3) > H(P) $.% which is greater than $ \entrX $
	\end{enumerate}
	%\begin{align}
	%	\begin{array}{l | c | c | c | c}
		%		\yvec_1	& \overbrace{0\hdots}^{1/4} & \overbrace{1\hdots}^{1/4} & \overbrace{2\hdots}^{1/4} & \overbrace{3\hdots}^{1/4} \\
		%		\yvec_2	&0.. 1 ..2 .. & 3\hdots & 3\hdots & 3 \hdots \\
		%		\xvec_3	&0\hdots & 1\hdots & 0\hdots &3 \hdots \\
		%			\xvec_4	&0 .. 1..  & 0.. 1..& 0\hdots &3 \hdots
		%	\end{array} \non
	%\end{align}
 Consider a sequence $\xvec \in   \Bscr_{\delta}(\Typepx)  $.  The maximum  utility  over the set $ U_{\frac{1}{3}}^n $ is bounded as
	\begin{align}
		\max_{\yvec \in U_{\frac{1}{3}}^n } \ut_n(\yvec,\xvec) &\leq \sum_{j \in \Xscr\setminus\{3\}} \max_{i}P_{\yvec,\xvec}(i,j)  \ut(i,j) - \left(\frac{7}{12}- 2\delta\right)\min_{j}|\ut(j,3)| \label{eq:u(j,3)_term} \\
		&\leq \left(\frac{1}{36}+ 2\delta\right) \sum_{j  \in \Xscr\setminus\{3\}} \max_{i}\ut(i,j)  - \left(\frac{7}{12}- 2\delta\right)18  \label{eq:Pyx_leq_1_12} \\
		&= \frac{12+5+	5}{36}-\frac{21}{2}  +80\delta, \non % -\frac{7}{12},\non
	\end{align} %}_{\mbox{max. possible positive utility}}
where the  right term in \eqref{eq:u(j,3)_term} is the negative utility that the sender will necessarily incur by choosing a sequence from $U_{\frac{1}{3}}^n $ and   \eqref{eq:Pyx_leq_1_12} follows since $\max_{i}P_{\yvec,\xvec}(i,j) \leq P_{\xvec}(j) \leq 1/36 + 2\delta$ for all $ j \neq 3 $.  One the other hand the least utility over the set $U_{\prbX}^n $ is bounded as
%\begin{align}
	$\min_{\yvec \in U_{\prbX}^n }\ut_n(\yvec,\xvec) \geq -18(2\delta)$,  %\non
%\end{align}
which occurs when $ \sum_{i}P_{\yvec,\xvec}(i,3) = 2\delta $ and $ P_{\yvec,\xvec}(i,j) = 0$ for all other  $ i \neq j, j \neq 3 $. For small enough $ \delta $, we have $\min_{\yvec \in U_{\prbX}^n }\ut_n(\yvec,\xvec) \geq \max_{\yvec \in U_{\frac{1}{3}}^n } \ut_n(\yvec,\xvec)$ and hence $ s_n(\xvec) \in U_{\prbX}^n $. Moreover, a consequence of Lemma~\ref{lem:y-same-type} in the  proof of  Theorem~\ref{thm:lossl_rate_charac_gen} b) in Section~\ref{appen:thm_loss_rate_charac_gen}, we have $\ut_n(\yvec,\xvec) < 0 $ for all $ \yvec \in U_{\prbX}^n, \yvec \neq \xvec $.  Thus, %for $ \xvec \in U_{\prbX}^n $, $s_n(\xvec) = \xvec $ which gives a utility $ \ut_n(g_n \circ s_n(\xvec),\xvec) = 0 $ and hence
 $s_n(\xvec) = \xvec $  for all $ \xvec \in U_{\prbX}^n$ and $  U_{\prbX}^n \subseteq \Ascr_{\delta}^n(g_n,s_n) $ for all $s_n \in \best(g_n) $.  Moreover, from the proof of Claim~\ref{clm:R_sup_in_eg}, we have that $ U_{\frac{1}{3}}^n \cap  \Ascr_{\delta}^n(g_n,s_n) = \emptyset $. This gives $ R_{\delta}(g_n,s_n) = \log |U_{\prbX}^n|/n$ for all $ s_n \in \best(g_n)$ which proves that $ \entrX$ is achievable.

%{\color{red} Show $\entrX$ is achievable, in Theorem~\ref{thm:lossl_rate_charac_gen}?}

We now show point 2). Consider a type class $ \Uhat^n $ defined as
\begin{align}
	\Uhat^n  = \left\{\xvec \in \Xscr^n : P_{\xvec}(0) = \frac{1}{3},  P_{\xvec}(1) = \frac{1}{3}, P_{\xvec}(3) = \frac{1}{3} \right\}, \non
\end{align}
For a sequence $ \xvec \in \Uhat^n $, the maximum utility over the set $ U_{\frac{1}{3}}^n $ is bounded as
%\begin{align}
$	\max_{\yvec \in U_{\frac{1}{3}}^n} \ut_n(\yvec,\xvec) \geq 	\ut(2,0)/3 = 4,$ % \non
%\end{align}
which occurs by choosing a sequence $ \yvec $ such that  $ P_{\yvec,\xvec}(2,0) = 1/3 $ and $P_{\yvec,\xvec}(1,1) = P_{\yvec,\xvec}(3,3) = 1/3 $.
 The maximum utility over the set $U_{\prbX}^n$ is bounded as
\begin{align}
	\max_{\yvec \in U_{\prbX}^n} \ut_n(\yvec,\xvec) 	&= \max_{\yvec \in U_{\prbX}^n} \Big(\sum_{i \in \Xscr} P_{\yvec,\xvec}(i,0)\ut(i,0) +  \sum_{i \in \Xscr} P_{\yvec,\xvec}(i,1)\ut(i,1)\Big)\non \\
	&< \left(\frac{11}{12}- \frac{1}{3}\right)\ut(3,0)+\frac{1}{36}\sum_{i\in \{0,1\}}\ut(2,i) + \frac{1}{36}\ut(1,0)  \non \\
	%&= \frac{7}{12}\ut(3,0)+\frac{1}{24}\sum_{i\in \{0,1\}}\ut(2,i) + \frac{1}{24}\ut(1,0)  \non \\
	&= \frac{31.5}{12} + \frac{14}{36} < 4. \non %&<   \frac{1}{2}\ut(3,0)+\frac{1}{12}\ut(2,0)) + \frac{1}{12}\ut(1,0) + \frac{1}{4}\max_{j}\ut(j,1)  = \frac{2.1}{2} + \frac{1}{6} + \frac{1}{4} < 3. \non
\end{align}
Thus all for $ \xvec \in \Uhat^n$, we have $ s_n(\xvec) \in U_{\frac{1}{3}}^n$ and hence $  U_{\frac{1}{3}}^n \subseteq \image(g_n \circ s_n) \;\forall\; s_n \in \best(g_n)$. Observe that $\image(g_n \circ s_n) = U_{\prbX}^n \cup U_{\frac{1}{3}}^n$ whereas $ \Ascr_{\delta}^n(g_n,s_n) = U_{\prbX}^n$. 
%{\color{blue} can increase $ \ut(3,0)$ to 10, 11? will also affect \eqref{eq:Pyx_leq_1_12}, this messes up with sequences mapping to $ \Uhat $}
%There exists $ |\Uhat^n|$ number of unique sequences in $U_{\frac{1}{3}}^n $ that are mapped to by sequences in the set $ \Uhat^n $ and hence $  \Uhat^n \subseteq \image(g_n \circ s_n) $ (\textit{can use matching theorem}). 
Since $ U_{\frac{1}{3}}^n$ grows as $ \approx 2^{n\log(3)} $ for large $n $ (Stirling's approximation), we get that  $\lim_{n \rarr \infty} \log |\image(g_n \circ s_n)|/n = \log (3) > \Rscr^{\sup} $.	\end{proof}
\end{examp}}
%\todo{\textit{I did not understand this point} -- what is the unrecovered part of $\image(g_n\circ s_n)$? can you comment on it for completeness?}

\section{Conclusion}
\label{sec:concl}
In this paper, we proposed a novel theory of communication where a receiver wants to recover information from a utility-maximizing strategic sender. We posed this problem as a Stackelberg game with the receiver as the leader and studied the problem for two recovery criteria -- lossless and lossy. We determined a sufficient condition on the utility of the sender for the existence of  strategies for the receiver that achieve vanishing error. %\ach
  This condition is closely related to the notion of incentive compatibility from the theory of mechanism design. In our setting of communication, this condition characterizes the ``truthfulness'' of the sender. We showed that for the case of the binary alphabet, this condition is also a sufficient condition, without which the probability of error tends to one.

We show that for reliable communication, the receiver employs a selective decoding strategy  where it chooses to decode only a subset of messages correctly, and for the rest of the messages is deliberately makes an arbitrary error. Effectively, this strategy limits the choice of signals of the sender and restricts its tendency to misreport its information. We showed that it is sufficient to choose this subset such that the sequences in the high probability typical set is recovered within the specified distortion ensuring that the probability of error vanishes asymptotically.

We defined a generalized notion of rate of communication for our setting that computes the amount of resources required for communication. We showed that despite the existence of a clean channel, there may be a \textit{maximum rate} of communication above which the probability of error does not vanish asymptotically and may tend to one. In information theory, this bound arises mainly in the context of a noisy channel, where input messages beyond a rate threshold makes the probability of error tend to one. In our case, increasing the rate gives more freedom to the sender to lie about its information which causes the probability of error to asymptotically attain a non-zero value. Although we show the existence of a maximum rate less than unity only for the case of lossless recovery in binary alphabet, we conjecture through an example that similar bounds on the rate holds for other cases as well.

\section{Acknowledgments}

This research was supported by the grant CRG/2019/002975 of the Science and Engineering Research
Board, Department of Science and Technology, India. We thank the four anonymous reviewers and the Associate Editor  for their careful evaluation of the paper. Their comments and suggestions have vastly improved the paper in terms of the exposition of the model and the results.

\appendices

\section{Comparison with Related Models}
\label{appen:comparison_earlier_works}

	In this section, we compare our model with two models from literature, the mismatched decoding setting of Lapidoth~\cite{lapidoth1997role}  and the linked mechanism setup of Jackson and Sonnenschein~\cite{jackson2007overcoming}.
	%\subsection{Mismatched decoding \cite{scarlett2020information}}

	\subsection{Comparison with Lapidoth's mismatched distortion model}% \cite{lapidoth1997role}}

%In this section, we review the mismatched decoding model studied by Lapidoth in \cite{lapidoth1997role} and relate it to our model.
%We identify the connections to our setting and show that with a restriction on the decoding function, the setting studied by Lapidoth reduces to the formulation discussed in our paper.

Lapidoth in \cite{lapidoth1997role} considers a mismatched rate-distortion formulation where the encoder and the decoder are chosen using distinct distortion criteria.  The encoder observes the  memoryless source $ X \in \Xscr $ generated according to some distribution $ \prbX $. Let $ \wi{\Xscr} $ be a finite alphabet and consider a codebook of rate $ R $ defined as $ \Cscr = \{\xvec^1,\hdots,\xvec^{2^{nR}}\} \subseteq \wi{\Xscr}^n$. The encoder is concerned with a distortion function $  \dhat : \Xscr \times \wi{\Xscr} \to \Real  $ and chooses a message from the codebook according to the encoding function $ \phi_n : \Xscr^n \to \wi{\Xscr}^n $ where $ \phi_n(\xvec) = \argmin_{\yvec \in \Cscr} \dhat_n(\xvec,\yvec)$ with
%\begin{align}
	$\dhat_n(\xvec,\yvec) = \frac{1		}{n}\sum_{i} \dhat(x_i,y_i)$.  %\non
%\end{align}
Observe that the encoding function of the sender only depends on the codebook designed by the decoder.  The decoder designs a decoding function $ \psi_n : \Cscr \to \bar{\Xscr}^n $ considering a distinct distortion function $ d: \Xscr \times \bar{\Xscr} \rarr \Real $. Here $ \psi_n $ is allowed to be a randomized function. Let the optimal encoders in response to the codebook designed by the receiver be given as $ \best(\psi)$. For a given distortion level $ \bd > 0 $ and $ \epsilon, \delta > 0 $, a rate $ R $ is achievable if there exists a  codebook at a rate $ R $ such that
\begin{align}
	\max_{\phi_n \in \best(\psi_n)} \Pbb( d_n(\psi_n \circ \phi_n(\Xvec),\Xvec) > \bd + \delta) < \epsilon.
\end{align}
The maximum over the optimal encoders follow because ties are broken pessimistically. The function under study in this mismatched setting is the distortion-rate function $ D(R) $ which is the infimum of all $ \bd $ for which the pair $ (R,\bd) $ is achievable.

The following points summarize the differences and the similarities of this model and the setting studied in our paper.
\begin{itemize}
	\item A crucial distinction is that the choice of the encoder chosen by the sender in the mismatched setting depends on the decoder only implicitly via the codebook. In contrast, the sender in our setting chooses its strategy as a direct response to the receiver's strategy. Moreover, $ \psi_n $ and $ \phi_n $ can be randomized strategies, unlike in our setting where $g_n $ and $s_n$ are deterministic. % This makes recovery of information from the sender more difficult.
	\item However, in some cases, the mismatched setting helps in constructing an achievable strategy for our setting.
	Let $ \wi{\Xscr} = \bar{\Xscr} = \Xscr $ and consider a restriction of the decoding function to deterministic functions, where for a choice of codebook $ \Cscr $, the decoder is $ \psi_n(\xvec) = \xvec $ for all $ \xvec \in \Cscr $.  Further, let $ \ut(x,y) = -\dhat(x,y) $ for all $ x, y  \in \Xscr $. In this case, our setting coincides with the mismatched setting, since the payoff of the sender matches the distortion function of the encoder and the restricted decoder is a valid strategy for the receiver. We use this observation to characterize the achievable rates in our setting and is discussed in Theorem~\ref{thm:dist_rate_bound} in the Appendix~\ref{appen:suff_cond}.

\end{itemize}

\subsection{Comparison with Jackson and Sonnenschein's linked mechanism  model}% \cite{jackson2007overcoming}}

In this section, we compare our setting with the linked mechanism model studied by Jackson and Sonnenschein in \cite{jackson2007overcoming}. Although, the two models are similar thematically, due to a fundamental difference in the objective of the receiver (``principal'' in \cite{jackson2007overcoming}), the two settings have distinct conclusions.  %Recall from Section~\ref{sec:related_work} where we briefly mention the setting studied in \cite{jackson2007overcoming}. Here we present the mathematical model and highlight the differences and similarities with the setting of our paper.

The authors in \cite{jackson2007overcoming} study a setting where a principal wishes to design a series of mechanisms that approximate a social choice function, where each mechanism links $K$ independent problems together. The terminology presented in this study can be quite different from the communication theory literature. For this reason, we modify some of the notation and  make the following adjustments to adapt the model to the setting of our paper. We restrict the setting of \cite{jackson2007overcoming} to a setting with a single agent and deterministic strategies.
%To ease the comparison with our work, we also reuse some of our notation.  We also restrict the functions to deterministic functions.

%{\color{red} repeat of incentive compatibility section?}

Let $ \Xscr^K $ be the information set of the agent and let $ \Yscr^K $ be the outcome space. The agent observes the information generated independently according to $ \prbX \in \Pscr(\Xscr) $. The social choice function is defined as $ f: \Xscr \to \Yscr $\footnote{The social choice function occurs frequently in the game theory literature in auction theory and in social choice theory (Ch.12 and Ch.21 in \cite{maschler2013game}). For instance,  an allocation function in auction theory is such a function that maps the true valuations (preferences) of each agent to the allocation (outcome) of a certain object, say for instance to the agent with the highest valuation.}.  A $ K$-linked mechanism is the pair $  (\Mscr^K,g_K ) $, where $ \Mscr^K $ is the message space of the agent and $ g_K : \Mscr^K \to \Yscr^K $ is a \textit{decision} function that maps the messages to the outcome space. The strategy of the agent is  $ s_K: \Xscr^K \to \Mscr^K $. For a pair of $ (\xvec,\yvec) $, with $ \xvec \in \Xscr^K$ and $ \yvec \in \Yscr^K $, the utility of the agent is given as $ \sum_k \ut(y_k,x_k) $, where $ \ut : \Xscr \times \Yscr \to \Real $.

The authors show that if the social choice function $f$ is \textit{ex-ante Pareto efficient}, \ie, it satisfies
\begin{align}
	\sum_{x}\prbX(x)\ut(f'(x),x) \leq \sum_{x}\prbX(x)\ut(f(x),x) \quad \forall x \in \Xscr, \forall\;f', \label{eq:ex_ante_efficient}
\end{align}
then there exists a sequence of linked mechanisms $(\Mscr^K,g_K), K \geq 1$ and a corresponding sequence of best responses $ \{ s_K\}_{K\geq 1} $ that approximate  $ f $ as
\begin{align}
	\lim_{K\to\infty} \max_{k \leq K} \Pbb\Big((g_K \circ s_K (\Xvec))_k  \neq f(X_k)\Big) = 0. \label{eq:approximate_f}
\end{align}

In the case of an outcome space $ \Yscr = \Xscr $, %  and a social choice function that is an identity function,
the setting studied by Jackson and Sonnenschein has certain parallels to the setting of our paper which we discuss in the following points.
\begin{itemize}
\item  In our setting,
%the agent is the sender and the principal is the receiver. The
the utility function of the sender is similar to the utility function of the agent and the social choice function in our setting is the identity function. The $K$-linked mechanism is akin to the $n$-block structure of our setting.
\item   However, the notion of \textit{approximation} is different in our setting.  %The approximation in defined in  \eqref{eq:approximate_f} looks at all the outcomes that do not match the target outcomes.
In \eqref{eq:approximate_f}, a social choice function is approximated if the maximum of the error asymptotically vanishes to zero. In contrast, we consider a distortion function and  require that for all $ \delta >  0 $,
\begin{align}
	\max_{s_n \in \best(g_n)} \Pbb(d_n(g_n \circ s_n(\Xvec),\Xvec) > \delta) \rarr 0. \non
\end{align}
Importantly, we study the worst-case error whereas \eqref{eq:approximate_f} only requires an existence of a corresponding sequence of best responses. This leads to a qualitative difference in our results. In particular, in our setting, if $\ut(i,j)>0$ for some $i,j$, perfect recovery is impossible (cf. Theorem~\ref{thm:lossl_delta=0_empty_R}). %Whereas in \cite{jackson2007overcoming} recovery is possible when $f$ is identity if ex-ante Pareto efficiency holds for this $f$.
\item The condition that enables the approximation is the notion of ex-ante Pareto efficiency. In our setting, we have an analogous condition $ \Gamma(\ut) < 0 $ which, as discussed in Lemma~\ref{lem:y-same-type} discussed in Section~\ref{appen:thm_loss_rate_charac_gen} in the Appendix, is equivalently given as $$  \sum_{x,y}P_{X,Y}(x,y)\ut(x,y)  \leq 0 \quad \forall \;P_{X,Y} \mbox{ with } P_X = P_Y, $$ and $  \sum_{x,y}P_{X,Y}(x,y)\ut(x,y) = 0 $ if and only if $ P_{X,Y}(i,i) = P_X(i)  $ for all $ i \in \Xscr $.
\item The proof ideas of \cite{jackson2007overcoming} %of Theorem~\ref{thm:jackson_sonnenschein}
has similar ideas as the proof of our results. They consider a type class with type closest to $ \prbX $ and restrict the messages of the sender to this type class. This restriction along with \eqref{eq:ex_ante_efficient} ensures that the function $f$ is approximated according to \eqref{eq:approximate_f}. %With the restriction of messages to sequences belonging to the type class  $ U_{P^*}^K $, the Definition~\ref{defn:ex-ante_pareto} ensures that the sender does is forced to report truthfully when its information lies in the set $ U_{P^*}^K $. The concentration result implies that for large $ K $, with high probability, the information observed by the sender will lie in the set $U_{P^*}^K $ which ensures that the function is implemented approximately.

In our setting too, the sender has to choose a sequence from a type class, wherein the condition $ \Gamma(\ut) < 0$ ensures that it reports correctly for sequences in and around a neighbourhood of the type class.
\item Finally, we present an information-theoretic analysis of the setting of information extraction from a strategic sender. We are concerned with determining strategies for the receiver with vanishing error %\ach
  and the corresponding achievable rates. Jackson and Sonnenschein do not present any such analysis in their linked mechanism setting.
\end{itemize}

%Thus, our setting is analogous to the setting of Jackson and Sonnenschein in regards to the linking and the ex-ante Pareto efficiency. We consider a distinct notion of approximation that is based on a distortion function. Moreover, we also study a corresponding lossy formulation which is not studied by the authors.
%
%\hltodo{We finally remark that our setting cannot recover the result of  Jackson and Sonnenschein, specifically}{not sure} Theorem~\ref{thm:jackson_sonnenschein}. To see this, recall the Theorem~\ref{thm:lossl_delta=0_empty_R} that states that  if there exists a pair of symbols $ i,j \in \Xscr $ such that $ \ut(i,j) > 0 $ then for all sequences  of strategies $ \{g_n\}_{n \geq 1} $ we have $
%\lim_{n \rarr \infty} \max_{s_n \in \best(g_n)} \Pbb( g_n \circ s_n(\Xvec) \neq \Xvec) > 0 $.
%This further implies that
%\begin{align}
%\lim_{n \rarr \infty} \max_{s_n \in \best(g_n)} \Pbb( (g_n \circ s_n(\Xvec))_k \neq \Xvec_k) > 0 \quad \mbox{ for some } k \in [n],\non
%\end{align}
%and hence \eqref{eq:approximate_f_identity} cannot hold.

%Thus an equivalent Theorem~\ref{thm:jackson_sonnenschein} does not hold in our setting.

%discussion of proof
%The proof of the theorem basically states that when the agent is restricted to choose a message from a type class, then the sender reports truthfully on them and as linkings grow by LLN it follows that max error goes to zero. Can our setting of vanishing distortion be related?

\section{Proofs of Theorem~\ref{thm:lossl_delta=0_empty_R} and Convexity of Rate Region}
\label{appen:proofs_prelim_res}
% \subsection{Characteristics of rate region}

\subsection{Proof of Theorem~\ref{thm:lossl_delta=0_empty_R}}%{lem:D_ind_set}}
\label{append:lossl_delta=0_empty_R}

Before proceeding with the proof, we first define a notion of a \textit{sender graph}. We then discuss two results that highlight the structure of the graph.

The utility determines the extent of the misreporting nature of the sender and this can be succinctly characterized by a graph  on the space of sequences $ \Xscr^ n$. We term this as the sender graph and it is induced by the utility of the sender.% and is defined as follows.%Consider the following definition of a sender graph.
\begin{definition} [Sender graph] \label{defn:sender_graph}
The sender graph, denoted as $ \Gs^n = (\Xscr^n,E_{\sf s}) $, is the graph where $ (\xvec,\yvec) \in E_{\sf s} $ if
\begin{align}
	\mbox{ either }\;\ut_n(\yvec,\xvec) \geq  0  \;\; \mbox{or}\;\;    \ut_n(\xvec,\yvec) \geq  0. \non
\end{align}
For $n = 1$, the graph $\Gs^1$ is denoted as $\Gs$.% and referred to as the \textit{base graph}.
\end{definition}
Thus, two vertices $\xvec$ and $\yvec$ are adjacent in $ \Gs $ if the sender has an incentive to report one sequence as the other.

% For the proof, we use the structure of the sender graph $ \Gs^n $ defined in Section~\ref{append:sender_graph_ind_set}.

For any pair of strategies $ (g_n,s_n) $ and $ \bd = 0 $, recall that $ \Dscr(g_n,s_n) $ is the set of correctly recovered sequences and is given as $ \Dscr(g_n,s_n) = \{ \xvec \in \Xscr^n : g_n \circ s_n(\xvec) = \xvec \} $ (cf \eqref{eq:seq-rec-within-d}).
The following lemma states that for the worst case $ s_n $, the set $\Dscr(g_n,s_n) $ is an independent set in $ \Gs^n $. The lemma is proved as a part of Theorem~3.1 in \cite{vora2024shannon}.
\begin{lemma} \label{lem:D(gn,sn)_ind_set}
Let $ \Gs^n $ be the graph induced by the utility $ \ut $ and   let $ g_n $ be any strategy. Then, for all $ s_n \in \pbest(g_n) $, the set $ \Dscr(g_n,s_n) $ is an independent set in the graph $ \Gs^n $.
\end{lemma}

We define few sets which will be used in the subsequent analysis. For $ \epsilon > 0 $, recall the typical set $\Typep$ defined as
%\begin{align}
$\Typep  = \Big\{ \xvec \in \Xscr^{n} :  p(i) - \epsilon < P_{\xvec}(i) < p(i) + \epsilon \quad \forall \; i \in \Xscr \Big\}$, %\label{eq:TKq-set-defn}
%\end{align}
where $P_{\xvec}$ is the type of the sequence $\xvec \in \Xscr^{n}$.
For a distribution $P \in \Pscr_n(\Xscr)$, recall the  type class  $U_{P}^n \subseteq \Xscr^n$ defined as
%\begin{align}
$U_{P}^n = \Big\{\xvec \in \Xscr^n : P_{\xvec} = P \Big\}$. % \label{eq:type-class}
%\end{align}
For any pair of sets $ U_1^n, U_2^n $ and for any $ \xvec \in U_1^n $, the number of sequences in $ U_2^n $ that are adjacent to $ \xvec $ is given by
%\begin{align}
$ \Big|\left\{ \yvec \in U_2^n : \xvec \sim \yvec \mbox{ in } \Gs^n \right\}\Big|.$ %\label{eq:edges_bet_two_types} %\non%, \non \\
%\Delta_2(\widehat{\xvec}, U_{P_1}^n ,U_{P_2}^n) = \Big|\left\{ \bar{\xvec} \in U_{P_1}^n : \bar{\xvec} \sim \widehat{\xvec} \mbox{\; in \;} \Gs^n \right\}\Big|.
%\end{align}

%Thus, $ \Delta(\xvec, U_1^n ,U_2^n)  $ is the number of sequences in $ U_2^n $ that are adjacent to $ \xvec $. Similarly, for $ \yvec \in U_2^n $, we have $   \Delta(\yvec, U_2^n ,U_1^n) $ as the number of sequences in $ U_1^n $ that are adjacent to $ \yvec $.

For any two type classes we have the following lemma that demonstrates the bi-regular nature of subgraph of $ \Gs^n $ induced by the sets.%Thus, for a $ \xvec \in U_{P_1}^n $, $\Delta(\xvec,U_{P_1}^n,U_{P_2}^n,\delta) $ is the outdegree of  $ \xvec $, while $  \Delta(\yvec,U_{P_2}^n,U_{P_1}^n,\delta) $ is the indegree of $ \yvec \in U_{P_2}^n $.
\begin{lemma} \label{lem:bireg_subgraph_undir}
Let $ P_1, P_2 \in \Pscr_n(\Xscr) $ be two types and let $ U_{P_1}^n $ and $ U_{P_2}^n$ be the respective type classes. Then,
\begin{align}
\Big|\left\{ \yvec \in U_{P_2}^n : \xvec \sim \yvec \mbox{ in } \Gs^n \right\}\Big| = \Big|\left\{ \yvec \in U_{P_2}^n : \xvec' \sim \yvec \mbox{ in } \Gs^n \right\}\Big| \quad \forall \; \xvec,\xvec' \in U_{P_1}^n. \non
\end{align}
%  \begin{align}
%     \Delta(\xvec, U_{P_1}^n ,U_{P_2}^n) = \Delta(U_{P_1}^n ,U_{P_2}^n) \quad \forall \; \xvec \in U_{P_1}^n, \non \\
%    \Delta(\yvec,U_{P_2}^n, U_{P_1}^n) = \Delta(U_{P_2}^n, U_{P_1}^n) \quad \forall \; \yvec \in U_{P_2}^n, \non
%  \end{align}
% for some $  \Delta(U_{P_1}^n ,U_{P_2}^n) , \Delta(U_{P_2}^n, U_{P_1}^n) \geq  0 $.
Further, for any $ \xvec_0 \in U_{P_1}^n, \yvec_0 \in U_{P_2}^n $,
\begin{align}
\Big|\left\{ \yvec \in U_{P_2}^n : \xvec_0 \sim \yvec \mbox{ in } \Gs^n \right\}\Big| \; |U_{P_1}^n| = \Big|\left\{ \xvec \in U_{P_1}^n : \yvec_0 \sim \xvec \mbox{ in } \Gs^n \right\}\Big| \; | U_{P_2}^n |. \non%\frac{ |U_{P_1}^n|}{|U_{P_2}^n|} = \frac{\Delta_2(U_{P_1}^n,U_{P_2}^n,\delta)}{\Delta_1(U_{P_1}^n,U_{P_2}^n,\delta)} . \non
\end{align}
\end{lemma}
\begin{proof}
Since $ U_{P_1}^n $ and $ U_{P_2}^n $ are type classes, all sequences in $ U_{P_1}^n $ will have the same number of adjacent sequences in $ U_{P_2}^n $. This holds similarly for $ U_{P_2}^n $ and the first assertion is thus evident.

Now, the total number of edges with one end in $ U_{P_1}^n $ and other end in  $ U_{P_2}^n $ is given by $ \Big|\left\{ \yvec \in U_{P_2}^n : \xvec_0 \sim \yvec \mbox{ in } \Gs^n \right\}\Big| \; |U_{P_1}^n|  $. Moreover, this number must be equal to the number of edges with one end in $ U_{P_2}^n $ and other end in  $ U_{P_1}^n $ which is  $ \Big|\left\{ \xvec \in U_{P_1}^n : \yvec_0 \sim \xvec \mbox{ in } \Gs^n \right\}\Big| \; | U_{P_2}^n | $. This completes the proof. % Clearly, they have to be equal which gives $ \Delta_1(U_{P_1}^n,U_{P_2}^n,\delta) |U_{P_1}^n| = \Delta_2(U_{P_1}^n,U_{P_2}^n,\delta) | U_{P_2}^n |$.
\end{proof}
Define
\begin{align}
\Delta(U_{P_1}^n,U_{P_2}^n) = \Big|\left\{ \yvec \in U_{P_2}^n : \xvec \sim \yvec \mbox{ in } \Gs^n \right\}\Big|  \quad \forall \;\xvec \in U_{P_1}^n. \non
\end{align}

%{\color{red} check references}
%Let $ \Ihat^n  $ be the largest independent set in the graph $ \Gs^n $. Since for any $ g_n $ and $ s_n \in \pbest(g_n) $, $ \Dscr(g_n,s_n) $ is an independent set in $ \Gs^n $, we have $ \Dscr(g_n,s_n) \subseteq \Ihat^n $ and hence $ \Pbb(\Dscr(g_n,s_n)) \leq \Pbb(\Ihat^n) $.

\begin{proof}[Proof of Theorem~\ref{thm:lossl_delta=0_empty_R}]
We now proceed with the proof of Theorem~\ref{thm:lossl_delta=0_empty_R}. Let $ g_n $ be an arbitrary strategy of the receiver. We use Lemma~\ref{lem:D(gn,sn)_ind_set} and Lemma~\ref{lem:bireg_subgraph_undir} to show that  only a fraction of the sequences in the high probability set can be recovered correctly by the receiver under any strategy. This gives a  lower bound on  $ \max_{s_n \in \best(g_n)} \Pbb( g_n \circ s_n(X) \neq X) $ which, asymptotically, is strictly greater than zero. Thus, for any sequence of strategies $ \{g_n\}_{n \geq 1}$, the worst-case probability of error does not vanish.

Without loss of generality, let  $ \ut$  be such that $ \ut(1,0) \geq 0 $. As we shall see, the choice of symbols $ 0, 1 $ does not matter for the proof. Let $ \Gs^n $ be the graph induced by the utility $ \ut $.  Let $ \epsilon < \min_{i \in \Xscr} \prbX(i)$ and consider the typical set $ \Typepx$. Let $ \Pscr_n(\Typepx) \subseteq \Pscr_n(\Xscr) $ be the set of types induced by the sequences in  $ \Typepx $, \ie,
\begin{align}
\Pscr_n(\Typepx) = \Big\{ P \in \Pscr_n(\Xscr) : \prbX(i) - \epsilon < P(i) <  \prbX(i) + \epsilon \Big\}. \non
\end{align}
Let $ \Pscr_1,\hdots, \Pscr_T$, with $ \Pscr_i \subseteq \Pscr_n(\Typepx) $ be classes of types with the property that for all $ i \in \{1,\hdots,T\}$ and any two pair of types $ P_1, P_2 \in \Pscr_i $, we have $ P_1(j) = P_2(j) $ for all $ j \in \Xscr \setminus \{0,1\} $. Thus the types within a set $ \Pscr_i $ differ only in the symbols $ \{0,1\} $. We show in the following claim that only a fraction of sequences from the type classes with types from the set $ \Pscr_i, i \in \{1,\hdots,T\} $ are recovered correctly. For the purpose of the proof, we consider $ \Pscr_1 $ and  $ \shat_n \in \pbest(g_n)$. %For the type classes $ U_P^n $ with type $ P \in \Pscr_i $, we have the following claim.

\begin{claim}
Let $ \Pscr_1 $ be the above defined set of types. Then, there exists a $ \beta < 1 $ such that
\begin{align}
	\left| \bigcup_{P_i \in \Pscr_1} U_{P_i}^n \cap \Dscr(g_n,\shat_n)\right| < \beta  \left| \bigcup_{P_i \in \Pscr_1} U_{P_i}^n \right|. \non
\end{align}%and let $ \{U_{P_i}^n,P_i \in \Pscr_1 \}$ be the type classes with types in $ \Pscr_1$. Then, all but one type class from the set $ \{U_{P_i}^n,P_i \in \Pscr_1 \}$ have less than half the fraction of sequences in $ \Dscr(g_n,\shat_n) $.
\end{claim}
\begin{proof}
Let $ P_1,P_2 \in \Pscr_1 $ be two types such that $ P_1(1) < P_2(1) $.  Let $ U_{P_1}^n $ and $ U_{P_2}^n $ be the respective type classes. First, we show that  $ \Delta(U_{P_1}^n,U_{P_2}^n), \;\Delta(U_{P_2}^n,U_{P_1}^n) $ in Lemma~\ref{lem:bireg_subgraph_undir} are positive. Due to the structure of $ \Pscr_1 $, for every sequence $ \bar{\xvec} \in U_{P_1}^n  $, there exists a corresponding sequence  $ \widehat{\xvec} \in U_{P_2}^n $ such that
% \begin{align}
	$  \ut_n(\widehat{\xvec},\bar{\xvec}) = (P_2(1)-P_1(1)) \ut(1,0) \geq 0$. % \non
	% \end{align}
Thus, there exists edges between  $ U_{P_1}^n $ and $ U_{P_2}^n $ in the graph $ \Gs^n $ and hence $ \Delta(U_{P_1}^n,U_{P_2}^n), \;\Delta(U_{P_2}^n,U_{P_1}^n) $ in Lemma~\ref{lem:bireg_subgraph_undir} are positive.

% Observe that the above holds for any two types $ P_1, P_2 \in \Pscr_1 $. %We now use this bound $\Pbb(\Ihat^n)$ that will in turn bound $ \Pbb(\Dscr(g_n,\shat_n)) $.

We now prove the claim. Since $ \shat_n \in \pbest(g_n) $, from Lemma~\ref{lem:D(gn,sn)_ind_set}, we have that $ \Dscr(g_n,\shat_n) $ is an independent set in the graph $ \Gs^n $.  We show that \textit{only one} type class in %all \textit{except one} type class from
$ \Pscr_1 $ %the set $ \{U_{P_i}^n,P_i \in \Pscr_1 \}$
can have more than half of its sequences in %only have a partial intersection  with
the set $ \Dscr(g_n,\shat_n) $. %Specifically, the ratio of the size of intersection to the size of the type class is no more than  $1/2$. We prove this by contradiction.

For $ P_i \in \Pscr_1 $, let $ \beta_i = |U_{P_i}^n \cap \Dscr(g_n,\shat_n)|/ |U_{P_i}^n| $. Let $ i^* = \argmax_{i} \beta_i $ and   % We now show that for all $ i \neq i^*$, we have $ \beta_i \leq 1/2 $, \ie,  all but one type class from  $ \{U_{P_i}^n,P_i \in \Pscr_1 \}$  have less than  half the fraction of sequences in $ \Dscr(g_n,\shat_n) $.
% suppose there is a fraction $ \beta_i $, $ i \neq i^* $ such that $ \beta_i > 1/2 $. Then, it follows
suppose that  $ \beta_{i^*} > 1/2 $.
For $ i \neq i^*$ define
\begin{align}
	V(U_{P_i}^n,U_{P_{i^*}}^n \cap \Dscr(g_n,\shat_n)) = \Big\{ \yvec \in U_{P_{i}}^n  : \exists \; \xvec  \in U_{P_{i^*}}^n \cap \Dscr(g_n,\shat_n)\; s.t.\; \xvec \sim \yvec \mbox{ in } \Gs^n\Big\}. \non
\end{align}
The above set consists of all the sequences in $ U_{P_i}^n $ that are adjacent to some sequence in $ U_{P_{i^*}}^n \cap \Dscr(g_n,\shat_n) $. Also, the sequences in $ U_{P_{i^*}}^n \cap \Dscr(g_n,\shat_n) $ only have edges with $ V(U_{P_i}^n,U_{P_{i^*}}^n \cap \Dscr(g_n,\shat_n)) $ in the set $ U_{P_i}^n $. Now the number of edges with one end in $ U_{P_{i^*}}^n \cap \Dscr(g_n,\shat_n)$ and other in $  V(U_{P_i}^n,U_{P_{i^*}}^n \cap \Dscr(g_n,\shat_n)) $ is $ \Delta(U_{P_{i^*}}^n ,U_{P_i}^n)\; |U_{P_{i^*}}^n \cap \Dscr(g_n,\shat_n)| $. The number of edges with one end in $ V(U_{P_i}^n,U_{P_{i^*}}^n \cap \Dscr(g_n,\shat_n)) $ and other in $ U_{P_{i^*}}^n $ is given by $ \Delta(U_{P_i}^n,U_{P_{i^*}}^n )\;|V(U_{P_i}^n,U_{P_{i^*}}^n \cap \Dscr(g_n,\shat_n)) |  $. The latter may include edges with sequences in $ U_{P_{i^*}}^n $ outside $ U_{P_{i^*}}^n \cap \Dscr(g_n,\shat_n) $ as well and hence\footnote{The arguments in this section borrows ideas from Section IV in \cite{kangarshahi2021single}}
\begin{align}
	\Delta(U_{P_i}^n,U_{P_{i^*}}^n )\big|V(U_{P_i}^n,U_{P_{i^*}}^n \cap \Dscr(g_n,\shat_n))\big| &\geq  \Delta(U_{P_{i^*}}^n,U_{P_i}^n) \big|U_{P_{i^*}}^n \cap \Dscr(g_n,\shat_n)\big| \non \\
	&= \Delta(U_{P_{i^*}}^n,U_{P_i}^n) \beta_{i^*}|U_{P_{i^*}}^n|. \non
\end{align}
Using Lemma~\ref{lem:bireg_subgraph_undir}, we get that $|V(U_{P_i}^n,U_{P_{i^*}}^n \cap \Dscr(g_n,\shat_n))| \geq  \beta_{i^*}|U_{P_i}^n| $ and hence at least $ \beta_{i^*}$ fraction of sequences from $ U_{P_i}^n $ have edges with sequences in $ U_{P_{i^*}}^n \cap \Dscr(g_n,\shat_n) $. Since $ \Dscr(g_n,\shat_n)$ is an independent set, we have that $ |U_{P_i}^n \cap  \Dscr(g_n,\shat_n)| \leq (1-\beta_{i^*}) |U_{P_i}^n| $. This gives that $ \beta_i \leq (1-\beta_{i^*}) < 1/2 $ and hence   % for all $ i \neq i^* $ which is a contradiction. Thus, we have that for all $ i \neq i^* $, $ \beta_i \leq 1/2 $.
% This gives that
%\begin{align}
$	\left| \bigcup_{P_i \in \Pscr_1} U_{P_i}^n \cap \Dscr(g_n,\shat_n)\right| < \beta  \left| \bigcup_{P_i \in \Pscr_1} U_{P_i}^n \right|$, % \non
%\end{align}
where $ \beta  $ subsumes the fractions $ \beta_i $. Since $ \beta_i \leq 1/2 $ for all $ i \neq i^* $,we have that $ \beta  < 1 $.

If $ \beta_{i^*} \leq 1/2 $, the claim is trivially true.
\end{proof}

Recall that we had fixed an arbitrary set $ \Pscr_1 $. Similar expressions hold for all sets of types $ \Pscr_i \in \Pscr_n(\Typepx) $. This implies that $ \left| \Typepx  \cap \Dscr(g_n,\shat_n)\right| < \alpha  \left| \Typepx \right| $ with $ \alpha < 1 $.   Thus, for arbitrary $ \shat_n \in \pbest(g_n)$ we get that
%\begin{align}
$\Pbb( \Dscr(g_n,\shat_n)) \leq \alpha \Pbb(\Typepx) + \Pbb(\Xscr^n \setminus \Typepx$) % \non
%\end{align}
and hence \newline $ \lim_{n \rarr \infty} \min_{s_n \in \best(g_n) } \Pbb( \Dscr(g_n,s_n)) \leq \alpha < 1 $. Finally, this gives
\begin{align}
\lim_{n \rarr \infty} \max_{s_n \in \best(g_n)} \Pbb( g_n \circ s_n(\Xvec) \neq \Xvec)  = 1 - \lim_{n \rarr \infty} \min_{s_n \in \best(g_n)}  \Pbb(\Dscr(g_n,s_n)) \geq 1- \alpha > 0. \non
\end{align}
Thus, asymptotically, the worst-case probability of error stays strictly bounded away from zero for all sequences of strategies $ \{g_n\}_{n \geq 1} $. This completes the proof of Theorem~\ref{thm:lossl_delta=0_empty_R}.
\end{proof}

\subsection{Proof of Theorem~\ref{thm:ach-rate-convex}}
\label{appen:ach-rate-convex}

First, we prove the following lemma that shows that for any $ \bd \in [0,1] $, the receiver can restrict its strategies to a particular class of functions.

%{\color{red} move to appendix}

\begin{lemma} \label{lem:simplified_gn}
	Let $ \bd \in [0,1] $ and $n \in \Nbb$. The strategies of the receiver can be restricted as
	\begin{align}
		g_n(\xvec) = \left\{
		\begin{array}{c l}
			\xvec & \mbox{ if } \xvec \in I^n \\
			\xvec_0 &  \mbox{ else }
		\end{array},\right. \label{eq:general_form_g_n}
	\end{align}
	where $ I^n \subseteq \Xscr^n $ and $\xvec_0 \in I^n $.
\end{lemma}
\begin{proof}
	Let $ \gbar_n$ be any strategy and for any strategy $ \sbar_n \in \best(\gbar_n) $ and let $ \image( \gbar_n \circ \sbar_n)$ be the image of the composite function $ \gbar_n \circ \sbar_n $. Writing explicitly, the image is given as
	\begin{align}
		\image(\gbar_n \circ \sbar_n) = \{ \xvec \in \image(\gbar_n) : \gbar_n \circ \sbar_n(\xvec') = \xvec, \xvec' \in \Xscr^n  \}. \non
	\end{align}
	Thus, the image consists of the sequences or \textit{outcomes} that are achieved by the strategy $ \gbar_n $. Revelation principle \cite{myerson1997game} states that the same outcomes can be achieved by a ``direct'' strategy defined as
	\begin{align}
		g_n(\xvec) = \left\{
		\begin{array}{c l}
			\xvec & \mbox{ if } \xvec \in \image(\gbar_n) \\
			\xvec_0 & \mbox{ else }
		\end{array},\right. \non
	\end{align}
	where $ \xvec_0 \in \image(\gbar_n) $.
	To see this, observe that for any $ \wi{\xvec} \in \Xscr^n $  and $ \wi{\yvec} \in \image(\gbar_n) $, if $ \gbar_n \circ \sbar_n(\wi{\xvec}) = \wi{\yvec} $, then in response to the strategy $ g_n $, the sender can achieve the same outcome for $ \wi{\xvec} $ by choosing a strategy $ s_n $  such that $ s_n(\wi{\xvec}) = \wi{\yvec} $. This implies that there exists a strategy $ s_n \in \best(g_n) $ such that $g_n \circ s_n(\xvec) = \gbar_n \circ \sbar_n(\xvec) $ for all $ \xvec \in \Xscr^n $. This holds for all strategies  $ \sbar_n \in \best(\gbar_n)$ which gives that
	%	  \begin{align}
		$  \max_{s_n \in \best(g_n)} \Escr_{\bd}(g_n,s_n) =	\max_{\sbar_n \in \best(\gbar_n)} \Escr_{\bd}(\gbar_n,\sbar_n) $ % \non
		%	 \end{align}
	and the worst-case probability of error is the same with the restriction on $ g_n $.
\end{proof}

This lemma shows that for constructing  a sequence of  strategies $ \{g_n\}_{n \geq 1} $ achieving \eqref{eq:receiver_objective}, it is sufficient for the receiver to choose a sequence of sets $ \{I^n\}_{n \geq 1} $ such that  $ \max_{s_n \in \best(g_n)} \Escr_{\bd + \delta}(g_n,s_n) $ is arbitrarily small for every $ \delta > 0 $.

Recall that $ \best(g_n) $, given in Definition~\ref{defn:best_resp}, is the set of utility maximizing strategies of the sender in response to the strategy $ g_n $ of the receiver. The definition of the set of worst-case best response strategies of the sender is defined in \eqref{eq:worst_case_br} as
%\begin{align}
$ \pbest_{\bd}(g_n) = \argmax_{s_n \in \best(g_n)}\; \Escr_{\bd}(g_n,s_n)$. % \non % \label{eq:worst_case_br}
%\end{align}
\subsubsection{Proof of convexity of rate region}
%\begin{proof}[\textit{of Theorem~\ref{thm:ach-rate-convex}}]
 Let $ \epsilon, \delta > 0 $.  Consider two achievable rates $ \Rhat,\Rtilde \in \Rscr_{\bd} $. Thus, we have sequences  $ \widehat{\epsilon}_n, \widetilde{\epsilon}_n \rarr 0 $,  $ \widehat{\delta}_n, \widetilde{\delta}_n \rarr 0 $ and two sequences of strategies $\{\ghat_n\}_{n \geq 1}$ and $\{ \gtilde_n\}_{n \geq 1}$ such that for all $ n $ %  Further, let  $ \shat_n^* \in \pbest(\ghat_n)  $, $ \stilde_n^* \in \pbest(\gtilde_n) $   and sequences be such that
  \begin{align}
\max_{\shat_n \in \best(\ghat_n)}\Escr_{\bd+ \widehat{\delta}_n}(\ghat_n,\shat_n) \leq \widehat{\epsilon}_n,  \quad\max_{\stilde_n \in \best(\gtilde_n)}\Escr_{\bd+\widetilde{\delta}_n}(\gtilde_n,\stilde_n) \leq \widetilde{\epsilon}_n. \non
  \end{align}
From Lemma~\ref{lem:simplified_gn}, without loss of generality, assume that $ \ghat_n $ and $ \gtilde_n $ are of the form  given in \eqref{eq:general_form_g_n}.
 Further, let  $ \shat_n^* \in \pbest_{\bd+\widehat{\delta}_n}(\ghat_n)  $, $ \stilde_n^* \in \pbest_{\bd+\widetilde{\delta}_n}(\gtilde_n) $ be such that
  \begin{align}
    \lim_{ n \rarr \infty} R_{\bd+\widehat{\delta}_n}(\ghat_n,\shat_n^*) &= \Rhat,  \quad    \lim_{ n \rarr \infty}  R_{\bd+\widetilde{\delta}_n}(\gtilde_n,\stilde_n^*) = \Rtilde. \non
  \end{align}
%  Let $ n $ be large enough such that $ \widehat{\epsilon}_n + \widetilde{\epsilon}_n < \epsilon $ and  $ \widehat{\delta}_n + \widetilde{\delta}_n < \delta $.
  % Let $ \shat_n^*$ and $ \stilde_n^*$ be the strategies that achieve the maximum in the above equations. \tododone{Was a typo}
%  Let $n \in  \Nbb$ and define $\Ibar^n, \Itilde^n \subseteq \Xscr^n$ as
%   \begin{align}
% \Ibar^n = \Big\{ \xvec \in \image(\ghat_n) : \exists \;\bar{\xvec} \;s.t.\; d_n(\bar{\xvec},\xvec) \leq \bd,  \ghat_n \circ \shat_n(\bar{\xvec}) = \xvec \Big\}. \non
%   \end{align}
%   Thus, $R_{\bd}(\ghat_n,\shat_n) = \log |\Ibar^n|/n$. We define $\Itilde^n$ similarly.

%   Let $\alpha \in [0,1]$ and define a set
%   \begin{align}
% D^n = (\ghat_{n_\alpha})^{\inv}(\Ibar^{n_\alpha}) \times (\gtilde_{n-n_\alpha})^{\inv}(\Itilde^{n-n_\alpha}),
%   \end{align}
%   where $ n_\alpha := \floor{\alpha n}$ and $(\ghat_{n_\alpha})^{\inv}(\Ibar^{n_\alpha})$ is the inverse image of $\Itilde^{n_\alpha}$ under $\ghat_{n_\alpha}$ and is given as
%   \begin{align}
%     (\ghat_{n_\alpha})^{\inv}(\Ibar^{n_\alpha}) = \{\xvec \in \Xscr^{n_\alpha} : g_{n_\alpha}(\xvec) \in \Ibar^{n_\alpha}\}. \non
%   \end{align}

We use the time-sharing arguments from information theory to show that any convex combination of the rates $ \Rhat $ and $ \Rtilde $ is also achievable. For $\alpha \in [0,1]$,  let $ n_\alpha := \floor{\alpha n}$. Let $ n $ be large enough such that $ \widehat{\epsilon}_{n_\alpha} + \widetilde{\epsilon}_{n-n_\alpha} < \epsilon $ and  $ \widehat{\delta}_{n_\alpha} + \widetilde{\delta}_{n-n_\alpha} < \delta $. Define %a set
  \begin{align}
    \Zscr^n &= \Dscr_{\bd+\widehat{\delta}_{n_\alpha}}(\ghat_{n_\alpha}, \shat_{n_\alpha}^*) \times \Dscr_{\bd+\widetilde{\delta}_{n-n_\alpha}}(\gtilde_{n-n_\alpha}, \stilde_{n-n_\alpha}^*) \non \\
     \Ascr^n &= \Ascr_{\bd+\widehat{\delta}_{n_\alpha}}^{n_\alpha}(\ghat_{n_\alpha}, \shat_{n_\alpha}^*) \times \Ascr_{\bd+\widetilde{\delta}_{n-n_\alpha}}^{n-n_\alpha}(\gtilde_{n-n_\alpha}, \stilde_{n-n_\alpha}^*) \non
              %\Big\{\xvec \in \Xscr^{n_\alpha} : g_{n_\alpha}(\xvec) \in \Ibar^{n_\alpha}\Big\} \non \\
%    &\times \Big\{\xvec \in \Xscr^{n-n_\alpha} : g_{n-n_\alpha}(\xvec) \in \Itilde^{n-n_\alpha}\Big\},
  \end{align}
 % From Lemma~\ref{lem:simplified_gn}, without loss of generality, for all $ \ybar \in \Ascr_{\bd+\widehat{\delta}_{n_\alpha}}^{n_\alpha}(\ghat_{n_\alpha}, \shat_{n_\alpha}^*) $, we have $ \ghat_{n_\alpha}(\ybar) = \ybar $. We have similar property for  sequences in $  \Ascr_{\bd+\widetilde{\delta}_{n-n_\alpha}}^{n-n_\alpha}(\gtilde_{n-n_\alpha}, \stilde_{n-n_\alpha}^*) $ as well. \textbf{Why?}

%{\color{red}required?}
%Here, the set $ \Zscr^n $ contains all the sequences constructed by taking the Cartesian product of the sets of correctly recovered sequences by the  pairs of strategies $ \{\ghat_{n_\alpha}, \shat_{n_\alpha}^*\} $ and $ \{\gtilde_{n-n_\alpha}, \stilde_{n-n_\alpha}^*\} $. The set $ \Ascr^n $ contains all the sequences constructed taking the Cartesian product of the sets of \textit{utilized} sequences by the respective pair of strategies.

Let $ \yvec = (\widehat{\yvec},\widetilde{\yvec}) \in \Ascr^n $ with $\widehat{\yvec} \in \Ascr_{\bd+\widehat{\delta}_{n_\alpha}}^{n_\alpha}(\ghat_{n_\alpha}, \shat_{n_\alpha}^*)$ and $ \widetilde{\yvec} \in \Ascr_{\bd+\widetilde{\delta}_{n-n_\alpha}}^{n-n_\alpha}(\gtilde_{n-n_\alpha}, \stilde_{n-n_\alpha}^*) $. Fix $ \yvec_0 \in \Ascr^n $ and define a strategy $g_n$ for the receiver as %\todo{not sure I agree with this construction. Why is there a composition with $ s $?}
  \begin{align}
    g_n(\yvec) = \left\{
    \begin{array}{c l}
      ( \widehat{\yvec},\widetilde{\yvec})   &  \mbox{ if }\;\yvec = (\widehat{\yvec},\widetilde{\yvec}) \in \Ascr^n \\
      \yvec_0  &   \mbox{ if } \;\yvec \notin \Ascr^n
    \end{array}
\right..  \non
  \end{align}
  From the structure of the receiver's strategy and due to the additive nature of the utility, it is clear that the sender also responds with a strategy that maps the $ n_\alpha $ and $ n-n_\alpha $ components of $ \xvec $ separately.  Specifically,  for  $ \xvec = (\widehat{\xvec},\widetilde{\xvec}) \in \Xscr^n $, we have % {\color{red}remove y*, and line in eq} let $ \widehat{\yvec}^* = \shat_{n_\alpha}^*(\widehat{\xvec}) $ and $ \widetilde{\yvec}^* = \stilde_{n-n_\alpha}^*(\widetilde{\xvec}) $. Using this, we have
  \begin{align}
    \ut_n(g_n \circ s_n(\xvec),\xvec) &= \alpha \ut_{n_\alpha}(\widehat{\yvec},\widehat{\xvec}) + (1-\alpha) \ut_{n-n_\alpha}(\widetilde{\yvec},\widetilde{\xvec})   \non \\
                           % &\leq \alpha \ut_{n_\alpha}(\widehat{\yvec}^*,\widehat{\xvec}) + (1-\alpha)\ut_{n-n_\alpha}(\widetilde{\yvec}^*,\widetilde{\xvec}) \non \\
    &\leq \alpha \ut_{n_\alpha}(\ghat_{n_\alpha} \circ \shat_{n_\alpha}^*(\widehat{\xvec}),\widehat{\xvec}) + (1-\alpha)\ut_{n-n_\alpha}(  \gtilde_{n-n_\alpha} \circ \stilde_{n-n_\alpha}^*(\widetilde{\xvec}),\widetilde{\xvec} ) \non.
  \end{align}
  This holds for all  $ \xvec \in \Xscr^n $ and hence the strategy $ s_n^*(\xvec) =  (\shat_{n_\alpha}^*(\widehat{\xvec}), \stilde_{n-n_\alpha}^*(\widetilde{\xvec})) $ is a best response strategy of the sender for the strategy $ g_n $.  It follows that for all $ \xvec = (\widehat{\xvec},\widetilde{\xvec}) \in \Zscr^n $,%=  \Dscr_{\bd}(\ghat_{n_\alpha}, \shat_{n_\alpha}^*)$ % \times \Dscr_{\bd}(\gtilde_{n_\alpha}, \stilde_{n_\alpha}^*)$, we have $\widehat{\xvec} \in \Dscr_{\bd}(\ghat_{n_\alpha},\shat_{n_\alpha}^*)$ and $\widetilde{\xvec} \in \Dscr_{\bd}(\gtilde_{n-n_\alpha},\stilde_{n-n_\alpha}^*)$. Thus,
\begin{align}
d_n(g_n \circ s_n^*(\xvec),\xvec) %&= \Big(\alpha \; d_{n_\alpha}(\ghat_{n_\alpha} \circ \shat_{n_\alpha}(\widehat{\xvec}),\widehat{\xvec}) \non \\
%&+ (1-\alpha)d_{n-n_\alpha}(\gtilde_{n-n_\alpha} \circ \stilde_{n-n_\alpha} (\widetilde{\xvec}),\widetilde{\xvec}) \Big) \non  \\
&\leq \alpha (\bd+\widehat{\delta}_{n_\alpha}) + (1-\alpha) (\bd+\widetilde{\delta}_{n-n_\alpha}) \leq \bd + \delta. \non
\end{align}
Furthermore, $ s_n^* \in \pbest_{\bd+\delta}(g_n)$ which gives that $ \Dscr_{\bd + \delta} (g_n,s_n^*) \supseteq \Zscr^n $  %\Dscr_{\bd + \widehat{\delta}_n}(\ghat_{n_\alpha},\shat_{n_\alpha})  \times \Dscr_{\bd + \widetilde{\delta}_n}(\gtilde_{n-n_\alpha},\stilde_{n-n_\alpha}) $.
and hence, the probability of error is bounded as
  \begin{align}
 \Escr_{\bd + \delta}(g_n,s_n^*)
        %= \sum_{\xvec \in \Zscr^n} P_{X}(\xvec)\Pbb( d_n(g_n \circ s_n^*(X),\xvec) \leq \bd ) \non \\
    %                           &\leq  \sum_{\widehat{\xvec} \in \Ibar^{n_\alpha}} P_{\widehat{\xvec}}(\widehat{\xvec})\Pbb( d_{n_\alpha}(\ghat_{n_\alpha} \circ \shat_{n_\alpha}(\widehat{\xvec}),\widehat{\xvec}) \leq \bd ) \non \\
    %                           &+  \sum_{\widetilde{\xvec} \in \Itilde^{n-n_\alpha}} P_{\widetilde{\xvec}}(\widetilde{\xvec}) \Pbb( d_{n-n_\alpha}(\gtilde_{n-n_\alpha} \circ \stilde_{n-n_\alpha}(\widetilde{\xvec}),\widetilde{\xvec}) \leq \bd ) \non \\
    &\leq \Escr_{\bd + \widehat{\delta}_{n_\alpha}}(\ghat_{n_\alpha},\shat_{n_\alpha}^*) + \Escr_{\bd + \widetilde{\delta}_{n-n_\alpha}}(\gtilde_{n-n_\alpha},\stilde_{n-n_\alpha}^*) < \epsilon. \non
  \end{align}
 % which gives $\lim_{n \rarr \infty}     \Escr_{\bd}(g_n,s_n^*) = 0$.
  Thus, $\{g_n\}_{n \geq 1}$ is a sequence of strategies achieving vanishing probability of error with $\{s_n^*\}_{n \geq 1}$ as the corresponding sequence of best responses of the sender.

  We now compute the rate of this sequence.
 % Recall the definition of $ \Ascr^n(g_n,s_n) $ from \eqref{eq:defn-Ascr_d}.
%  For the strategy $ g_n $, we have
%  \begin{align}
%\image(g_n) = \Ascr_{\bd + \widehat{\delta}_{n_\alpha}}^{n_\alpha}(\ghat_{n_\alpha},\shat_{n_\alpha}^*) \times \Ascr_{\bd + \widetilde{\delta}_{n-n_\alpha}}^{n-n_\alpha}(\gtilde_{n-n_\alpha},\stilde_{n-n_\alpha}^*).\non
%  \end{align}
Since $ s_n^* \equiv (\shat_{n_\alpha}^*, \stilde_{n-n_\alpha}^*)$,  from \eqref{eq:defn-Ascr_d} we have%  Thus, for all $ \xvec \in \image(g_n) $, there exists an $ \widehat{\xvec} \in \Dscr_{\bd}(g_n, s_n^*) $ for which $ g_n \circ s_n^*(\widehat{\xvec}) = \xvec $. %  For any pair of strategies $ (g_n,s_n) $,
%Thus,
  \begin{align}
    \Ascr_{\bd+\delta}^n(g_n,s_n^*) &= \Big\{ \xvec \in \image(g_n) : g_n \circ s_n^*(\xvec') = \xvec %\non \\
    % &\hspace{2cm}
     \mbox{ for some }\;\xvec' \in \Dscr_{\bd + \delta}(g_n, s_n^*) \Big\} \non \\
                       &= \Ascr_{\bd + \widehat{\delta}_{n_\alpha}}^{n_\alpha}(\ghat_{n_\alpha},\shat_{n_\alpha}^*) \times \Ascr_{\bd + \widetilde{\delta}_{n-n_\alpha}}^{n-n_\alpha}(\gtilde_{n-n_\alpha},\stilde_{n-n_\alpha}^*). \non
  \end{align}
  % By definition of the strategy $g_n$ in \eqref{eq:defn-gn}, we have that for all $\xvec = (\widehat{\xvec},\widetilde{\xvec}) \in \Zscr^n$, $\ghat_{n_\alpha}(\bar{\xvec}) \in \Ibar^{n_\alpha}$ and $\gtilde_{n-n_\alpha}(\widetilde{\xvec}) \in \Itilde^{n-n_\alpha}$.Thus, $\image(g_n) = \Ibar^{n_\alpha} \times \Itilde^{n-n_\alpha}$.
  This gives %$ R_{\bd + \delta}(g_n,s_n^*) = \alpha R_{\bd + \widehat{\delta}_n}(\ghat_{n_\alpha},\shat_{n_\alpha}) + (1-\alpha) R_{\bd + \widetilde{\delta}_n}(\gtilde_{n-n_\alpha},\stilde_{n-n_\alpha})$.
  \begin{align}
    R_{\bd + \delta}(g_n,s_n^*) \non%&= \frac{1}{n}\log |    \Ascr^n(g_n,s_n^*)| \non \\
                %&= \frac{1}{n} \big(\log |\Ascr^{n_\alpha}(\ghat_{n_\alpha},\shat_{n_\alpha})| + \log |\Ascr^{n-n_\alpha}(\gtilde_{n-n_\alpha},\stilde_{n-n_\alpha})|\big) \label{eq:rate=I1+I2} \\
    &= \alpha R_{\bd + \widehat{\delta}_{n_\alpha}}(\ghat_{n_\alpha},\shat_{n_\alpha}^*) + (1-\alpha) R_{\bd + \widetilde{\delta}_{{n-n_\alpha}}}(\gtilde_{n-n_\alpha},\stilde_{n-n_\alpha}^*). \non
  \end{align}
  Taking the limit, we get $  \lim_{n \rarr \infty} R_{\bd + \delta}(g_n,s_n^*)  = \alpha \Rhat + (1-\alpha) \Rtilde $.
  % \begin{align}
  % \lim_{n \rarr \infty} R_{\bd + \delta}(g_n,s_n^*)
  %   &= \alpha \Rhat + (1-\alpha) \Rtilde. \non
  % \end{align}
This holds for all $\alpha \in [0,1]$ and arbitrarily small $ \epsilon, \delta > 0 $. Thus, the achievable rate region is convex.

%\subsection{Proof of Corollary~\ref{lem:Rdmax_Rdmin}}
%\label{appen:Rdmax_Rdmin}%\begin{proof}[\textit{of Corollary~\ref{lem:Rdmax_Rdmin}}]

  	\subsubsection{Proof of convexity of $\Rscr_{\bd}^{\inf}$ and concavity of $\Rscr_{\bd}^{\sup}$ }
  	Consider $ \bd_1, \bd_2 \in \Real^* $ and let $  \alpha \in [0,1] $. Let $ \Rscr_{\bd_1}^{\inf} $ and $ \Rscr_{\bd_2}^{\inf} $ be the respective infimum of the achievable rates for the distortion levels $ \bd_1 $ and $ \bd_2 $ respectively. Thus, given $ \epsilon_1, \epsilon_2 > 0 $, there exists a sequence of  strategies %\ach
  	  achieving the rate $ \Rscr_{\bd_1}^{\inf}  + \epsilon_1 $ and $ \Rscr_{\bd_2}^{\inf}  + \epsilon_2 $. Using a construction similar to the proof of Theorem~\ref{thm:ach-rate-convex}, we can show that the rate $ R = \alpha (\Rscr_{\bd_1}^{\inf}  + \epsilon_1 )+ (1-\alpha)(\Rscr_{\bd_2}^{\inf}  + \epsilon_2) $ is achievable for the distortion level $ \bd = \alpha \bd_1 + (1-\alpha)\bd_2 $. Since this holds for arbitrarily small $\epsilon_1, \epsilon_2 $, it follows that $ \Rscr_{ \bd}^{\inf} \leq  \alpha \Rscr_{\bd_1}^{\inf} + (1-\alpha)\Rscr_{\bd_2}^{\inf}  $. The proof of concavity of $ \Rscr_{ \bd}^{\sup} $ follows similarly by taking a sequence of  strategies achieving the rate $ \Rscr_{\bd_1}^{\sup}  - \epsilon_1 $ and $ \Rscr_{\bd_2}^{\sup}  - \epsilon_2 $.

       \section{ Proof of Lossless Recovery : Binary Alphabet }
       \label{sec:bin_alpha_lossl}

For the case of binary alphabet, for any type $ \Phat \in \types $, we denote $ \Phat(0)$  as $  \phat $. Similarly, for any sequence $ \xvec $, we denote $  P_{\xvec}(0) $ as $ p_{\xvec} $.
       For $ \epsilon > 0 $, recall  the typical set $\Typep$ and for a distribution $P \in \Pscr_n(\Xscr)$, recall the type class  $U_{P}^n$.%. \subseteq \Xscr^n$ defined  in Section~\ref{append:lossl_delta=0_empty_R} after \eqref{eq:TKq-set-defn}.% \eqref{eq:type-class}.

       We prove the result in two parts.  First, we prove the necessity of the condition \eqref{eq:u01_u10_cond}, that is $ \ut(0,1) + \ut(1,0) < 0 $, for the rate region to be non-empty. This is done in Appendix~\ref{appen:proof_necc_cond_lossl}. Then, we determine the rate region when \eqref{eq:u01_u10_cond} holds in Appendix~\ref{appen:suff_cond}.
%\inserttodo{Recall what (10) is here} \tododone{edited}

%This theorem can only say that there is an achievable rate at most $ R $ that achieves $ \Dbar(R)$. This is because the rate in our setting is

              \subsection{Proof of necessity of \eqref{eq:u01_u10_cond} and  Theorem~\ref{thm:lossl_rate_charac} a)}%Proof of Theorem~\ref{thm:necc_u01_u10_cond}}
       \label{appen:proof_necc_cond_lossl}

%{\color{red} can merge with general alphabet for $ \ut(0,1)+\ut(1,0) > 0 $? extra arguments for $ = 0 $ }
 %        It has to be that for all $ \xvec \in D_p $, there exists a corresp. $ \bar{\xvec} \in \Cscr^n $ such that $ d_n(\xvec,\bar{\xvec}) \leq \delta_n $.
  We prove the claim by showing that if  %\eqref{eq:u01_u10_cond} does not hold, \ie,
 $ \ut(0,1) + \ut(1,0) \geq 0 $, then for any sequence of strategies of the receiver, only a fraction of sequences from $ \Xscr^n $ can be recovered correctly.

       Let $ \delta > 0 $ and without loss of generality, let $ \ut(1,0) = a $ and $ \ut(0,1) = - b $ and $ a \geq b \geq 0 $.  Using Lemma~\ref{lem:simplified_gn},  let $ I^n \subseteq \Xscr^n $ and  $ \xvec_0 \in I^n $ define a strategy  $ g_n $ of the receiver as %and defining % define a strategy $ g_n $ as
        \begin{align}
          g_n(\xvec) = \left\{
    \begin{array}{c l}
      \xvec & \mbox{ if } \xvec \in I^n \\
      \xvec_0 & \mbox{ if } \xvec \notin I^n
    \end{array}.\right. \non
        \end{align}
        Let $ s_n \in \pbest_{\delta}(g_n) $ and consider the set of sequences  recovered within distortion $ \delta $, \ie, $ \Dscr_\delta(g_n,s_n) $.   We show that only a $2\delta$-ball of sequences from any type class can be in $ \Dscr_\delta(g_n,s_n) $. We use this fact to show that the probability of error for any sequence of strategies of the receiver will tend to one as $ n $ grows large. %We prove this by considering two cases.
        % We prove the necessity of  \eqref{eq:u01_u10_cond} by considering two cases. %In both the cases we show that to recover any sequence correctly, then only a fraction of the sequences of the same type can be in the image.

 Fix a sequence $ \widehat{\xvec} \in \Dscr_\delta(g_n,s_n) $  and let $ p_{\widehat{\xvec}} = P_{\widehat{\xvec}}(0) $.  From the structure of $ g_n $, we can assume that  that $ s_n(\bar{\xvec}) \in I^n $ for all $ \bar{\xvec} \in \Xscr^n $. Let the best response sequence of $ \widehat{\xvec} $ be $ \widehat{\yvec} \in I^n $, \ie,  $ s_n(\widehat{\xvec}) = \widehat{\yvec} $. It follows that $ d_n(\widehat{\yvec},\widehat{\xvec}) \leq \delta $ and $ \ut_n(\widehat{\yvec},\widehat{\xvec}) = \max_{\yvec \in I^n} \ut_n(\yvec,\widehat{\xvec}) $. Furthermore, since $ s_n $ is among the worst-case strategies for the receiver, we also have that  % for all $ \yvec \in I^n $ such that $ \ut_n(\yvec,\widehat{\xvec}) = \ut_n(\widehat{\yvec},\widehat{\xvec}) $, it must be that $ d_n(\yvec,\widehat{\xvec}) \leq \delta $.
\begin{align}
	\mbox{ if }  d_n(\yvec,\widehat{\xvec}) > \delta  \mbox{ for some } \yvec \in I^n, \mbox{ then } \ut_n(\yvec,\widehat{\xvec}) < \ut_n(\widehat{\yvec},\widehat{\xvec}). \label{eq:worst_case_y_y'}
\end{align}
\begin{claim}
  Let $\xvec \in \Dscr_\delta(g_n,s_n)$ be a sequence  having the same type as $\widehat{\xvec}$. Then $ d_n(\widehat{\xvec}, \xvec) \leq 2\delta $.% are at most $2\delta$ apart.
\end{claim}
\begin{proof}

%  \removetodo{        We show that for all sequences $ \xvec \in \Dscr_\delta(g_n,s_n) \cap U_{p_{\widehat{\xvec}}}^n $, we have $ d_n(\xvec,\widehat{\xvec}) \leq 2\delta $. Thus, only a $2\delta$-ball of sequences from $ U_{p_{\widehat{\xvec}}}^n $ can be in $ \Dscr_\delta(g_n,s_n) $.}{not needed}
  We prove the claim by contradiction.
Let $ U_{p_{\widehat{\xvec}}}^n $ be the type class corresponding to the type $ p_{\widehat{\xvec}} $.  Suppose there exists a sequence $ \xvec' \in  \Dscr_\delta(g_n,s_n) \cap U_{p_{\widehat{\xvec}}}^n $ such that $ d_n(\xvec',\widehat{\xvec}) > 2\delta $. Let $ \yvec' \in I^n $ be the best response sequence of $ \xvec' $ such that $ s_n(\xvec') = \yvec' $, $ d_n(\yvec',\xvec') \leq \delta $ and $ \ut_n(\yvec',\xvec') = \max_{\yvec \in I^n} \ut_n(\yvec,\xvec') $.    %it holds that
%  \begin{align}
%    \mbox{ if } \ut_n(\yvec,\xvec') = \ut_n(\yvec',\xvec') \mbox{ for some } \yvec \in I^n, \mbox{ then } d_n(\yvec,\xvec') \leq \delta.
%  \end{align}
  %Note that $ \widehat{\yvec} $ is a best response sequence to $ \widehat{\xvec} $ and $ \yvec' $ is a best response sequence to $ \xvec' $.  As earlier,
  We write $ P_{\yvec'}(0) $  as $ p_{\yvec'} $ and $ P_{\widehat{\yvec}}(0) $  as $ p_{\widehat{\yvec}} $. Observe that since $ d_n(\xvec',\widehat{\xvec}) > 2\delta $ and $ d_n(\yvec',\xvec')  \leq \delta $, we have $ d_n(\yvec',\widehat{\xvec})  > \delta $.
We can write % $P_{\widehat{\yvec},\widehat{\xvec}}(0,1)$ and $P_{\widehat{\yvec},\widehat{\xvec}}(1,0)$ as
\begin{align}
	 P_{\widehat{\yvec},\widehat{\xvec}}(0,1)  &=  p_{\widehat{\yvec}}  - p_{\widehat{\xvec}}  +   P_{\widehat{\yvec},\widehat{\xvec}}(1,0) \non \\
	P_{\widehat{\yvec},\widehat{\xvec}}(1,0)  &= \frac{d_n(\widehat{\yvec},\widehat{\xvec}) +p_{\widehat{\xvec}}  - p_{\widehat{\yvec}} }{2}. \non
\end{align}
Using the above in $ \ut_n(\widehat{\yvec},\widehat{\xvec}) $ and rearranging terms, we get
\begin{align}
	\ut_n(\widehat{\yvec},\widehat{\xvec})  = -(p_{\widehat{\yvec}}  - p_{\widehat{\xvec}} ) b+ \frac{d_n(\widehat{\yvec},\widehat{\xvec}) +p_{\widehat{\xvec}}  - p_{\widehat{\yvec}} }{2}(a-b). \label{eq:u_yx_in_terms_dyx}
\end{align}
Since $ d_n(\yvec',\widehat{\xvec}) > \delta $, from \eqref{eq:worst_case_y_y'} we get $ \ut_n(\widehat{\yvec},\widehat{\xvec}) > \ut_n(\yvec',\widehat{\xvec}) $.  Using \eqref{eq:u_yx_in_terms_dyx} and the fact that $ p_{\widehat{\xvec}} = p_{\xvec'}$, we have
\begin{align}
	(p_{\widehat{\yvec}} - p_{\yvec'})a +  \frac{d_n(\yvec',\widehat{\xvec}) - d_n(\widehat{\yvec},\widehat{\xvec}) }{2}(a-b) < 0. \non
\end{align}
Since $d_n(\yvec',\widehat{\xvec}) >  d_n(\widehat{\yvec},\widehat{\xvec})$, it must be that $ p_{\widehat{\yvec}} < p_{\yvec'}$.

However,  for the sequence $ \xvec' $ too, \eqref{eq:worst_case_y_y'} holds. Thereby, repeating the above arguments with  $ \xvec'$ in consideration, we get that $ p_{\widehat{\yvec}} > p_{\yvec'}$. Since both cannot be simultaneously true, it must be that $ d_n(\widehat{\xvec},\xvec') \leq 2\delta$.
  In other words, if multiple sequences from  $ U_{p_{\widehat{\xvec}}}^n $ are recovered correctly, then they must be within $ 2\delta$ distance. Since $ \widehat{\xvec} $ was arbitrary, this holds for all types $ \phat \in \Pscr_n(\Xscr) $.
\end{proof}
       Suppose $ \xvec^1,\hdots, \xvec^L \in \Dscr_\delta(g_n,s_n) $ are sequences having distinct types that are recovered within distortion $ \delta $. Then,
  %     \begin{align}
  $   \Dscr_\delta(g_n,s_n) \subseteq \bigcup_{i \in [L]} B_{2\delta}(\xvec^i) \cap U_{P_{\xvec^i}}^n$. %\non
%       \end{align}
Clearly, for small enough $ \delta $, we have that $ \min_{s_n \in \best(g_n)}\Pbb(\Dscr_\delta(g_n,s_n)) \rarr 0 $ as $ n \rarr \infty $. This holds for any sequence of strategies of the receiver and   hence  %there does not exist a sequence of strategies achieving vanishing error. %\ach
  %Thus,
  the achievable rate region is empty. This completes the proof.% of necessity of \eqref{eq:u01_u10_cond} and  Theorem~\ref{thm:lossl_rate_charac} a).

  \subsection{Proof of sufficiency of \eqref{eq:u01_u10_cond} and characterization of $ \Rscr$} %Proo\subsection{Proof of Theorem~\ref{thm:suff_u01_u10_cond}}
              \label{appen:suff_cond}

                          %\subsection{Using a rate-distortion result from mismatched setting}
%\todo{This can be introduced where it is actually used}
      % The objective was to devise a codebook such that the receiver achieves a certain distortion level.

   Before proving the sufficiency, we prove the following result. Lapidoth in \cite{lapidoth1997role} studied a setting where the encoding and decoding fidelity criteria of the sender and receiver were mismatched.  We state the following theorem which is a version of Theorem~$1$ from \cite{lapidoth1997role}. We use this theorem to determine achievable rates for our setting. %Before that, we briefly discuss the setting of the mismatched decoding studied in \cite{lapidoth1997role}.

% For a distribution $ P $ and $ \delta > 0 $,  define
% \begin{align}
%   N(P,\delta) = \{ P' : |P'-P|_\infty \leq \delta \}.
% \end{align}

   %\todo{Check new wording:  Let $ R $  be fixed and let $ \Dbar(R) $ be defined as... Then, there exists a sequence of strategies $ \{g_n\}_{n \geq 1}$, an $ s_n \in \pbest(g_n) $ for all $ n $,  and $ \delta_n \rarr 0 $ such that $  \lim_{n \rarr \infty} R_{\bd+\delta_n}(g_n,s_n) = R_0 $, where $ R_0 \leq R $ and $ \bd = \Dbar(R) $ }

\begin{theorem} \label{thm:dist_rate_bound}
  Let $ R $  be fixed and let $ \Dbar(R) $ be defined as
  \begin{align}
    \Dbar(R) = \min_{\Pcode} \max_{\Pbar \in \Fscr} \;\Ebb_{\Pbar}[d(\Xhat,X)], \non
  \end{align}
  where $ \Fscr $ is given as
  \begin{align}
    \Fscr &= \argmax_{\Ptilde_{\Xhat,X} \in \Wscr(\prbX,\Pcode,R) } \;\Ebb_{\Ptilde_{\Xhat,X}}[\ut(\Xhat,X)], \non \\
    \Wscr(\prbX,\Pcode,R) &= \Big\{ \Ptilde_{\Xhat,X} \in \Pscr(\Xscr \times \Xscr) :  % \non \\
    %& \hspace{1cm}
    \Ptilde_X = \prbX, \Ptilde_{\Xhat} = \Pcode, I(X;\Xhat) \leq R \Big\}. \non
  \end{align}
Then, there exists a sequence of strategies %\ach
$ \{g_n\}_{n \geq 1}$, an $ s_n \in \pbest_{\bd_0+\delta_n}(g_n) \;\forall  n $,  and $ \delta_n \rarr 0 $ such that $\lim_{n \rarr \infty} \Escr_{\bd_0+\delta_n}(g_n,s_n)$ and $  \lim_{n \rarr \infty} R_{\bd_0+\delta_n}(g_n,s_n) = R_0 $, where $ R_0 \leq R $ and $ \bd_0 = \Dbar(R) $. % \tododone{done}
\end{theorem}
\begin{proof}
  The proof follows on the lines of the proof of Theorem $1$ in \cite{lapidoth1997role}.  First observe the following.  From Lemma~\ref{lem:simplified_gn}, it suffices to define the strategies of the receiver $ g_n $ by choosing a set $ \Cscr^n \subseteq \Xscr^n $  and $ \xvec_0 \in \Cscr^n $ and defining
  \begin{align}
    g_n(\xvec) = \left\{
    \begin{array}{c l}
      \xvec & \mbox{ if } \xvec \in \Cscr^n \\
      \xvec_0 & \mbox{ else }
    \end{array}\right.. \label{eq:g_n_image_codebook}
\end{align}
Due to the structure of $g_n $, we can restrict the strategies of the sender to the set $ \{ s_n : \Xscr^n \rarr \Cscr^n \}$. %This is because for a best response strategy $ s_n^*$ and for a set of sequences $ \Zscr^n \subseteq \Xscr^n $ where $ s_n^*(\xvec')  \notin \Cscr^n $ for all $ \xvec' \in \Zscr^n$, the receiver maps the sequences to $ \xvec_0 $ which gives the utility $ \ut_n(\xvec_0,\xvec') $ for all $ \xvec' \in \Zscr^n $. The same utility can be achieved by a restricted strategy $ s_n' $ where $ s_n'(\xvec') = \xvec_0 $ for all $ \xvec' \in \Zscr^n $.
%With this observation, the strategies of the sender can be restricted to  $ s_n : \Xscr^n \rarr \Cscr^n $ and
Thus, the receiver can be restricted to the set $ \{ g_n : \Cscr^n \rarr \Cscr^n | \;g_n(\xvec) = \xvec \;\; \forall \; \xvec \in \Cscr^n \}$. Thus, the problem then boils down to choosing a codebook $ \Cscr^n  = \{1,\hdots, 2^{nR}\} \subseteq \Xscr^n $ of rate $ R $, or in other words, choosing an image of $ g_n $, that achieves arbitrarily small probability of error.

 We use the proof construct of Theorem~1 from \cite{lapidoth1997role} as follows. Fix  a  distribution $ \Pcode $ and a  sequence $ \delta_n \rarr 0 $. Generate the sequences in the set $ \Cscr^n $  according to $ \Phat $ defined as
 \begin{align}
   \Phat(\xvec) = \left\{
   \begin{array}{c l}
     \frac{1}{\gamma(n,\delta_n)}\prod_{i=1}^{n} \Pcode(x_i) & \mbox{ if } \xvec = (x_1,\hdots,x_n) \in T_{\Pcode,\delta_n}^n \\
     0 & \mbox{ else }
   \end{array}
   \right., \non
 \end{align}
 where $ \gamma(n,\delta_n) $ is the normalizing constant and tends to one with $ n $.

Then, following the arguments of proof of Theorem~1 from \cite{lapidoth1997role}, we get that with $ \bd_0 = \max_{\Pbar \in \Fscr} \Ebb_{\Pbar} d_n(\Xhat,X) $, there exists a sequence of strategies $ \{g_n\}_{n \geq 1 }$  defined as \eqref{eq:g_n_image_codebook} such that  %and a sequence  $ \epsilon_n \rarr 0 $ such that
% \begin{align}
$   \lim_{n \rarr \infty}\max_{s_n \in \best(g_n)}\Escr_{\bd_0 + \delta_n}(g_n,s_n) = 0$. %\non%< \epsilon_n, \non
% \end{align}
 %
 % distortion for the distribution $ P_0 $ and with a strategy $ g_n $ chosen as \eqref{eq:g_n_image_codebook} is at most $ \max_{\Pbar \in \Fscr} \Ebb_{\Pbar} \ut(\Xhat,X) $.
 Note that the sender may not use all the sequences in the set $ \Cscr^n $ to map its information and there could be unused sequences. Thus, using the \eqref{eq:defn-Ascr_d}  and \eqref{eq:defn-rate-Rdn}, we get
% \begin{align}
   $R_{\bd_0+\delta_n}(g_n,s_n)  \leq \frac{1}{n}  \log|\Cscr^n| \; \forall \;s_n \in \best(g_n)$. %\non
% \end{align}
This gives   $\lim_{n \rarr \infty} R_{\bd_0+\delta_n}(g_n,s_n)  \leq R $ for all sequences $ \{s_n\}_{n \geq 1}, s_n \in \pbest_{\bd_0+\delta_n}(g_n)$. The distortion $  \max_{\Pbar \in \Fscr} \Ebb_{\Pbar} \ut(\Xhat,X) $ is for a fixed distribution $ \Pcode $. Taking the minimum over all such distributions, we get that there exists a  sequence of strategies achieving vanishing error %\ach
with a rate at most $ R_0 \leq R $ for distortion $ d_0 = \Dbar(R) = \min_{\Pcode} \max_{\Pbar \in \Fscr} \;\Ebb_{\Pbar}[d(\Xhat,X)] $.
\end{proof}

\begin{proof}[Proof of Theorem~\ref{thm:lossl_rate_charac} b)]
	  $\vphantom{new}$

	First consider the following lemma that gives an alternate representation of the utility that is useful for the proofs.

	%       {\color{red} short}
	\begin{lemma}\label{lem:write_util_in_khat_k}
		Let $ \xvec, \widehat{\xvec} \in \Xscr^n $ be sequences such that $ p_{\xvec} \leq p_{\widehat{\xvec}} $. Then,
		\begin{align}
			\ut_n(\xvec,\widehat{\xvec}) =  \frac{1}{n} \left( \ut(1,0)\khat + (\ut(1,0) + \ut(0,1))k\right), \label{eq:util_khat_k_terms_type_less}
		\end{align}
		where $ \khat = n(p_{\widehat{\xvec}} - p_{\xvec}) $ and $ k \in [0,np_{\xvec}]$.
		When $ p_{\widehat{\xvec}} \leq p_{\xvec}$, we have
		\begin{align}
			\ut_n(\xvec,\widehat{\xvec}) =  \frac{1}{n}\left(\vphantom{\khat}  \ut(0,1) \mhat + (\ut(1,0)+\ut(0,1))m\right), \label{eq:util_khat_k_terms_type_more}
		\end{align}
		where $ \mhat = n(p_{\xvec} - p_{\widehat{\xvec}}) $ and $ m \in [0,np_{\widehat{\xvec}}]$.
	\end{lemma}
	\begin{proof}
		Consider the case where $ p_{\xvec} \leq p_{\widehat{\xvec}} $.   %Since the number of zeros in $ \widehat{\xvec} $ are more than the zeros in  $ \xvec $, there are at least $ n(p_{\widehat{\xvec}} - p_{\xvec})$ coordinates where $ x_i = 1 $ and $ \widehat{\xvec}_i = 0$. Thus, we have $ P_{\xvec,\widehat{\xvec}}(1,0) = p_{\widehat{\xvec}} - p_{\xvec} + k/n $ for some $ k \geq 0 $.
		Let $ P_{\xvec,\widehat{\xvec}}(0,1) = k/n $ where $ k \geq 0  $. Using
		%  \begin{align}
			$  p_{\xvec} = P_{\xvec,\widehat{\xvec}}(0,0) + P_{\xvec,\widehat{\xvec}}(0,1)$ and % \non \\
			$ p_{\widehat{\xvec}} = P_{\xvec,\widehat{\xvec}}(0,0) + P_{\xvec,\widehat{\xvec}}(1,0)$,  %\non
			%  \end{align}
		we get
		\begin{align}
			P_{\xvec,\widehat{\xvec}}(1,0)  &=  p_{\widehat{\xvec}}- p_{\xvec} +  P_{\xvec,\widehat{\xvec}}(0,1)= (p_{\widehat{\xvec}} - p_{\xvec}) + k/n. \non
		\end{align}
		We can thus write
		\begin{align}
			\ut_n(\xvec,\widehat{\xvec}) &= P_{\xvec,\widehat{\xvec}}(1,0)\ut(1,0)+P_{\xvec,\widehat{\xvec}}(0,1)\ut(0,1) \non \\
			&= \frac{1}{n}((\khat + k)\ut(1,0) + k\ut(0,1)),\non %=  \frac{1}{n}(\khat \ut(1,0) + (\ut(1,0) - \ut(0,1))k), \non% = \frac{1}{n}(\khat a + (a - b)k), \non% \label{eq:util_khat_k_terms}
		\end{align}
		where $ \khat = n(p_{\widehat{\xvec}} - p_{\xvec}) $ and $ k \in [0,np_{\xvec}]$. The maximum value of $ k $ occurs when $ P_{\xvec,\widehat{\xvec}}(0,0) = 0 $ and  $ P_{\xvec,\widehat{\xvec}}(1,0) = p_{\widehat{\xvec}} $, which gives $ k = n p_{\xvec}$. This can be rearranged to the form  in \eqref{eq:util_khat_k_terms_type_less}.

		For the case when $ p_{\widehat{\xvec}} \leq p_{\xvec} $, we take $ P_{\xvec,\widehat{\xvec}}(1,0) = m/n $ where $ m \geq 0  $. Using
		\begin{align}
			P_{\xvec,\widehat{\xvec}}(0,1)  =    (p_{\xvec}- p_{\widehat{\xvec}} ) + P_{\xvec,\widehat{\xvec}}(1,0)  = (p_{\xvec}- p_{\widehat{\xvec}} ) + m/n, \non
		\end{align}
		we can write the utility as  %   Thus, the utility is given as
		% \begin{align}
			$ \ut_n(\xvec,\widehat{\xvec}) = \frac{1}{n}( m\ut(1,0) + (\mhat+m)\ut(0,1)),$ % \non %=  \frac{1}{n}(\mhat \ut(0,1) + (\ut(1,0)-\ut(0,1))m), \non% = \frac{1}{n}(-b\mhat  + (a-b)m), \non
			%  \end{align}
		where $ \mhat = n(p_{\xvec} - p_{\widehat{\xvec}}) $ and $ m \in [0,np_{\widehat{\xvec}}]$. This  can be rearranged to the form  in \eqref{eq:util_khat_k_terms_type_more}.
	\end{proof}

     \textit{1) Proof of $ \Rscr^{\inf} = \entrP$}:

     We use Theorem~\ref{thm:dist_rate_bound} to prove the claim. We show that for the rate $ R = \entrP $,  $ \Dbar(R) = 0 $. This implies that  there exists a sequence of strategies %\ach %an {\color{blue} sequence of strategies } $ \{g_n\}_{n \geq 1}$, an $ s_n \in \pbest_{\bd_0+\delta_n}(g_n) $ for all $ n $,  and $ \delta_n \rarr 0 $ such that $ \Escr_{\bd_0+\delta_n}(g_n,s_n) \rarr 0 $ and $  \lim_{n \rarr \infty} R_{\bd_0+\delta_n}(g_n,s_n) = R_0 $, where
     that achieves the rate $ R_0 \leq \entrP $ for $ \bd_0 = \Dbar(\entrP) $.
      As discussed after Definition~\ref{defn:achiev-rate}, our notion of the minimum rate coincides with the information-theoretic notion of minimum rate and hence any achievable rate cannot be lower than $ \entrP $. It follows that $ \Rscr^{\inf} = \entrP $.

     Let $ \ut $ be such that $ \ut(1,0) + \ut(0,1) < 0 $. Taking  $ \Pcode = \prbX $ and for $ R = \entrP $, we have
  \begin{align}
    \Wscr(\prbX,\prbX,\entrP)
    &= \Big\{ \Ptilde_{\Xhat,X} \in \Pscr(\Xscr \times \Xscr) : \Ptilde_X = \prbX, \Ptilde_{\Xhat} = \prbX, %\non \\
                           %& \hspace{3.8cm}
                            I(X;\Xhat) \leq \entrP \Big\} \non \\
   % &= \Big\{ \Ptilde_{\Xhat,X} \in \Pscr(\Xscr \times \Xscr) : \Ptilde_X = \prbX, \Ptilde_{\Xhat} = \prbX, \non \\
   % & \hspace{3.8cm} H(Y|X) \geq 0 \Big\} \non \\
    &= \Big\{ \Ptilde_{\Xhat,X} \in \Pscr(\Xscr \times \Xscr) : \Ptilde_X = \prbX, \Ptilde_{\Xhat} = \prbX \Big\}. \non
  \end{align}
  % Thus, the set $ \Dscr(\prbX,\prbX,R) $ is a subset of all doubly-stochastic matrices.
Since the alphabet is binary, all the joint distributions $ \Ptilde_{\Xhat,X} \in  \Wscr(\prbX,\prbX,\entrP) $ are such that $ \Ptilde_{\Xhat,X}(0,1) = \Ptilde_{\Xhat,X}(1,0)$.  Further, $ \ut(0,1) + \ut(1,0) < 0 $ gives that $ \Ebb_{\Ptilde_{\Xhat,X}} \ut(\Xhat,X) = \Ptilde_{\Xhat,X}(0,1)(\ut(0,1) + \ut(1,0)) \leq 0 $ for all $ \Ptilde_{\Xhat,X} \in \Wscr(\prbX,\prbX,\entrP) $. In particular, \newline
%  \begin{align}
$\max_{\Ptilde_{\Xhat,X} \in \Wscr(\prbX,\prbX,\entrP) }  \Ebb_{\Ptilde_{\Xhat,X}} \ut(\Xhat,X) = 0$ %\non
%  \end{align}
 and this is achieved only by the diagonal matrix $ \Ptilde_{\Xhat,X}^* $ where $ \Ptilde_{\Xhat,X}^*(i,i) = \prbX(i) $ for all $ i \in \Xscr $. %Moreover, $ P^* $ lies in the set $ \Fscr $
%   \begin{align}
%   \Fscr &= \argmax_{\Ptilde \in \Wscr(\prbX,\prbX,\entrP) } \;\Ebb_{\Ptilde}\ut(\Xhat,X).
%  \end{align}
  This gives that $ \Fscr = \{\Ptilde_{\Xhat,X}^*\} $  and hence $ \Dbar(\entrP) = \min_{P \in \Fscr} \; \Ebb_P d(\Xhat,X) = 0 $. Using  Theorem~\ref{thm:dist_rate_bound} completes the proof.%, we get that for distortion $ \bd = 0 $, there exists an achievable sequence of strategies with rate $  R_0 \leq \entrP $.  Since an achievable rate cannot be lower than $ \entrP $, it follows that $ \Rscr^{\inf}  = \entrP $.

    \textit{2) Proof of $ \Rscr^{\sup} \leq  H\left( \min\left\{\frac{b}{a} p , \frac{1}{2}\right\}\right) $}:
%\todo{would be better to state this as $\Rscr^{\sup} < H(\min{...})$}
    We prove the claim by giving an upper bound on the number of utilized sequences for any  sequence of strategies achieving vanishing error. %\ach{\color{blue} achievable sequence of strategies.}

    Let $  I^n \subseteq \Xscr^n $ and $ \xvec_0 \in I^n $ define a sequence of strategies achieving vanishing error as  %\ach%n  {\color{blue} achievable sequence of strategies} $ \{g_n\}_{n \geq 1} $ where
    \begin{align}
    g_n(\xvec) = \left\{
    \begin{array}{c l}
      \xvec & \mbox{ if } \xvec \in I^n \\
      \xvec_0 & \mbox{ else } %\xvec \notin I^n
    \end{array}\right.. \label{eq:g_n_defn_proof_Rmax}
    \end{align}
    Fix $ \epsilon, \delta > 0 $ and let $ n $ be large enough such that $ \Pbb(T_{p,\delta}^n) > 1- \epsilon $ and
    \begin{align}
      \min_{s_n \in \best(g_n)}\Pbb(\Dscr_\delta(g_n,s_n)) = 1 - \max_{s_n \in \best(g_n)}\Escr_\delta(g_n,s_n) > 1 - \epsilon. \non
    \end{align}
    It follows that $ \Dscr_\delta(g_n,s_n) \cap T_{p,\delta}^n \neq \emptyset $ for all $ s_n \in \pbest_{\delta}(g_n) $ where $\pbest_{\delta}(g_n) $ is the set of worst-case best responses defined in  \eqref{eq:worst_case_br}.  Fix $ \shat_n \in \pbest_{\delta}(g_n) $ and let $ \widehat{\xvec} \in  \Dscr_\delta(g_n,\shat_n) \cap T_{p,\delta}^n$. As earlier, we assume that $ \shat_n(\wi{\xvec}	) \in I^n $. Since  $ \widehat{\xvec} \in \Dscr_\delta(g_n,\shat_n) $, there must be a best response sequence in  $\yvec^* \in T_{p,2\delta}^n \cap I^n $ %. Let this best response sequence be $ \yvec^* \in T_{p,2\delta}^n \cap I^n $, where
    such that $ \shat_n(\widehat{\xvec}) = \yvec^* $.

   We prove the claim by bounding the size of the set of utilised sequences given by $$ \Ascr_\delta^n(g_n,\shat_n)   =  \Big\{ \xvec \in I^n : g_n \circ \shat_n(\bar{\xvec}) = \xvec   \mbox{ for some }\;\bar{\xvec} \in \Dscr_\delta(g_n,\shat_n) \Big\}. $$ We now consider two cases based on the structure of the utility. In the first case, we directly bound the size of $ \Ascr_\delta^n(g_n,\shat_n) $. In the second case, we bound the size of $ I^n $ which gives a bound on the size of $ \Ascr_\delta^n(g_n,\shat_n) $. We prove these cases as the following claims. %Assume that $ b(p+2\delta)/a \leq 1/2 $. %}{I was not able to follow what is being claim in the cases below and how they contribute to the proof}{yellow}

    \begin{claim}\label{clm:bound_on_Adelta}
      Let $ g_n $ and $ \shat_n $ be as above. If  $ \ut(1,0) = a $ and $ \ut(0,1) = -b $ and $ b > a \geq 0 $, then there exists a $ \pbar \in [p,b(p+2\delta)/a)$ such that
      \begin{align}
      \lim_{n \rarr \infty} \frac{1}{n} \log |\Ascr_\delta^n(g_n,\shat_n)| \leq H(\pbar). \non
    \end{align}

    \end{claim}%\tododone{two cases of these proofs are now claims}
    \begin{proof}
      Let  $ \phat $ be a type where $ \phat \geq b(p+2\delta)/a $ and we call such a type as a ``faraway type''. Let $ U_{\phat}^n  $ be the type class corresponding to the type $ \phat $. We show that $ \Ascr_\delta^n(g_n,\shat_n)  \cap U_{\phat}^n  = \emptyset $.   Thus, sequences having type in the set $[ b(p+2\delta)/a,1]$ do not contribute to the rate.

      Let $ \wi{\zvec} $ be any sequence having the type $ \phat $. We show that $  \shat_n(\zvec) \neq \wi{\zvec} $ for all $ \zvec \in B_\delta(\wi{\zvec}) $ and thereby the sequence $ \wi{\zvec} $ does not contribute to the rate. %We prove this by showing that $$ \max_{z \in U_{\phat}^n}  \ut_n(\zvec,\bar{\zvec}) < \min_{\yvec \in T_{p,2\delta}^n} \ut_n(\yvec,\bar{\zvec}). $$
      % since there is a sequence $ \yvec^* \in T_{p,2\delta}^n \cap I^n $, then  $ \wi{\zvec} \n1otin \Ascr_\delta(g_n,\shat_n) $ \todo{where $\Ascr_\delta$ is?}. We prove this by showing that for all $ z $ in $ B_\delta(\wi{\zvec}) $,  we have $ \shat_n(\zvec) \notin U_{\phat}^n $.
      Let $ \bar{\zvec} \in  B_\delta(\wi{\zvec}) $ be such that $ p_{\bar{\zvec}} =  p_{\wi{\zvec}} + \delta_1 $, $ \delta_1 \in [0, \delta] $.    We compute the maximum utility $ \ut_n(\zvec,\bar{\zvec}) $ obtained over all sequences  $ \zvec \in U_{\phat}^n $. Since $ p_{\bar{\zvec}} \geq  p_{\zvec} = p_{\wi{\zvec}} $, using Lemma~\ref{lem:write_util_in_khat_k} we can write
      \begin{align}
        \ut_n(\zvec,\bar{\zvec}) =  \frac{1}{n} \left(\ut(1,0)\khat + (\ut(1,0) + \ut(0,1))k\right) = \frac{1}{n}(\khat a + (a - b)k), \label{eq:util_khat_k_proof}
      \end{align}
      where $ \khat = n(p_{\bar{\zvec}} - p_{\wi{\zvec}}) $ and $ k \in [0,np_{\wi{\zvec}}]$. Since $ a < b $, the maximum utility occurs when $ k = 0 $ and hence   % we have
     % \begin{align}
        $\max_{\zvec \in U_{\phat}^n}  \ut_n(\zvec,\bar{\zvec}) =  (p_{\bar{\zvec}} - p_{\wi{\zvec}})a = \delta_1 a$, %\non
     % \end{align}
      where $ \delta_1 \in [0,\delta] $.

      Now suppose $ \bar{\zvec} \in  B_\delta(\wi{\zvec}) $ is such that $p_{\bar{\zvec}} =  p_{\wi{\zvec}} - \delta_1 $, $ \delta_1 \in [0, \delta] $.  Then, for any $ \zvec \in U_{\phat}^n $, since $ p_{\bar{\zvec}} \leq p_{\zvec} = p_{\wi{\zvec}} $, we can write % using  Lemma~\ref{lem:write_util_in_khat_k} we can write
      % \begin{align}
      %   P_{z',\bar{\zvec}}(1,0) = P_{\bar{\zvec},z'}(0,1) = P_{z',\bar{\zvec}}(0,1)  - (p_{z'}-p_{\bar{\zvec}}) = k/n. \non
      % \end{align}
      % Thus, the utility is given as
      \begin{align}
        \ut_n(\zvec,\bar{\zvec}) = \frac{1}{n}\left(\vphantom{\khat}  \ut(0,1) \mhat + (\ut(1,0)+\ut(0,1))m\right) = \frac{1}{n}(-b\mhat  - (b - a)m), \label{eq:util_khat_k_terms_case2}
      \end{align}
      where $ \mhat = n(p_{\zvec} - p_{\bar{\zvec}}) $ and $ m \in [0,np_{\bar{\zvec}}] $.  Here, the maximum occurs at $ m = 0$, which gives
     % \begin{align}
      $  \max_{\zvec \in U_{\phat}^n}  \ut_n(\zvec,\bar{\zvec})  =  - \delta_1b$. %  \non
        % &\leq   (p_{\wi{\zvec}} - \delta_1)a - ((p_{\yvec^*}(0)+2\delta)b/a - \delta_1)b \non \\
        % &\leq (p_{\wi{\zvec}} - \delta_1)a - (p_{\yvec^*}(0)+2\delta)b - \delta_1b
%      \end{align}
      % Now for a fixed $ p_{\wi{\zvec}} $, we have for some $ \delta_2 \geq \delta $
      % \begin{align}
      %   (p_{\wi{\zvec}} - \delta_1)b &= ((p_{\yvec^*}+2\delta_2)b/a - \delta_1)b \non \\
      %   &\geq (p_{\yvec^*} + 2\delta_2 - \delta_1)b  \non \\
      %   &= p_{\yvec^*}b + (2\delta_2 - \delta_1)b > p_{\yvec^*}b. \non
      % \end{align}
       Together, for any $ \bar{\zvec} \in  B_\delta(\wi{\zvec}) $ we have for some   $ \delta_1 \in [0,\delta] $, % and $ \delta' > 0 $.
      \begin{align}
        \max_{\zvec \in U_{\phat}^n}  \ut_n(\zvec,\bar{\zvec}) = \left\{
        \begin{array}{c l}
          \delta_1 a  & \mbox{ if } p_{\bar{\zvec}} = p_{\wi{\zvec}} + \delta_1\\
          - \delta_1b  & \mbox{ if } p_{\bar{\zvec}} = p_{\wi{\zvec}} -\delta_1
        \end{array}
                         \right.. \label{eq:max_over_z_Bz}
      \end{align}

      We now compute the minimum utility $ \ut_n(\yvec,\bar{\zvec}) $ over all sequences $ \yvec \in T_{p,2\delta}^n $. Let $ \bar{\zvec} \in  B_\delta(\wi{\zvec}) $ be such that $p_{\bar{\zvec}} =  p_{\wi{\zvec}} + \delta_1 $, $ \delta_1 \in [0, \delta] $.  Using the format in Lemma~\ref{lem:write_util_in_khat_k} we get %\eqref{eq:util_khat_k_proof}, we get
    %  {\color{red} check what is $ y^* $ here}
      \begin{align}
        \min_{\yvec \in T_{p,2\delta}^n} \ut_n(\yvec,\bar{\zvec}) &= \min_{\yvec \in T_{p,2\delta}^n} \Big((p_{\bar{\zvec}}-p_{\yvec})a  - p_{\yvec}(b-a)\Big) \non \\
           &> (p_{\wi{\zvec}} + \delta_1)a - (p+2\delta)b  \label{eq:ineq_due_type_y} \\
           &\geq (p+2\delta)b + \delta_1 a  - (p+2\delta)b = \delta_1 a. \non
      \end{align}
      Here \eqref{eq:ineq_due_type_y}  follows because $p_{\bar{\zvec}} =  p_{\wi{\zvec}} + \delta_1 $  and $ p_{\yvec} < p+ 2\delta $ for all $\yvec \in T_{p,2\delta}^n$. The last inequality follows since $ p_{\wi{\zvec}} \geq p+ 2\delta $. For $ \bar{\zvec} \in  B_\delta(\wi{\zvec}) $  such that $p_{\bar{\zvec}} =  p_{\wi{\zvec}} - \delta_1 $, $ \delta_1 \in [0, \delta] $, we have
      \begin{align}
        \min_{\yvec \in  T_{p,2\delta}^n} \ut_n(\yvec,\bar{\zvec}) &> (p_{\wi{\zvec}} - \delta_1)a - (p+2\delta)b              %(p_{\bar{\zvec}}-p_{\yvec^*})a  - p_{\yvec^*}(b-a) \non \\
                                                     \geq - \delta_1 a \non. % - p_{\yvec^*}b \non \\
                                                     % &\geq (p+2\delta)(a-b) - \delta_1a. \non
      \end{align}
      Thus,  for any $ \bar{\zvec} \in  B_\delta(\wi{\zvec}) $ we have for some $ \delta_1 \in [0,\delta] $ and $ \delta', \delta'' > 0 $,
      \begin{align}
        \min_{\yvec \in  T_{p,2\delta}^n} \ut_n(\yvec,\bar{\zvec}) = \left\{
        \begin{array}{c l}
          \delta_1 a + \delta'  & \mbox{ if } p_{\bar{\zvec}} = p_{\wi{\zvec}} + \delta_1 \\
          - \delta_1 a + \delta'' & \mbox{ if } p_{\bar{\zvec}} = p_{\wi{\zvec}} - \delta_1
        \end{array}
                                    \right.. \label{eq:min_over_y_Tdelta}
      \end{align}
       From \eqref{eq:max_over_z_Bz} and \eqref{eq:min_over_y_Tdelta} we get that for  any $ \bar{\zvec} \in B_\delta(\wi{\zvec}) $,   %, gives that
     % \begin{align}
      $  \max_{\zvec \in U_{\phat}^n}  \ut_n(\zvec,\bar{\zvec}) < \min_{\yvec \in T_{p,2\delta}^n} \ut_n(\yvec,\bar{\zvec})$. %\non
%      \end{align}

      Thus,  when the sender observes $ \bar{\zvec} \in B_\delta(\wi{\zvec}) $, the least utility obtained over sequences in $ T_{p,2\delta}^n $ is greater than the highest utility obtained over  sequences in $ U_{\phat}^n $. Since $ \yvec^* \in I^n \cap T_{p,2\delta}^n $, the best response sequence for any $  \bar{\zvec} \in B_\delta(\wi{\zvec}) $ does not lie in the type class $ U_{\phat}^n$  %\tododone{$ \yvec^* \in  I^n \cap T_{p,2\delta}^n  $ is used here}
      and hence $ \shat_n(\bar{\zvec}) \neq \wi{\zvec} $. Since $ \wi{\zvec} $ was arbitrary, this holds for all sequences in the type class $ U_{\phat}^n $. Moreover, since $ \phat $ was arbitrary, this holds for all faraway types. This gives that $ U_{\phat}^n \cap  \Ascr_\delta^n(g_n,\shat_n)  = \emptyset$ for all $ \phat \geq b(p+2\delta)/a $ and hence $  \Ascr_\delta^n(g_n,\shat_n)  \subseteq \bigcup_{\phat < b(p+2\delta)/a} U_{\phat}^n.$ Thus,  there exists a $ \pbar \in [p,b(p+2\delta)/a)$ such that
  %    \begin{align}
      $  \lim_{n \rarr \infty} \frac{1}{n} \log |\Ascr_\delta^n(g_n,\shat_n)| \leq H(\pbar)$. %\non
%      \end{align}
       %Since $ \delta $ can be arbitrarily small, we get that all achievable rates lie in the set $ [H(p), H( \min\{ b p/a, 1/2\})] $.
    \end{proof}
    We now consider the case where $ \ut(1,0) = -b $ and $ \ut(0,1) = a $ and give an upper bound on the limit of $  \frac{1}{n} \log |\Ascr_\delta^n(g_n,\shat_n)|  $.

    \begin{claim}\label{clm:bound_on_In}
   Let $ I^n $  and $ g_n $ be as in \eqref{eq:g_n_defn_proof_Rmax} and let $ \shat_n $ as above.  If $ \ut(1,0) = -b $ and $ \ut(0,1) = a $ and $ b > a \geq 0 $, then there exists a  $ \pbar \in [p,b(p+2\delta)/a)$ such that
      \begin{align}
    \lim_{n \rarr \infty} \frac{1}{n} \log |\Ascr_\delta^n(g_n,\shat_n)| \leq  \lim_{n \rarr \infty} \frac{1}{n} \log |I^n| \leq H(\pbar). \non
    \end{align}
    Moreover, for any set $ \Ihat^n $ and the corresponding strategy $ \ghat_n $, if $$ \lim_{n \rarr \infty} \frac{\log |\Ihat^n|}{n} > H(b(p+2\delta)/a),$$ then
 %   \begin{align}
     $ \lim_{n \rarr \infty} \max_{s_n \in \best(\ghat_n)}\Escr_\delta(\ghat_n,s_n) = 1$.% \non
%    \end{align}
    \end{claim}
    \begin{proof}
      In this case, we prove the claim by bounding the size of the set $ I^n $. We show that if $ I^n $ includes sequences having a faraway type, \ie, a type in the set $ [b(p+2\delta)/a,1] $, then no sequence in $ T_{p,\delta}^{n} $ can be recovered correctly.  %a sequence such that $ P_{\wi{\zvec}} \in  [b(p+2\delta)/a,1]%Let $ I^n \cap U_{p_{\wi{\zvec}}}^n \neq \emptyset $. %Again, let $ \wi{\zvec} $ be a sequence such that $ P_{\wi{\zvec}} \in  [b(p+2\delta)/a,1] $.%Suppose there is a sequence having type in $  [b(p+2\delta)/a,1 ] $ in the set $ I^n $. %We prove this by showing that for any $ \xvec \in T_{p,\delta}^{n} $, the worst-case distortion

      % Since $ \xvec \in \Dscr_\delta(g_n,\shat_n) \cap T_{p,\delta}^n $, it implies that $ \sbar(\xvec) \in T_{p,2\delta}^n $ and hence it should hold that
      % \begin{align}
      %   \min_{\yvec' \in I^n \cap  T_{p,2\delta}^n}  \ut_n(\yvec',\xvec) \geq    \max_{z' \in I^n \cap U_{p_{\wi{\zvec}}}^n} \ut_n(\zvec',\xvec). \label{eq:max_util_min_util}
      % \end{align}
Recall the sequence $\widehat{\xvec} \in  \Dscr_\delta(g_n,\shat_n) \cap T_{p,\delta}^n$.       Let the type of $ \widehat{\xvec} $ be such that $ p_{\widehat{\xvec}} = p+\delta' $ with $ \delta' \in (-\delta,\delta)$. For $ \phat \in [b(p+2\delta)/a,1] $ and let $ U_{\phat}^n $ be the corresponding type class. We have that
      \begin{align}
        \min_{\zvec \in U_{\phat}^n} \ut_n(\zvec,\widehat{\xvec}) &= (\phat-p_{\widehat{\xvec}})a - p_{\widehat{\xvec}}(b-a) \non \\
                                                &> (p+2\delta)b - (p+\delta')b = (2\delta -\delta') b. \non
      \end{align}
      Further, using the structure from proof of Claim~\ref{clm:bound_on_Adelta}, %{\color{red} check} %Case $i$,
      \begin{align}
        \max_{\yvec \in T_{p,2\delta}^n} \ut_n(\yvec,\widehat{\xvec}) &=   \max_{\yvec \in T_{p,2\delta}^n} (p_{\yvec}-p_{\widehat{\xvec}})a  \non \\
          &< (p + 2\delta - p - \delta') a = (2\delta -\delta')a \non
      \end{align}
      and hence
    %  \begin{align}
       $ \max_{\yvec \in T_{p,2\delta}^n} \ut_n(\yvec,\widehat{\xvec}) <  \min_{\zvec \in U_{\phat}^n} \ut_n(\zvec,\widehat{\xvec})$. %\non
   %   \end{align}
      % This implies that
      % \begin{align}
      %   \max_{\yvec \in I^n \cap T_{p,2\delta}^n} \ut_n(\yvec,\widehat{\xvec}) <  \min_{z \in U_{\phat}^n} \ut_n(\zvec,\widehat{\xvec}) \leq  \max_{z \in I^n \cap U_{\phat}^n} \ut_n(\zvec,\widehat{\xvec}). \non
      % \end{align}
       This implies that for a sequence in $  T_{p,\delta}^n $, the sender prefers a sequence in the set $ U_{\phat}^n $ over all sequences in $ T_{p,2\delta}^n $.   %This holds for all faraway types $ \phat $.
          If the image of $g_n$ has a sequence from $ U_{\phat}^n$, then  no sequence from $ T_{p,\delta}^n $ will be recovered.  %having a faraway type. %However, we know that $ \widehat{\xvec} \in  \Dscr_\delta(g_n,\shat_n) \cap T_{p,\delta}^n $ and that the sequence $ \{g_n\}_{n \geq 1} $ achieves vanishing error. %\ach %is an {\color{blue} achievable sequence of strategies.}.
      Thus, $ I^n $ in \eqref{eq:g_n_defn_proof_Rmax} does not include a sequence with type $ \phat \in [b(p+2\delta)/a,1]  $.  Since $ \Ascr_\delta^n(g_n,\shat_n) \subseteq I^n $, we have that there exists a $ \pbar \in [p,b(p+2\delta)/a] $ such that %Hence, for an achievable sequence of strategies $ \{g_n\}_{n \geq 1} $, it has to be that sdf
      \begin{align}
       \lim_{n \rarr \infty} \frac{1}{n} \log |\Ascr_\delta^n(g_n,\shat_n)| \leq  \lim_{n \rarr \infty} \frac{1}{n} \log |I^n| \leq H(\pbar). \non
      \end{align}
%      {\color{red} is $\ghat$ required?}
     % The above arguments also prove that if a strategy $ \ghat_n $ with an image $ \Ihat^n $ contains a sequence having a faraway type, then no sequence from $ T_{p,\delta}^n $ is recovered within distortion $ \delta $. Thus,  the worst-case error $ \Escr_\delta(\ghat_n,s_n) $ over all $ s_n \in \best(\ghat_n) $ tends to one.
      The above arguments also prove that if the image of any strategy has a ``faraway type'' sequence, then no sequence from $ T_{p,\delta}^n $ is recovered within $ \delta $ and  the worst-case error  tends to one.
    \end{proof}
    From Claim~\ref{clm:bound_on_Adelta} and Claim~\ref{clm:bound_on_In}, we get that for all  sequence of strategies $ \{g_n\}_{n \geq 1} $ achieving vanishing error  %\ach
    and $ \shat_n \in \pbest(g_n)$, there exists a $ \pbar \in [p,b(p+2\delta)/a)$ such that \newline
  %  \begin{align}
    $    \lim_{n \rarr \infty} \frac{1}{n} \log |\Ascr_\delta^n(g_n,s_n)| \leq H(\pbar)$. % \non
%      \end{align}
Using Definition~\ref{defn:rate-of-comm} and the definition of the achievable rates from Definition~\ref{defn:achiev-rate}, we get that all achievable rates $ R $ are such that $ R \leq  H( \min\{ b (p+2\delta)/a, 1/2\})$.   Since $ \delta $  is arbitrary,  all achievable rates lie in the set $ [H(p), H( \min\{ b p/a, 1/2\})] $. Moreover, for the case  $ \ut(1,0) = -b $ and $ \ut(0,1) = a $, the probability of error for rates above $ H( \min\{ b p/a, 1/2\}) $ tends to one. This completes the proof.% of  Theorem~\ref{thm:lossl_rate_charac} b).
    \end{proof} %\todo{good.. well written}

\begin{proof}[Proof of Theorem~\ref{thm:lossl_rate_charac} c)]
  Since $ \ut(i,j) < 0 $ for all $ i \neq j $, it follows that for any distinct $ \xvec, \yvec \in \Xscr^n $, $ \ut_n(\xvec,\yvec)  < 0 $
  % \begin{align}
  %   \ut_n(\xvec,\yvec) = \frac{1}{n}\sum_{i} \ut(x_i,y_i) < 0, \non
  % \end{align}
  and hence the sender is truthful about its information. Since it trivially holds that $ \ut(0,1) + \ut(1,0) < 0 $, we have that $ \entrP  $ is achievable. Now for a strategy $ g_n $ defined as an identity function on $ \Xscr^n $, we have $ s_n(\xvec) = \xvec $ for all $ \xvec \in \Xscr^n $ and  for all $ s_n \in \best(g_n) $. Thus, $ \Dscr_{\bd}(g_n,s_n) = \Xscr^n $ for all $ s_n \in \best(g_n) $ and unity rate is also achievable.% and this completes the proof.
%
%
%  This corresponds to the case of cooperative communication  and hence the rate region is same as the Shannon rate region.
\end{proof}

\section{Proof of Lossy Recovery : Binary Alphabet }
\label{appen:bin_alph_lossy}

We prove the theorem in two parts. The proof of necessity of of \eqref{eq:u01_u10_cond_lossy} is given in Appendix~\ref{appen:proof_necc_cond} and the proof of sufficiency of \eqref{eq:u01_u10_cond_lossy} and Theorem~\ref{thm:lossy_rate_charac} a) and b) is given in Appendix~\ref{appen:suff_parta_partb}.
% \begin{proof}[\textit{of Theorem~\ref{thm:necc_lossy_bin}}]
%   Consider a sequence $ \xvec $ that is correctly recovered. Then there exists a $ \yvec^* $ such that $ d_n(\yvec^*,\xvec) \leq \bd $. Now take a sequence $ \bar{\xvec} $ of the same type such that $ \bar{\xvec}_k \neq x_k = y_k $. It can be shown that there exists $ \bar{\xvec} $ such that $ \ut(\yvec^*,\bar{\xvec}) > 0 ?$ and $ \d_n(\yvec^*,\bar{\xvec}) > \bd $.
%   From the proof of Theorem~\ref{thm:necc_u01_u10_cond}, we know that the graph $\Gs^n $ is a complete graph. Now consider two distinct sequences $ \xvec,\yvec $ of the same type. Then, we have
%   \begin{align}
%     \ut_n(\xvec,\yvec) &= d_n(\xvec,\yvec)(\ut(0,1) + \ut(1,0)) %&= \frac{1}{n}\sum_{i:x_i \neq y_i} \ut(x_i,y_i) < 0 \non \\
% %     &\leq d_n(\xvec,\yvec)\ut(0,1), \non
%   \end{align}
%   Clearly, the sender always chooses a sequence with the largest Hamming distance. So no sequence from the middle type class can be recovered.
%   If the type of $ \xvec,\yvec $ do not match, then
%   \begin{align}
%      d_n(\xvec,\yvec) \leq \ut_n(\xvec,\yvec) &= d_n(\xvec,\yvec) \ut(0,1) + \ut(1,0) %&= \frac{1}{n}\su
%   \end{align}
%   For type classes which have higher 0's, they all would prefer any sequence with lower 0's, again no one can be receovered.
%   For type classes
% \end{proof}
\subsection{Proof of necessity of \eqref{eq:u01_u10_cond_lossy}}
\label{appen:proof_necc_cond}
%\begin{proof}[\textit{of Theorem~\ref{thm:lossy_rate_charac} a)}]

%Proof of Theorem~\ref{thm:necc_u01_u10_cond}}
      % \label{appen:proof_necc_cond}

 %        It has to be that for all $ \xvec \in D_p $, there exists a corresp. $ \bar{\xvec} \in \Cscr^n $ such that $ d_n(\xvec,\bar{\xvec}) \leq \delta_n $.
The proof of necessity relies on the arguments in the proof of Theorem~\ref{thm:lossl_rate_charac} a). Let $ \bd \in (0,1/2) $ and let $ \delta > 0 $ be such that $ 2(\bd+\delta) < 1$.  Let $ \ut(1,0) = a $ and $ \ut(0,1) = - b $ and $ a \geq b \geq 0 $.  Using Lemma~\ref{lem:simplified_gn},  let $ I^n \subseteq \Xscr^n $ and  $ \xvec_0 \in I^n $ define a strategy  $ g_n $ of the receiver as %and defining % define a strategy $ g_n $ as
        \begin{align}
          g_n(\xvec) = \left\{
    \begin{array}{c l}
      \xvec & \mbox{ if } \xvec \in I^n \\
      \xvec_0 & \mbox{ if } \xvec \notin I^n
    \end{array}.\right. \non
        \end{align}
        Let $ s_n \in \pbest_{\bd+\delta}(g_n) $ be the worst-case best responses given by \eqref{eq:worst_case_br} and consider the set of sequences  recovered within distortion $ \bd + \delta $, \ie, $ \Dscr_{\bd + \delta}(g_n,s_n) $.    %We show that only a $\delta$-ball of sequences from every type class can be in $ \Dscr_\delta(g_n,s_n) $. We use this fact to show that the probability of error for any strategy of the receiver will tend to one as $ n $ grows large. %We prove this by considering two cases.
       If  $ \ut(0,1) + \ut(1,0)  \geq 0 $, then following the proof of Theorem~\ref{thm:lossl_rate_charac} a), we can show that only sequences in a $2(\bd + \delta)$-ball from every type class can be recovered within distortion $ \bd + \delta$.

         Suppose $ \xvec^1,\hdots, \xvec^L \in \Dscr_\delta(g_n,s_n) $ are sequences having distinct types that are recovered within distortion $ \delta $. Then,
    %   \begin{align}
     $\Dscr_{\bd +\delta}(g_n,s_n) \subseteq  \bigcup_{i \leq L} B_{2(\bd+\delta)}(\xvec^i) \cap U_{P_{\xvec^i}}^n$. % \non
%       \end{align}
Clearly, as long as  $ 2(\bd+\delta) < 1  $, we have that $ \lim_{n \rarr \infty}\min_{s_n \in \best(g_n)}\Pbb(\Dscr_{\bd+\delta}(g_n,s_n)) < 1 $. This holds for any sequence of strategies of the receiver and hence  %the error does not vanish for any sequence of strategies. %\ach
 %Thus,
 the achievable rate region is empty.

\subsection{Proof of sufficiency of \eqref{eq:u01_u10_cond_lossy} and Theorem~\ref{thm:lossy_rate_charac} a) and b)}
\label{appen:suff_parta_partb}

\begin{proof}[Proof of Theorem~\ref{thm:lossy_rate_charac} a)]
$\vphantom{new}$\newline
  1) $ \Rscr_{\bd}^{\inf} = \rateD.$
 % \todol{what is the proof plan? there is no clarity. You need to set the stage by saying what we know. $\Rscr_{ \bd}^{\inf} \geq R(d)$. So the goal is to show equality. Now what is the plan for the steps? Need a series of claims}

  We now use Theorem~\ref{thm:dist_rate_bound} to prove the claim. %The proof is similar to the proof of Theorem~\ref{thm:lossl_rate_charac} b).
   We show that for $ R = \rateD $, $ \Dbar(R) = \bd $.  This implies that  there exists a sequence of strategies %\ach
   that achieves the rate $ R_0 \leq \rateD $ for $ \bd_0 = \Dbar(\rateD) $ and hence
    %Thus,  there exists a  sequence of strategies %\ach%{\color{blue} achievable sequence of strategies.}
    %$ \{g_n\}_{n \geq 1}$, an $ s_n \in \pbest_{\bd_0+\delta_n}(g_n) $ for all $ n $,  and $ \delta_n \rarr 0 $ achieving the rate $  \lim_{n \rarr \infty} R_{\bd_0+\delta_n}(g_n,s_n) = R_0 $, where $ R_0 \leq \rateD $ and $ \bd_0 = \Dbar(\rateD) $.
  $ \Rscr_{\bd}^{\inf} = \rateD $.

  Let $  P_{\Xhat}^*  $ be the output distribution that achieves the minimum in the rate-distortion function. Define $ p^* = P_{\Xhat}^*(0) = (p-\bd)/(1-2\bd)$ (\cite{cover2012elements}, Ch. 10).  Taking $ \Pcode = P_{\Xhat}^* $ and $ R = R(\bd) $, we have
  \begin{align}
    & \Wscr(\prbX,P_{\Xhat}^*,R(\bd))  %\non \\
%    &= \Big\{ \Ptilde_{\Xhat,X} \in \Pscr(\Xscr \times \Xscr) : \Ptilde_X = \prbX, \Ptilde_{\Xhat} = \Pbb_{X'}, \non \\
%                           & \hspace{3.8cm} H(Y) - H(Y|X) \leq H(X'%) \Big\} \non \\
    = \Big\{ \Ptilde_{\Xhat,X} \in \Pscr(\Xscr \times \Xscr) : \Ptilde_X = \prbX, \;\Ptilde_{\Xhat} = P_{\Xhat}^*, \;I(\Xhat;X) \leq R(\bd) \Big\}. \non% \non \\
  \end{align}
  It follows from the definition of the rate-distortion function, that for all $ \Ptilde_{\Xhat,X} $ where $\Ptilde_X = \prbX $, if   $ I(\Xhat;X) \leq R(\bd) $, then $ \Ebb_{\Ptilde_{\Xhat,X}} d(\Xhat,X) \geq \bd $. Moreover, the minimum expected distortion $ \Ebb_{\Ptilde_{\Xhat,X}} d(\Xhat,X) = \bd $ is achieved for the  distribution $ \Ptilde_{\Xhat,X} =  \Ptilde_{\Xhat,X}^* $ where $ I(\Xhat;X) = R(\bd) $ and
  % $ \Ebb_{P_{\Xhat,X}} d(\Xhat,X) \leq \bd $, then $ I_{P_{\Xhat,X}}(\Xhat;X) \geq R(\bd) $. Alternatively, if
  \begin{align}
    \Ptilde_{\Xhat,X}^* = \left(
    \begin{array}{c c}
      p^* (1-\bd) &  p^*  \bd  \\
    p-p^* (1-\bd) &  1-p - p^*  \bd
    \end{array}
    \right).\non
  \end{align}
To show  $ \Dbar(\rateD) = \bd $, it suffices to prove that $ \Ptilde_{\Xhat,X}^* $ is unique. %the unique distribution that maximizes the expected utility. %\hltodo{}{make this into a claim... this proof also feels aimless. Structure it as a series of claims} \tododone{done}
\begin{claim}
 $  \argmax_{\Ptilde \in \Wscr(\prbX,P_{\Xhat}^*,R(\bd)) } \;\Ebb_{\Ptilde} \ut(\Xhat,X) = \left\{\Ptilde_{\Xhat,X}^* \right\}  $
\end{claim}
\begin{proof}
  We know from the rate-distortion theory that $ P_{\Xhat}^* $ is unique. Thus, for all distributions $ P'_{\Xhat,X} \in \Wscr(\prbX,P_{\Xhat}^*,R(\bd)),  P'_{\Xhat,X}  \neq  \Ptilde_{\Xhat,X}^*  $, we have $\Ebb_{ P'_{\Xhat,X}} d(\Xhat;X)  > \bd $ and for any such  $ P'_{\Xhat,X} $, %we have
  % such that %$ \Ebb_{P'} \ut(\Xhat,X) > \Ebb_{\Ptilde_{\Xhat,X}^*} \ut(\Xhat,X)$ and $ \Ebb_{ P'_{\Xhat,X}} d(\Xhat;X)  > \bd $. Now
  \begin{align}
    \Ebb_{ P'_{\Xhat,X}} d(\Xhat;X) = p^*P'_{X|\Xhat}(1|0) + (1-p^*)P'_{X|\Xhat}(0|1) > \bd. \non
  \end{align}
  Thus, one of the inequality $ P'_{X|\Xhat}(1|0) > \bd $ or $ P'_{X|\Xhat}(0|1) > \bd $ surely holds. %Note that using $ p^*  = (p-\bd)/(1-2\bd) $, we can write $ P'_{\Xhat,X}(1,0) = (1-p^*)\bd $.
  Without loss of generality, suppose $ P'_{X|\Xhat}(1|0) = \bd_1 > \bd $. Since $ P'_{\Xhat,X} \in \Wscr(\prbX,P_{\Xhat}^*,R(\bd)) $, equating the row marginal to $ P_{\Xhat}^* $ and column marginal to $ P_X $, we get that
  \begin{align}
    P'_{\Xhat,X} = \left(
    \begin{array}{c c}
      p^* (1-\bd_1) &  p^*  \bd_1  \\
      p-p^*(1 - \bd_1) &   1-p- p^*\bd_1
    \end{array}
                         \right).\non
  \end{align}
  This gives
  \begin{align}
    \Ebb_{P'} \ut(\Xhat,X)  &= (p-p^*) \ut(0,1) +   p^*\bd_1 ( \ut(0,1) + \ut(1,0))   \non \\
                            &< (p-p^*) \ut(0,1) +   p^*\bd ( \ut(0,1) + \ut(1,0))   =  \Ebb_{\Ptilde^*} \ut(\Xhat,X). \non
  \end{align}
  The inequality follows since $ \bd_1 > \bd  $ and $ \ut(0,1) + \ut(1,0 ) < 0 $. Moreover, this is true for any $ \bd_1 > \bd $. Thus, any  distribution $ P'_{\Xhat,X} \in \Wscr(\prbX,P_{\Xhat}^*,R(\bd)), P'_{\Xhat,X} \neq \Ptilde_{\Xhat,X}^* $ does not give a better expected utility to the sender.
\end{proof}
The above claim shows that $ \Dbar(\rateD) = \Ebb_{\Ptilde_{\Xhat,X}^*} d(\Xhat,X) = \bd $. Using Theorem~\ref{thm:dist_rate_bound} the proof is now complete. %\hltodocl{}{all usages of B.3 have to make sense with the new wording.}{cyan} %, we get that for the  distortion $ \bd  $, there exists an achievable sequence of strategies with rate $  R_0 \leq \rateD $. It follows that $ \Rscr_{\bd}^{\inf}  = \rateD $.
%\todo{still dont follow this proof. make it as series of claims.}\tododone{done}
 %This proves that $ R(\bd) $ is achievable when $ \ut(0,1) + \ut(1,0) < 0 $.

%Theorem~\ref{thm:dist_rate_bound}
2) $\Rscr_{\bd}^{\sup} \geq H(p+\bd) $ when $ p+\bd \leq 1/2 $ and $ \Rscr_{\bd}^{\sup} = 1  $ when  $p+\bd \geq 1/2$

To prove the claim, we first show that $ \Dbar( H(p+\bd)) = \bd $. Thus, using Theorem~\ref{thm:dist_rate_bound}, we get that there exists a sequence of strategies achieving the rate %\ach%n {\color{blue} achievable sequence of strategies} with a rate
$ R_0 \leq  H(p+\bd) $ for  $ \bd_0 = \Dbar(H(p+\bd))$. Then, we construct a sequence of strategies achieving  %\ach%{\color{blue} achievable sequence of strategies} to show that
 the rate $  H(p+\bd)$, which implies that rates between $ \rateD $ and $ H(p+\bd) $ are achievable and $\Rscr_{\bd}^{\sup} \geq H(p+\bd) $  when $ p+\bd \leq 1/2 $.

  Let $ \Pcode $ be a distribution where $  \Pcode(0) = p + \bd $. Taking $ R = H(p + \bd) $, we get
  \begin{align}
   \Wscr(\prbX,\Pcode,H(p + \bd))
%    &= \Big\{ \Ptilde_{\Xhat,X} \in \Pscr(\Xscr \times \Xscr) : \Ptilde_X = \prbX, \Ptilde_{\Xhat} = \Pbb_{X'}, \non \\
%                           & \hspace{3.8cm} H(Y) - H(Y|X) \leq H(X'%) \Big\} \non \\
    &= \Big\{ \Ptilde_{\Xhat,X} \in \Pscr(\Xscr \times \Xscr) : \Ptilde_X = \prbX, \Ptilde_{\Xhat} = \Pcode \Big\}. \non% \non \\
%                           & \hspace{3.8cm} H(Y|X) \geq 0 \Big\}.
  \end{align}
  % Thus, the set $ \Dscr(\prbX,\prbX,R) $ is a subset of all doubly-stochastic matrices.
With $ \ut(0,1) + \ut(1,0) < 0 $, we get
 % \begin{align}
  $  \max_{\Ptilde \in  \Wscr(\prbX,\Pcode,H(p + \bd)) }  \Ebb_{\Ptilde} \ut(\Xhat,X) = (\Pcode(0)-p)\ut(0,1) $ %\non
  %\end{align}
  which is achieved by the unique distribution $ \Ptilde_{\Xhat,X}^* $ where
  \begin{align}
    \Ptilde_{\Xhat,X}^* = \left(
    \begin{array}{c l}
      p     &  \Pcode(0)-p  \\
      0 & \Pcode(1)
    \end{array}
\right). \non
  \end{align}
  Thus, $ \Dbar(H(p+\bd)) = \Ebb_{\Ptilde_{\Xhat,X}^* } d(\Xhat,X) = \Pcode(0)- p = \bd $ and hence using Theorem~\ref{thm:dist_rate_bound}, we get that there exists a sequence of strategies achieving %\ach %n {\color{blue} achievable sequence of strategies} with
  the rate $  R_0 \leq H(p + \bd) $ for  $ \bd_0 =  \Dbar(H(p+\bd))$.

  The following claim completes the proof.
  \begin{claim} \label{clm:H(P+d)_achievable}
    The rate $ H(p + \bd) $ is achievable.
  \end{claim}
  \begin{proof}
 Suppose $ \ut(1,0) > 0 $.  We construct a sequence of strategies achieving the %\ach %n {\color{blue} achievable sequence of strategies} with the
 rate $ H(p+\bd) $.    Let $ P_n(0) = \argmin_{P \in \Pscr_n(\Xscr) : P(0) \geq \Pcode(0)} (P(0) - \Pcode(0)) $. Thus, $ P_n $ is the smallest type with $ P_n(0) \geq \Pcode(0)$. Consider the type class $ U_{P_n}^n $. Let $ \xvec_0 \in U_{P_n}^n $ and define a strategy $ g_n $ as
    \begin{align}
      g_n(\xvec) = \left\{
      \begin{array}{c l }
        \xvec & \mbox{ if } \xvec \in U_{P_n}^n \\
        \xvec_0 & \mbox{ else }% \xvec \notin U_{P_n}^n
      \end{array}
              \right.. \non
    \end{align}
    Let $ \epsilon_n > 0 $  and consider the  typical set $ T_{p,\epsilon_n}^n $. Let $ \epsilon_n \rarr 0 $ be such that $ p+\epsilon_n \leq P_n(0) $ for all $ n $ and $ \lim_{n \rarr \infty} \Pbb(T_{p,\epsilon_n}^n) = 1 $.   We show that for all $ n $, the  sequences in the typical set $ T_{p,\epsilon_n}^n $ are recovered within a distortion $ \bd + \delta_n $ for some $ \delta_n \rarr 0 $. Let $ p_{\zvec} = P_{\zvec}(0) $ and $ p_{\xvec} = P_{\xvec}(0) $. Using the form of utility given in Lemma~\ref{lem:write_util_in_khat_k}, for any sequence $ \zvec \in  T_{p,\epsilon_n}^n $ and a sequence $ \xvec \in U_{P_n}^n $, since $ p_{\zvec} \leq p_{\xvec} $, the utility is given as
    \begin{align}
       \ut_n(\xvec,\zvec) = \frac{1}{n}\left(\vphantom{\khat}  \ut(0,1) \mhat + (\ut(1,0)+\ut(0,1))m\right), \non
    \end{align}
         where $ \mhat = n(p_{\xvec} - p_{\zvec}) $ and $ m \in [0,np_{\zvec}]$. Since $ \ut(1,0) > 0 $, we have that $ \ut(0,1) < 0 $.  Thus, for a given $ \zvec $, the sequence $ \xvec^* $ that maximizes the utility is such that $ m^* = 0 $ and $ \mhat^* = \min_{\xvec \in U_{P_n}^n} (p_{\xvec} - p_{\zvec})$. Since $ \xvec^* $ is chosen from a type class, there always  exists a sequence with $ m^* = 0 $, \ie, with $ P_{\xvec^*,\zvec}(1,0) = 0 $ and $  P_{\xvec^*,\zvec}(1,1) = 1-p_{\xvec} $.  Furthermore, $ \mhat^* \leq P_n(0) - (p+\epsilon_n) = \bd + \gamma_n - \epsilon_n$, where $ \gamma_n \geq 0 $. This implies that $ d_n(\xvec^*,\zvec) \leq \bd +\delta_n $, where $ \delta_n = \gamma_n - \epsilon_n $. Thus, the sequence $ \zvec  $ is recovered within a distortion $ \bd + \delta_n $. Since $ \zvec \in T_{p,\epsilon_n}^n $ was arbitrary, $ d_n (g_n \circ s_n(\zvec),\zvec) \leq \bd +\delta_n $ for all $ s_n \in \best(g_n) $ and for all $ \zvec \in T_{p,\epsilon_n}^n $.
%           \begin{align}
% \max_{\xvec \in U_{P_n}^n}   \ut_n(\xvec,z) = \max_{\xvec \in U_{P_n}^n} (p_{\xvec} - p_z)\ut(0,1) = 0.\non % (p+\epsilon_n - P_n(0)) \ut(1,0). \non
%          \end{align}
%          Thus, for all $ s_n \in \best(g_n) $ is such that $ d_n(g_n \circ s_n(\zvec),z)$
%     \begin{align}
%      \ut_n(\xvec,z)  &= P_{\xvec,z}(0,1)\ut(0,1) +  P_{\xvec,z}(1,0)\ut(1,0) \non \\
%       &= (P_n(1) - ((1-p)-\epsilon_n))\ut(1,0) = (\bd + \delta_n)\ut(0,1), \non
%     \end{align}
%     where $ \delta_n = P_n(1)-(1-p)-\bd + \epsilon_n $.    Thus, for all $ s_n \in \best(g_n) $, we have that $ d_n(g_n \circ s_n(\zvec),z) \leq (\bd + \delta_n) $ for all $ z \in T_{p,\epsilon_n}^n  $.
         Thus,
   % \begin{align}
    $  \Escr_{\bd+ \delta_n}(g_n,s_n) \leq \Pbb\big(\Xscr^n \setminus T_{p,\epsilon_n}^n\big) \; \forall \; s_n \in \best(g_n). $ %\non
  %  \end{align}
    Observe that as $ \delta_n \rarr 0 $,   $ P_n(0) \rarr p +\bd $ and $ \epsilon_n \rarr 0 $ and hence $ \lim_{n \rarr \infty} \max_{s_n \in \best(g_n)} \Escr_{\bd+ \delta_n}(g_n,s_n) = 0 $.% and hence $ \{g_n\}_{n \geq 1} $ achieves vanishing error. %\ach %is an {\color{blue} achievable sequence of strategies}.

    We now compute the rate of the sequence of strategies.
    % using  Lemma~\ref{lem:rate_using_ind_set}. For this, we show that $ U_{P_n}^n $ is an independent set in $ \Gs^n $.
    For any pair of distinct sequences $ \yvec,\zvec \in U_{P_n}^n $, we have $  P_{\yvec,\zvec}(0,1) =  P_{\yvec,\zvec}(1,0)$ and hence
  %  \begin{align}
     $ \ut_n(\yvec,\zvec) = P_{\yvec,\zvec}(0,1)(\ut(0,1) +  \ut(1,0)) < 0$.   %\non
   % \end{align}
%Here, the last inequality follows since $ \ut(0,1) + \ut(1,0)  < 0 $.
 Thus, for all $ s_n \in \best(g_n)$, $ s_n(\zvec) = \zvec $ for all $ \zvec \in  U_{P_n}^n $ and hence $ \Ascr_{\bd+\delta_n}(g_n,s_n) = U_{P_n}^n $. % is an independent set in $ \Gs^n $.  From Lemma~\ref{lem:rate_using_ind_set},
 This gives that $ R_{\bd+\delta_n}(g_n,s_n) = \log |U_{P_n}^n| /n $ %for all $ s_n \in \best(g_n) $
and hence for any $ s_n \in \best(g_n) $,
   % \begin{align}
   \newline $  \lim_{n \rarr \infty} R_{\bd+\delta_n}(g_n,s_n) = H(p+\bd)$.  %\non
  %  \end{align}
    %Equality follows since $ P_n(0) \rarr p+\bd $.
  \end{proof}
  Thus,  the rate $ H(p + \bd) $ is achievable. Using convexity, the rates between $ \rateD $ and $ H(p + \bd) $ are achievable and hence the supremum of rates is at least $ H(p + \bd) $ when $ p+ \bd \leq  1/2 $. %This completes the proof. %Since $ p+\bd $ can be greater than $ 1/2 $, we bound the rate by $ 1 $ to get the following rate region
%   \begin{align}
% \left[\rateD,\; H\left(\min\left\{ p+\bd,\frac{1}{2}\right\}\right)\right]. \non
%   \end{align}
\end{proof}

\begin{proof}[Proof of Theorem~\ref{thm:lossy_rate_charac} b)]
Follows from part c) of Theorem~\ref{thm:lossl_rate_charac}.
\end{proof}

\section{Proof of Lossless Recovery : General Alphabet}
\label{append:gen_alph_lossl}

We prove the result in two parts. We first show that if $ \Rscr \neq \emptyset $, then $\Gamma(\ut) \leq 0 $. We show this by proving that if $ \Gamma(\ut) > 0 $, then $ \Rscr = \emptyset $.      Before proving the necessity of $ \Gamma(\ut) \leq 0 $, we state and prove two lemmas. The first lemma gives an alternate condition for $ \Gamma(\ut) \geq 0 $ in terms of a subset of symbols from $ \Xscr $. We then define a directed graph on the sequences $ \Xscr^n $ induced by the utility of the sender. The subsequent lemma then uses the above alternate condition and characterizes the edge relations for certain class of type classes in the directed graph. This lemma  %about the edge relations of the type classes
is crucial to prove  $ \Rscr = \emptyset $ when $ \Gamma(\ut) > 0 $. %\tododone{here I have only mentioned what the lemmas are, and that they will be used later, the outline is mentioned when the proof starts}

%In the next part, we
To show the existence of an achievable rate when $ \Gamma(\ut) < 0 $,   %We will use Theorem~\ref{thm:dist_rate_bound} for the proof and
we will proceed in a similar manner to the proof of  Theorem~\ref{thm:lossl_rate_charac} b).
The proof of necessity of $ \Gamma(\ut) \leq 0 $ is in Appendix~\ref{appen:nece_of_gamma}. The proof  of sufficiency of $ \Gamma(\ut) < 0 $ and part a) and b) of  Theorem~\ref{thm:lossl_rate_charac_gen} is in Appendix~\ref{appen:thm_loss_rate_charac_gen}.

Define   $ \ut^{\max} $ and $ \ut^{\min} $ as
%  \begin{align}
  $  \ut^{\max} =   \max_{i,j \in \Xscr} \;\ut(i,j)$, % \non \\
  $  \ut^{\min} =  \min_{i,j \in \Xscr} \;\ut(i,j)$. %\non
 % \end{align}
   For the purpose of the proof, we define the following directed graph on the set  $ \Xscr^n $.%$ U_{P_1}^n \cup U_{P_2}^n $.
  \begin{definition} \label{defn:general_sender_graph}
 For $ \delta > 0 $, let $ \Gsdir^n = (\Vtilde,\Etilde_\delta) $ be a graph on $ \Vtilde = \Xscr^n $ where there is a directed edge from $ \xvec $ towards $ \yvec $ if $  \ut_n(\yvec,\xvec) > \delta \ut^{\max} $.
    % \begin{align}
    %    \mbox{ either }  &\xvec \in U_{P_1}^n, \; \yvec \in U_{P_2}^n \mbox{ and }   \ut_n(\yvec,\xvec) > \delta \ut^{\max} \non \\
    %   \mbox{ or }  &\xvec \in U_{P_2}^n , \; \yvec \in U_{P_1}^n \mbox{ and } \ut_n(\yvec,\xvec) > \delta \ut^{\max}. \non
    % \end{align}
\end{definition}
A directed edge from $ \xvec $ to $ \yvec $ is denoted as $ \xvec \leadsto \yvec $.
%  \textbf{compare this graph with earlier graph}
 	As in Section~\ref{append:lossl_delta=0_empty_R},
 we have for the two sets $ U_1^n$ and $ U_2^n $ in the directed graph $ \Gsdir^n $, the number of edges \textit{emanating} from a sequence  $ \xvec \in U_1^n $ towards $ U_2^n $ is
 %\begin{align}
 $	\Big|\left\{ \yvec \in U_2^n : \xvec \leadsto \yvec \mbox{ in } \Gsdir^n \right\}\Big|.$ %  \non
 %\end{align}
 Moreover, for a fixed $ \yvec \in U_2^n
 $, the number of edges \textit{coming in}  from $ U_1^n $ is
 %\begin{align}
 $	\Big|\left\{ \xvec \in U_1^n : \xvec \leadsto \yvec \mbox{ in } \Gsdir^n \right\}\Big|$.  %\non
% \end{align}
  Similar to the Lemma~\ref{lem:bireg_subgraph_undir}, we have the following result.

\begin{lemma} \label{lem:bireg_subgraph}

  Let $ P_1, P_2 \in \Pscr_n(\Xscr) $ be two types and let $ U_{P_1}^n $ and $ U_{P_2}^n$ be the respective type classes. Then,
  \begin{align}
    	\Big|\left\{ \yvec \in U_{P_2}^n : \xvec \leadsto \yvec \mbox{ in } \Gsdir^n \right\}\Big| &= 	\Big|\left\{ \yvec \in U_{P_2}^n : \xvec' \leadsto \yvec \mbox{ in } \Gsdir^n \right\}\Big| \quad \forall \;\xvec,\xvec' \in U_{P_1}^n, \non \\
    		\Big|\left\{ \xvec \in U_{P_1}^n : \xvec \leadsto \yvec \mbox{ in } \Gsdir^n \right\}\Big| &= 	\Big|\left\{ \xvec \in U_{P_1}^n : \xvec \leadsto \yvec' \mbox{ in } \Gsdir^n \right\}\Big| \quad \forall\;\yvec, \yvec' \in U_{P_2}^n. \non
    	% \Delta_{\sf out}(\xvec, U_{P_1}^n ,U_{P_2}^n)  = \Delta_{\sf out}(U_{P_1}^n ,U_{P_2}^n) \quad \forall \; \xvec \in U_{P_1}^n, \non \\
    %\Delta_{\sf in}(\yvec, U_{P_2}^n ,U_{P_1}^n)  = \Delta_{\sf in}( U_{P_2}^n ,U_{P_1}^n)  \quad \forall \; \yvec \in U_{P_2}^n, \non
  \end{align}
%  for some $ \Delta_{\sf out}(U_{P_1}^n ,U_{P_2}^n), \Delta_{\sf in}( U_{P_2}^n ,U_{P_1}^n) \geq 0 $.
  Further, for any $ \xvec_0 \in U_{P_1}^n, \yvec_0 \in U_{P_2}^n $,
  \begin{align}
\Big|\left\{ \yvec \in U_{P_2}^n : \xvec_0 \leadsto \yvec \mbox{ in } \Gsdir^n \right\}\Big| \; |U_{P_1}^n| = \Big|\left\{ \xvec \in U_{P_1}^n : \xvec \leadsto \yvec_0 \mbox{ in } \Gsdir^n \right\}\Big|\; | U_{P_2}^n |. \non%\frac{ |U_{P_1}^n|}{|U_{P_2}^n|} = \frac{\Delta_2(U_{P_1}^n,U_{P_2}^n,\delta)}{\Delta_1(U_{P_1}^n,U_{P_2}^n,\delta)} . \non
  \end{align}
\end{lemma}
\begin{proof}
 Since $ U_{P_1}^n $ and $ U_{P_2}^n $ are type classes, any sequence in $ U_{P_1}^n $ will have the same number of outgoing edges towards $ U_{P_2}^n $. Also, every sequence in $ U_{P_2}^n$ will have the same number of incoming edges from $U_{P_1}^n$ and this  proves the first part.% and vice versa. The first assertion is thus evident.

  The  number of outgoing edges from $ U_{P_1}^n $ towards $ U_{P_2}^n $ is $ \Big|\left\{ \yvec \in U_{P_2}^n : \xvec_0 \leadsto \yvec \mbox{ in } \Gsdir^n \right\}\Big|  |U_{P_1}^n|$. Moreover, the  number of incoming edges towards $ U_{P_2}^n $ from  $ U_{P_1}^n $ is \newline $  \Big|\left\{ \xvec \in U_{P_1}^n : \xvec \leadsto \yvec_0 \mbox{ in } \Gsdir^n \right\}\Big| | U_{P_2}^n | $. Clearly, they must be equal, which proves the claim.% $ \Delta_{\sf out}(U_{P_1}^n,U_{P_2}^n) |U_{P_1}^n| = \Delta_{\sf in}(U_{P_2}^n,U_{P_1}^n) | U_{P_2}^n |$.
\end{proof}
Define
\begin{align}
	\Delta_{\sf out}(U_{P_1}^n,U_{P_2}^n) &= \Big|\left\{ \yvec \in U_{P_2}^n : \xvec \leadsto \yvec \mbox{ in } \Gsdir^n \right\}\Big|  \quad \forall \;\xvec \in U_{P_1}^n, \non \\
	\Delta_{\sf in}(U_{P_2}^n,U_{P_1}^n)  &= \Big|\left\{ \xvec \in U_{P_1}^n : \xvec \leadsto \yvec \mbox{ in } \Gsdir^n \right\}\Big|  \quad \forall\; \yvec \in U_{P_2}^n. \non
\end{align}

%{\color{red} check reference for these defns}
 % Consider distinct sequences $ \xvec, \xvec' \in U_{P_1}^n $ and let $ V(\xvec) = \{ \yvec \in U_{P_2}^n : \ut_n(\yvec,\xvec) > \delta\ut^{\max} \} $ and  $  V(\xvec') = \{ \yvec \in U_{P_2}^n : \ut_n(\yvec,\xvec') > \delta\ut^{\max} \}$. Let $ \pi : [n] \rarr [n]$ be the permutation such that $ \xvec' = \pi(\xvec) $. Now
  % \begin{align}
  %       \ut_n(\yvec,\xvec) =  \ut_n(\pi(\yvec),\pi(\xvec)) =  \ut_n(\pi(\yvec),\xvec') \non
  % \end{align}
  % and hence for all $ \yvec \in V(\xvec) $, the sequences $ \yvec' = \pi(\xvec') $ are such that $ \ut_n(\yvec',\xvec') > \delta \ut^{\max} $. Since $ \xvec, \xvec' $ have the same type and $  V(\xvec),  V(\xvec') \subseteq U_{P_2}^n $, it follows that $ V(\xvec') = \left\{ \pi(\yvec) \mbox{ where } \yvec \in U_{P_2}^n \right\}$ and hence $  |V(\xvec)| = |V(\xvec')| $. This holds for all pairs of types $ \xvec, \xvec' $ and hence  $  \Delta(\xvec,U_{P_1}^n,U_{P_2}^n,\delta) = \Delta_1(U_{P_1}^n,U_{P_2}^n,\delta) \;\; \forall \; \xvec \in U_{P_1}^n $ for some $ \Delta_1(U_{P_1}^n,U_{P_2}^n,\delta) \geq 0 $. We can similarly show the same for sequences in the type class $ U_{P_2}^n $ and get  $\Delta(\yvec,U_{P_2}^n,U_{P_1}^n,\delta) = \Delta_2(U_{P_1}^n,U_{P_2}^n,\delta) \;\; \forall \; \yvec \in U_{P_2}^n$ for some $ \Delta_2(U_{P_1}^n,U_{P_2}^n,\delta) \geq 0 $.
% Fix $ \gamma \in ()$

%Observe that $ \Gamma(\ut) > 0 $ implies $ \sum_{i,j \in \Xscr} Q(i,j)\ut(i,j) > 0 $ for all non-identity permutation matrices.

Observe that since a permutation can be decomposed into finite cyclic permutations, we have that for every permutation matrix $ Q $, there exist corresponding matrices $ Q^{(1)},\hdots, Q^{(M)} $ such that
\begin{align}
	\sum_{i,j \in \Xscr} Q(i,j)\ut(i,j) = \sum_{m \leq M} \sum_{i,j \in \Xscr} Q^{(m)}(i,j)\ut(i,j). \non
\end{align}
If $ \Gamma(\ut) > 0 $, then $\sum_{i,j \in \Xscr} Q(i,j)\ut(i,j) > 0 $ for some $ Q $.  Using the above form, we get that $\sum_{i,j \in \Xscr} Q^{(m)}(i,j)\ut(i,j) > 0 $ for some $ Q^{(m)} $.  We write
\begin{align}
	\sum_{i,j \in \Xscr} Q^{(m)}(i,j)\ut(i,j)  =  \ut(1,0) + \ut(2,1) + \hdots + \ut(0,K-1),
\end{align}
a subset of symbols $ \Kscr \subseteq \Xscr $ where $ \Kscr = \{0,1,\hdots,K-1\} $. Denote $ \bar{\Gamma}(\ut) = \ut(1,0) + \ut(2,1) + \hdots + \ut(0,K-1) $, where $ 0< \bar{\Gamma}(\ut) \leq \Gamma(\ut) $.
 Without loss of generality we assume that $\bar{\Gamma}(\ut) = \Gamma(\ut) $.  Since $ \Gamma(\ut) > 0 $, it follows that $ \ut^{\max} > 0 $.
%  there exists a subset of symbols $ \Kscr \subseteq \Xscr $ where $ \Kscr = \{0,1,\hdots,K-1\} $ that satisfy
%\begin{align}
%   \Gamma(\ut) =  \ut(1,0) + \ut(2,1) + \hdots + \ut(0,K-1) > 0.  \label{eq:neg-weight-chain_K}
%\end{align}

Define $ P^{\min} = \argmin_{i \in \Kscr} \prbX(i) $. The set of types contained in the typical set $ T_{\prbX,\delta}^n $ is denoted as  $\Pscr_{n}(T_{\prbX,\delta}^n) $ and is given by
\begin{align}
  \Pscr_{n}(T_{\prbX,\delta}^n) = \Big\{ P \in \types :  \prbX(i)-\delta < P(i) < \prbX(i)+\delta\quad \forall \; i \in \Xscr \Big\}. \label{eq:Pscr_n_delta}
\end{align}
%\tododone{definition changed}
Let $ \delta  > 0 $ be small enough such that
\begin{align}
   P^{\min} - 2\delta &> 0, \label{eq:delta_small_defn1} \\
 2q\delta \ut^{\min}  + (P^{\min}-2\delta)\Gamma(\ut) &>  \delta \ut^{\max}. \label{eq:delta_small_defn2}
  % \frac{P^{\min}}{K} &< P^{\min} - 2\delta, \label{eq:delta_small_defn1} \\
  % \delta \ut^{\max} &< \frac{P^{\min}}{K} \Gamma(\ut) +  \delta \ut^{\min}. \label{eq:delta_small_defn2}
\end{align}

% Define a set of types $ \Pscr_{n}(T_{\prbX,\delta}^n) $ as
% \begin{align}
%   \Pscr_{n}(T_{\prbX,\delta}^n) = \Big\{ P \in \types : P(i) \geq P^{\min}-\delta \quad \forall \; i \in \Kscr \Big\}. \label{eq:Pscr_n_delta}
% \end{align}
The following lemma shows that there exists edges between any two type classes having types from the set $ \Pscr_{n}(T_{\prbX,\delta}^n) $.
\begin{lemma} \label{lem:Delta_in_out_positive}
  Let  $ \delta > 0 $ satisfy \eqref{eq:delta_small_defn1} and \eqref{eq:delta_small_defn2} and let $ P_1, P_2 \in \Pscr_{n}(T_{\prbX,\delta}^n)  $. Then, $ \Delta_{\sf out}(U_{P_1}^n ,U_{P_2}^n)$ and $ \Delta_{\sf in}( U_{P_2}^n ,U_{P_1}^n)$ in Lemma~\ref{lem:bireg_subgraph} are positive.
\end{lemma}
\begin{proof}
  We will prove the claim by showing that there exists sequences in $ U_{P_2}^n $ that have directed edges with sequences in $ U_{P_1}^n $. For this, it suffices to prove that for a sequence $ \xvec \in U_{P_1}^n $ there exists a sequence $ \yvec \in U_{P_2}^n $ such that $ \ut_n(\yvec,\xvec) > \delta \ut^{\max} $. Fix $ \xvec \in U_{P_1}^n $ and let $ \yvec^* \in U_{P_2}^n $ be a sequence such that $ P_{\xvec,\yvec^*}(i,i) \geq \prbX(i) - \delta $ for all $ i \in \Xscr $. Thus, $ d_n(\xvec,\yvec^*) \leq 2q\delta $. We now show existence of a sequence $ \yvec \in U_{P_2}^n $ taking $ \yvec^* $ as a reference sequence such that $ \ut_n(\yvec,\xvec) > \delta \ut^{\max} $.
%if $ y_k^* \neq x_k $ then $ y_k = \yvec^*_k $. Thus,
  Let $ \yvec \in U_{P_2}^n$ such that  if $ y_k \neq y_k^* $, then $ y_k^* = x_k $. Further, the coordinates where $ y_k \neq y_k^* $ are chosen such that
  \begin{align}
    P_{\yvec,\yvec^*}(i,j) = \left\{\begin{array}{c l }
                             P^{\min} -2\delta & \mbox{ if } i = (j+1)\; \mbox{mod} K, j \in \Kscr \\
                            P_2(j) -   (P^{\min}-2\delta) & \mbox{ if } i=j \mbox{ and } j \in \Kscr \\
                            P_2(j) & \mbox{ if } i=j \mbox{ and } j \notin \Kscr \\
                            0 & \mbox{ otherwise }
                          \end{array} \right.. \non
  \end{align}
 It is clear from the marginal of $ P_{\yvec,\yvec^*} $ that the sequence $ \yvec \in U_{P_2}^n $. % From the condition $ P_{\xvec,\yvec^*}(i,i) \geq \prbX(i) - \delta $ for all $ i \in \Xscr $ and from the choice of $ \delta $ from \eqref{eq:delta_small_defn1}, it is clear that the sequence $ \yvec $ indeed exists in $ U_{P_2}^n $.
Using this, we have
   \begin{align}
     \ut_n(\yvec,\xvec)  %&= \frac{1}{n}\sum_{y_k \neq x_k }\ut(y_k,x_k) \non \\
                     &= \frac{1}{n}\sum_{y_k \neq x_k,  y^*_k \neq x_k }\ut(y_k,x_k) +  \frac{1}{n}\sum_{y_k \neq x_k,  y^*_k  = x_k}\ut(y_k,x_k) \non \\
     &= \frac{1}{n}\sum_{y_k = y_k^* \neq x_k  }\ut(y_k,x_k) + \sum_{i \neq j}P_{\yvec,\yvec^*}(i,j)\ut(i,j) \non \\
      %               &= \sum_{i \neq j}P_{\yvec,\xvec}(i,j)\ut(i,j) \non \\
     &= \ut_n(\yvec^*,\xvec) + \sum_{j}P_{\yvec,\yvec^*}((j+1)\; \mbox{mod} K,j)\;\ut((j+1)\; \mbox{mod} K,j)  \non \\
                     &\geq d_n(\yvec^*,\xvec)\ut^{\min} +  (P^{\min}-2\delta)(\ut(1,0) + \ut(2,1) + \hdots + \ut(0,K-1)) \non \\
                     &\geq  2q\delta \ut^{\min}  + (P^{\min}-2\delta)\Gamma(\ut)  %\label{eq:bnd_u_n(y*,x)} %\\
                 %&\geq 2q\delta \ut^{\min} + (P^{\min}-2\delta)\Gamma(\ut)
                     >     \delta \ut^{\max}. \label{eq:bnd_u_n(y*,x)}
   \end{align}
   %\textbf{above expression correct?}
  Here, the first inequality in \eqref{eq:bnd_u_n(y*,x)} follows by using  $ d_n(\yvec^*,\xvec) \leq 2q\delta $ and $ \ut^{\min} \leq 0 $ and the last inequality follows from the definition of $ \delta $ in \eqref{eq:delta_small_defn2}.  This gives that $     \Big|\left\{ \yvec \in U_{P_2}^n : \xvec \leadsto \yvec \mbox{ in } \Gsdir^n \right\}\Big|  > 0 $ and hence $ \Delta_{\sf out}(U_{P_1}^n,U_{P_2}^n) > 0 $. It also follows that $ \Delta_{\sf in}(U_{P_1}^n,U_{P_2}^n) > 0 $. %This holds for all pairs of types in the set $ \Pscr_{n}(T_{\prbX,\delta}^n) $.
\end{proof}
% We now use the above results to show that if  $ \Gamma(\ut) > 0 $, then there does not exist any achievable strategy for the receiver. For that we define the following.

% Let $ P^{\min} = \argmin_{i \in \Kscr} \prbX(i) $. Since $ \Gamma(\ut) > 0 $, it follows that $ \ut^{\max} > 0 $. For $ \gamma \in (0,1] $, define $\hat{\delta} > 0 $  as
% \begin{align}
%   \hat{\delta}\ut^{\max} < \gamma P^{\min}. \label{eq:defn_delta_hat}
% \end{align}
% For a fixed $ n $, define a set of types $ \Pscr_{n,\gamma}(\Xscr) $ as
% \begin{align}
%   \Pscr_{n,\gamma}(\Xscr) = \big\{ P \in \Pscr_n(\Xscr) : P(i)\geq \gamma P^{\min} \big\}.\non
% \end{align}

      \subsection{Necessity of $ \Gamma(\ut) \leq 0 $}
\label{appen:nece_of_gamma}
\begin{proof}[Proof of Necessity of $ \Gamma(\ut) \leq 0 $]
%	\todo{Add proof plan. We will show that for any two types... for this we first show a bound on the utility of recovered sequences etc etc.}
        As mentioned earlier, we prove the claim by showing that if $ \Gamma(\ut) > 0 $, then $ \Rscr = \emptyset $. We first determine a necessary condition for any sequence to be recovered correctly. We then consider type classes in the high probability typical set around $ \prbX $. We then show that if the image of any strategy of the receiver includes sequences from any of these type classes, then a ``large fraction" of sequences from typical set fail to satisfy this necessary condition and hence are not recovered correctly. Thereby, the error does not vanish for any sequence of strategies of the receiver and the rate region is empty.

       % \tododone{done}

        Fix $ \delta > 0 $  such that $ 2\delta $ satisfies \eqref{eq:delta_small_defn1} and \eqref{eq:delta_small_defn2}.    Using Lemma~\ref{lem:simplified_gn},  let $ I^n \subseteq \Xscr^n $ and  $ \xvec_0 \in I^n $ define a strategy  $ g_n $ of the receiver as %and defining % define a strategy $ g_n $ as
        \begin{align}
          g_n(\xvec) = \left\{
    \begin{array}{c l}
      \xvec & \mbox{ if } \xvec \in I^n \\
      \xvec_0 & \mbox{ if } \xvec \notin I^n
    \end{array}.\right. \non
        \end{align}
        Let $ s_n \in \best(g_n) $ and consider the set of sequences  recovered within distortion $ \delta $, \ie, $ \Dscr_\delta(g_n,s_n) $. We proceed with the proof in steps by proving two claims. First we prove a necessary condition for any sequence to be recovered within distortion $ \delta $.

        \begin{claim} \label{clm:util_recov_br}
          For a sequence $ \widehat{\xvec} \in \Dscr_\delta(g_n,s_n) $,  $ \ut_n(\yvec,\widehat{\xvec}) \leq  \delta \ut^{\max} $  for all $ \yvec \in I^n $.
        \end{claim}
        \begin{proof}
          Since $ \widehat{\xvec} \in \Dscr_\delta(g_n,s_n) $, there exists a $ \yvec^* \in I^n $ such that $ \yvec^* \in \argmax_{\yvec \in I^n }\ut_n(\yvec,\widehat{\xvec}) $ and $ d_n(\yvec^*,\widehat{\xvec}) \leq \delta $.  This implies
    %    \begin{align}
        $ \ut_n(\yvec,\widehat{\xvec})  \leq  \ut_n(\yvec^*,\widehat{\xvec}) %&= \sum_{i\neq j}P_{\yvec^*,\widehat{\xvec}}(i,j)\ut(i,j) \non \\
         % \leq \sum_{i\neq j}P_{\yvec^*,\widehat{\xvec}}(i,j)\ut^{\max}
          \leq d_n(\yvec^*,\widehat{\xvec}) \ut^{\max} \leq \delta \ut^{\max}\;\;\forall\;\yvec \in I^n $.  % \non
%        \end{align}
%This gives that for all $ \yvec \in I^n $, $ \ut_n(\yvec,\widehat{\xvec}) \leq \delta \ut^{\max} $.
        \end{proof}

        For the  given $ \delta $, let $ T_{\prbX,2\delta}^n $ be a typical set and recall the set of types in this typical set defined by $ \Pscr_n(T_{\prbX,2\delta}^n) $ given in \eqref{eq:Pscr_n_delta}.  Later, we define a typical set around $ \prbX $ that is contained in this set. Before that, we prove the following claim that shows that %since $ 2\delta $ satisfies \eqref{eq:delta_small_defn1} and \eqref{eq:delta_small_defn2},
         for any two types $ P, \Pbar \in \Pscr_n(T_{\prbX,2\delta}^n) $, if $ (1-\beta) $ fraction of sequences from $ U_P^n $ are present in $ I^n $, then no more that $ \beta $ fraction of sequences can be recovered within distortion $ \delta $ from any other type class $ U_{\Pbar}^n $.
        \begin{claim} \label{clm:beta_frac_recov}
          For  $ P, \Pbar \in \Pscr_n(T_{\prbX,2\delta}^n)$, let $ U_P^n $ and $ U_{\Pbar}^n $ be the  type classes corresponding to the  types $ P $ and $ \Pbar $. If $ | U_P^n \cap I^n | = (1-\beta) |U_P^n| $ for some $ \beta \in [0,1] $, then   $ | D_\delta(g_n,s_n) \cap U_{\Pbar}^n  | \leq \beta |U_{\Pbar}^n| $.
        \end{claim}
        % \hltodocl{\textbf{Step 2:}}{what is $\Pscr_{n,\delta}$?}{yellow}
        \begin{proof}
        %  \hltodo{Let $ P, \Pbar \in \Pscr_{n}(T_{\prbX,\delta}^n)$. Let $ U_P^n $ and $ U_{\Pbar}^n $ be the  type classes corresponding }{again write this as a claim}
          For the given $ \delta $, let $ \Gtilde_{{\sf s}, 2\delta}^n  $ be the directed graph defined in Definition~\ref{defn:general_sender_graph}. Let $ I_{P}^n =  U_P^n \cap I^n $ be the set of sequences from the type class $ U_P^n $ that are present in the set $ I^n $. For the type $ \Pbar $, let $ \Fscr_{\Pbar}^n \subseteq U_{\Pbar}^n $ be the set defined as
          \begin{align}
            \Fscr_{\Pbar}^n = \Big\{ \xvec \in U_{\Pbar}^n  : \exists \; \yvec \in I_P^n \;s.t.\; \xvec \leadsto \yvec \mbox{ in } \Gtilde_{{\sf s}, 2\delta}^n \Big\}. \non %d_n(\yvec,\xvec) > \delta,\;\;\ut_n(\yvec,\xvec) \geq \max_{\yvec' \in I_P^n, d_n(\yvec',\xvec) \leq \delta} \ut_n(\yvec',\xvec)  \Big\}. \non
          \end{align}
          The set  $\Fscr_{\Pbar}^n$ contains all the sequences  $ \xvec \in  U_{\Pbar}^n $ for which there exists a sequence $ \yvec \in  I_P^n $ such that $ \ut_n(\yvec,\xvec) > 2\delta \ut^{\max} $, \ie, there is a directed edge from $ \xvec $ to $ \yvec $. From Claim~\ref{clm:util_recov_br}, it follows that no sequence in $ \Fscr_{\Pbar}^n $ can be recovered within distortion $ \delta $ and  % since for all $ \xvec \in \Fscr_{\Pbar}^n $, there exists a sequence $ \yvec \in I^n $ such that $ \ut_n(\yvec,\xvec) > \delta \ut^{\max} $.
          hence $ \Dscr_\delta(g_n,s_n) \cap \Fscr_{\Pbar}^n  = \emptyset $.

          Since $ |\Dscr_\delta(g_n,s_n) \cap U_{\Pbar}^n| \leq  |U_{\Pbar}^n| - |\Fscr_{\Pbar}^n| $, to % from Step I, it follows that $  \Fscr_{\Pbar}^n  $ is the set of all the sequences in $ U_{\Pbar}^n $ that are not recovered within the distortion $ \delta $.
           prove the claim we show that  $  |\Fscr_{\Pbar}^n| \geq (1-\beta)|U_{\Pbar}^n| $. %\tododone{done}%}{how does this help? upper bound or lower bound? what is the proof strategy? sounds totally random. Why not be more clear and direct? We show that $\Fscr_{\Pbar}$ is of size at least $(1-\beta)|U|$, from which the result will follow.}{pink}
         % \todol{overall your writing style has been very irritating to read in this paper. No effort made to make it easier for the reader.}
          The number of edges coming in towards any sequence in $ U_P^n $ from $ U_{\Pbar}^n $ is $ \Delta_{\sf in}(U_P^n,U_{\Pbar}^n) $. Thus, the total number of such incoming edges towards $ I_P^n $ are given by $\Delta_{\sf in}(U_{P}^n,U_{\Pbar}^n) \;|I_P^n| $. Further, the number of edges from any sequence in  $  U_{\Pbar}^n $ towards the set $ U_P^n $ is  $ \Delta_{\sf out}(U_{\Pbar}^n,U_{P}^n) $. Thus the number of outgoing edges from $\Fscr_{\Pbar}^n $ towards $ I_P^n $ is given by $  \Delta_{\sf out}(U_{\Pbar}^n,U_{P}^n)\;| \Fscr_{\Pbar}^n| $. Since the set $  \Fscr_{\Pbar}^n  $ can have edges towards sequences in $ U_P^n $ outside $ I_P^n $, the latter must be at least $ \Delta_{\sf in}(U_{P}^n,U_{\Pbar}^n) \;|I_P^n| $. Note that since $ 2\delta $ satisfies \eqref{eq:delta_small_defn1} and \eqref{eq:delta_small_defn2}, due to  Lemma~\ref{lem:Delta_in_out_positive}, we have $ \Delta_{\sf out}(U_{\Pbar}^n,U_{P}^n),  \Delta_{\sf in}(U_{P}^n,U_{\Pbar}^n) > 0 $. This gives that $    \Delta_{\sf out}(U_{\Pbar}^n,U_{P}^n)\;| \Fscr_{\Pbar}^n|  \geq  \Delta_{\sf in}(U_{P}^n,U_{\Pbar}^n) \;|I_P^n|  $ and hence
          \begin{align}
            | \Fscr_{\Pbar}^n| %&\geq \frac{\Delta_{\sf in}(U_{P}^n,U_{\Pbar}^n) }{ \Delta_{\sf out}(U_{\Pbar}^n,U_{P}^n)} |I_P^n| \non \\
                               &\geq \frac{|U_{\Pbar}^n|}{|U_P^n|} |I_P^n| = (1-\beta)|U_{\Pbar}^n|. \non
          \end{align}
          The last equality follows from Lemma~\ref{lem:bireg_subgraph}.  Thus, $ |\Dscr_\delta(g_n,s_n) \cap U_{\Pbar}^n| \leq  |U_{\Pbar}^n| - |\Fscr_{\Pbar}^n| \leq \beta |U_{\Pbar}^n| $.% \todo{unclear with no overall strategy.}
        \end{proof}

        %\textbf{Step 3:}
        We now put together the assertions of Claim~\ref{clm:util_recov_br} and Claim~\ref{clm:beta_frac_recov} to complete the proof.%In this step, we put together the earlier steps to complete the proof.

        For the $ \delta  $ defined as in \eqref{eq:delta_small_defn1}, \eqref{eq:delta_small_defn2}, let $ \epsilon \leq \delta $ and consider the typical set $ \Typepx $.  Let $ \epsilon $ be small enough such that all the types in the $\delta$-ball around $ \Typepx $ is contained in $ \Pscr_n(T_{\prbX,2\delta}^n) $ (cf \eqref{eq:Pscr_n_delta}), \ie,
   \begin{align}
\forall \;\zvec \in  \{ \xvec \in \Xscr^n : \exists \; \yvec \in \Typepx \;\;s.t.\;\; d_n(\yvec,\xvec) \leq \delta\} \;\mbox{ we have }\;  P_{\zvec} \in \Pscr_n(T_{\prbX,2\delta}^n).    \label{eq:delta_ball_in_P2delta}
   \end{align} %\tododone{here I have introduced a eps-typ. set with eps small enough such that types in the delta ball around the set is in P2delta set}
   Finally, let $ n $ be large enough such that $ \Pbb( \Typepx ) \geq 1-\epsilon $.
            %             From the choice of $ \epsilon $ and from the definition of $ \Pscr_n(T_{\prbX,2\delta}^n) $ given in \eqref{eq:Pscr_n_delta}\tododone{done}, it follows that all the types occurring in $  \Typepx$ also lie $ \Pscr_n(T_{\prbX,2\delta}^n) $.%, \ie, for all $ \xvec \in \Typepx $, we have $ P_x \in \Pscr_{n}(T_{\prbX,\delta}^n) $.

%\tododone{modified arguments according to the defn of $\Pscr_n(T_{\prbX,2\delta}^n)$}
Let $ P^* \in  \Pscr_n(T_{\prbX,2\delta}^n) $, be such that  $ |I^n \cap U_{P^*}^n | \geq | I^n \cap U_{\Pbar}^n |  $  for all $ \Pbar \in  \Pscr_n(T_{\prbX,2\delta}^n) $. Thus, the type class $ U_{P^*}^n $ has the largest share of sequences in $ I^n $ among all other types from $ \Pscr_n(T_{\prbX,2\delta}^n)  $. Let $\beta^* =  |I^n \cap U_{P^*}^n | /|U_{P^*}^n| $.
Since   $ \Pscr_n(T_{\prbX,2\delta}^n)  $ contains all types occurring in $ \Typepx $, from Claim~\ref{clm:beta_frac_recov}, it follows that $ |\Dscr_\delta(g_n,s_n) \cap U_{\Pbar}^n  | \leq (1-\beta^*) |U_{\Pbar}^n | $ for all types $ \Pbar $ in  $  \Typepx $. % for all $ \Pbar \in \Pscr_{n}(T_{\prbX,\delta}^n), \Pbar \neq P $, we have $ |\Dscr_\delta(g_n,s_n) \cap U_{\Pbar}^n  | \leq (1-\beta) |U_{\Pbar}^n | $.  we have that for all type classes $ U_{P'}^n \subseteq  \Typepx $,  $ |\Dscr_\delta(g_n,s_n) \cap U_{\Pbar}^n  | \leq (1-\beta) |U_{\Pbar}^n | $.
%Assuming that $ U_P^n \subseteq \Typepx $,
We can write
\begin{align}
  \Pbb(\Dscr_\delta(g_n,s_n)) %&= \Pbb\left( \Dscr_\delta(g_n,s_n) \cap \Typepx  \right) + \Pbb\left( \Dscr_\delta(g_n,s_n) \cap (\Xscr^n \setminus \Typepx  )\right) \non \\
                              &\leq \Pbb\left( \Dscr_\delta(g_n,s_n) \cap (\Typepx \setminus U_{P^*}^n) \right) + \Pbb\left( \Dscr_\delta(g_n,s_n) \cap U_{P^*}^n \right) + \Pbb\left( \Xscr^n \setminus \Typepx  \right) \non \\
  &\leq  1- \beta^* + \Pbb\left( U_{P^*}^n \cup \left(\Xscr^n \setminus \Typepx \right)  \right). \non %  \leq 1-\beta + \epsilon', \non
\end{align}
Taking the limit, we get that
\begin{align}
  \lim_{n \rarr \infty} \max_{s_n' \in \best(g_n)}\Escr_\delta(g_n,s_n') \geq 1- \lim_{n \rarr \infty} \min_{s_n' \in \best(g_n)}\Pbb(\Dscr_\delta(g_n,s_n')) \geq \beta^*. \non
\end{align}
For correct recovery of sufficient number of sequences from the high probability typical set $ \Typepx $, the set $ I^n $ must include sequences from a $\delta$-ball around $ \Typepx $ and hence $I^n \cap U_{\Pbar}^n  \neq \emptyset $ for some $  \Pbar \in \Pscr_n(T_{\prbX,2\delta}^n) $. However, if $ \beta^* > 0 $, then the error in the limit is always positive. %However,  $ \beta = 0 $ implies that for all $ P' \in \Pscr_n(T_{\prbX,2\delta}^n) $ and hence, due to \eqref{eq:delta_ball_in_P2delta}, for all types in the $\delta$-ball around $ \Typepx$, we have $|I^n \cap U_{P'}^n | = \emptyset $. But .
%\hltodo{}{unclear} \tododone{done} %If $ U_P^n \cap \Typepx = \emptyset $, then the above probability is even lower.
% Thus,  if the set $ I^n $ has a positive fraction of sequences from any type class in $ \Pscr_n(T_{\prbX,2\delta}^n) $, then the probability of error is always bounded away from zero.
Thus the error does not tend to zero for  any sequence of strategies $ \{g_n\}_{n \geq 1 } $  and hence  there does not exist any achievable sequence of strategies. This completes the proof.
\end{proof}

\subsection{Proof of sufficiency of $ \Gamma(\ut) < 0 $ and part a) and b) of  Theorem~\ref{thm:lossl_rate_charac_gen}}%$ \Rscr^{inf} = \entrX$}%and characterization of rate region }
      %\subsection{Proof of Theorem~\ref{thm:lossl_rate_charac_gen}}
\label{appen:thm_loss_rate_charac_gen}
%Before proving Theorem~\ref{thm:lossl_rate_charac_gen} a) and b),
We will use Theorem~\ref{thm:dist_rate_bound} for the proof and we will proceed in a similar manner to the proof of  Theorem~\ref{thm:lossl_rate_charac} b). We first state a  lemma about joint distributions that have the same marginals. It is proved in \cite{vora2024shannon} (Lemma~4.3) and  can also be deduced from Lemma~1 in \cite{renault2013dynamic}. %\tododone{not saying much here since a short proof and similar to binary case}%induced by sequences from the same type class.
    \begin{lemma}\label{lem:y-same-type}
Let $ \ut $ be such that $\Gamma(\ut) < 0$. Then, for all $ P_{X,Y} $ where $ P_X = P_Y $, we have $  \Ebb_{P}\ut(X,Y) \leq 0. $ Moreover, $  \Ebb_{P}\ut(X,Y) = 0 $ if and only if $ P_{X,Y}(i,i) = P_X(i)  $ for all $ i \in \Xscr $.
\end{lemma}

\begin{proof}[Proof of Theorem~\ref{thm:lossl_rate_charac_gen} a)]
  $\vphantom{new}$\newline
  1) $ \Rscr^{\inf} = \entrX $

Recall the proof of  Theorem~\ref{thm:lossl_rate_charac} a). Taking $ \Pcode = \prbX  $ and $ R = \entrX $ we get
  \begin{align}
     \Wscr(\prbX,\Pcode,\entrX)
    &= \Big\{ \Ptilde_{\Xhat,X} \in \Pscr(\Xscr \times \Xscr) : \Ptilde_X = \prbX, \Ptilde_{\Xhat} = \prbX \Big\}. \non% \non \\
%                           & \hspace{3.8cm} H(Y|X) \geq 0 \Big\}.
  \end{align}
  % Thus, the set $ \Dscr(\prbX,\prbX,R) $ is a subset of all doubly-stochastic matrices.
  Since $ \Gamma(\ut) < 0 $, it follows from Lemma~\ref{lem:y-same-type}
 % \begin{align}
$\max_{\Ptilde \in \Wscr(\prbX,\prbX,\entrX) }  \Ebb_{\Ptilde} [\ut(\Xhat,X)] = 0$, % \non
 % \end{align}
  which is attained only by the distribution $ \Ptilde_{\Xhat,X}^* $ where $   \Ptilde_{\Xhat,X}^*(i,i) = \prbX(i)$ for all $ i \in \Xscr $. Thus, the set $ \Fscr = \argmax_{\Ptilde\in \Wscr(\prbX,\Pcode,\entrX) } \;\Ebb_{\Ptilde}[\ut(\Xhat,X)] $ only contains the distribution $ \Ptilde_{\Xhat,X}^* $. With this, we get that $ \Dbar(\entrX) =  \max_{P \in \Fscr} \Ebb_{P} d(\Xhat,X) = 0 $. Using Theorem~\ref{thm:dist_rate_bound}, there exists an achievable sequence of strategies such that $ R_0 \leq \entrX$ and $ \bd_0 = 0 $. This implies $ \Rscr^{\inf} = \entrX $.
\end{proof}
% The result follows from Corollary~\ref{coro:rate_bnd_dist_rate}.

  % 2) $ \Rscr^{\max} = \log \;\Xi(\ut,\prbX)$

  % \textbf{We need the larget set that takes the probability to 1, not the highest weight independent set. Fix eps. and del., what is the largest ind set that helps recover 1-eps. probability of sequences? }
\begin{proof}[Proof of Theorem~\ref{thm:lossl_rate_charac_gen} b)]

The proof is identical to the proof of Theorem~\ref{thm:lossl_rate_charac} c).
\end{proof}

\section{Proofs of Lossy Recovery : General Alphabet }
\label{append:gen_alph_lossy}

% \begin{proof}[\textit{of Theorem~\ref{thm:achie-rate-lossy_gen_alph}}]
% Can take a couple of i,j and add remove d. Then three of them and add remove d/2? or d/3? Increase and possibly can add remove d/q?
% \end{proof}

We will prove the result using Theorem~\ref{thm:dist_rate_bound}. Recall the set $ \Wscr(\prbX,\Pcode,R) $ defined for a distribution $ \Pcode $ and a rate $ R $. We will construct a $ \Pcode $ and show that the joint distribution from the set  $ \Wscr(\prbX,\Pcode,R) $ maximizing the expected utility of the sender is such that the worst-case expected distortion $ \Dbar(R) = \bd $. Using Theorem~\ref{thm:dist_rate_bound}, we will then get the existence of an achievable sequence of strategies such that $ R_0 \leq R $ and $ \bd_0 = \bd $. Appropriately choosing the distribution $ \Pcode $ will give us the upper bound on the $ \Rscr^{\inf} $. %\tododone{added here}% The claim will then follow by varying $ \delta \in [0,\bd/(q-1)] $.

We first state and prove a necessary condition for the joint distribution maximizing the expected utility.
For a distribution $ \Pcode $ and a rate $ R $, define
%$ P_{\Xhat,X}^* $ be given as % an optimal joint distribution given as %of the problem
%\begin{align}
$P_{\Xhat,X}^* \in  \argmax_{P \in \Wscr(\prbX,\Pcode,R)} \Ebb_{P} \ut(\Xhat,X). $ %\non
%\end{align}
\begin{lemma} \label{lem:no-permutation}
  Let $ \ut $ be such that $ \Gamma(\ut) < 0 $.
  %Then,  does not exist matrices $ V \in \Real_+^{q \times q}, Q \in \Qscr \setminus \{\Iscr\} $ and $ \alpha > 0 $ such that
%\begin{align}
%  P_{\Xhat,X}^* = V + \alpha Q. \label{eq:P*_sum_perm_mat}
%\end{align}
Then,
\begin{align}
&  P_{\Xhat,X}^*(i_0,i_1) P_{\Xhat,X}^*(i_1,i_2) \hdots P_{\Xhat,X}^*(i_{l-1},i_l) P_{\Xhat,X}^*(i_l,i_0) = 0, \non % \non \\
\end{align}
 $ \forall \;i_0,i_1,\hdots,i_l \in \Xscr $ and $\forall \;l \in \{1,\hdots,q-1\}$, where $ i_k \neq i_m $ for all $ k \neq m $.% when $|k-m| = 1$. %\label{eq:no-permutation}
\end{lemma}
\begin{proof}
  We prove this by contradiction. Suppose there exists distinct symbols $ \forall \;i_0,i_1,\hdots,i_l \in \Xscr $, $ l \leq q-1 $  such that
%  \begin{align}
  $P_{\Xhat,X}^*(i_0,i_1) P_{\Xhat,X}^*(i_1,i_2) \hdots P_{\Xhat,X}^*(i_{l-1},i_l) P_{\Xhat,X}^*(i_l,i_0) > 0.$  %\non % \non \\
%\end{align}
Let $ \alpha = \min\{ P_{\Xhat,X}^*(i_0,i_1), P_{\Xhat,X}^*(i_1,i_2), \hdots, P_{\Xhat,X}^*(i_l,i_0)\} $. Define a distribution $ P' $ as
\begin{align}
  P'(j,k) = \left\{
  \begin{array}{c l}
     P_{\Xhat,X}^*(j,k) - \alpha &  \mbox{if } (j,k) \in \{ (i_0,i_1) ,\hdots, (i_l,i_0) \} \\
     P_{\Xhat,X}^*(j,j) + \alpha & \mbox{if } j \in \{ i_0,i_1,\hdots,i_l\} \\
     P_{\Xhat,X}^*(j,k) & \mbox{ otherwise }
  \end{array}
  \right. \non
\end{align}
Clearly, $ P' \in \Wscr(\prbX,\Pcode,R) $. This gives that
\begin{align}
  \Ebb_{P'} \ut(\Xhat,X) &= \Ebb_{P^*} \ut(\Xhat,X)  -  \alpha ( \ut(i_0,i_1)  + \hdots + \ut(i_l,i_0)). \non
\end{align}
The summation on the right side is
%\begin{align}
 $ \ut(i_0,i_1)  + \hdots + \ut(i_l,i_0) = \sum_{j,k}Q(j,k)\ut(j,k)$,  % \non
%\end{align}
for some permutation matrix $ Q \in \Qscr \setminus \{\Iscr\} $.
Since $ \Gamma(\ut) < 0 $, it follows that
%\begin{align}
 $ \Ebb_{P'} \ut(\Xhat,X) = \Ebb_{P^*} \ut(\Xhat,X)  - \sum_{j,k}Q(j,k)\ut(j,k)  %\non \\
  > \Ebb_{P_{\Xhat,X}^*} \ut(\Xhat,X)$. %\non
%\end{align}
This is a contradiction since $ P_{\Xhat,X}^* $ is the optimal distribution. Thus our assumption is false and the claim is proved. %for all distinct symbols $ i_0,i_1,\hdots,i_l \in \Xscr $, $ 1 \leq l \leq q-1 $, we have
%\begin{align}
% $ P_{\Xhat,X}^*(i_0,i_1) P_{\Xhat,X}^*(i_1,i_2) \hdots P_{\Xhat,X}^*(i_{l-1},i_l) P_{\Xhat,X}^*(i_l,i_0) = 0$. %\non
%\end{align}
\end{proof}
% We prove this by showing that $ P_{\Xhat,X}^* $ cannot be written as
% \begin{align}
%   P_{\Xhat,X}^* = V + \alpha Q, \label{eq:P*_sum_perm_mat}
% \end{align}
% where $ V \in \Real_+^{q \times q}, Q \in \Qscr \setminus \{\Iscr\} $ and $ \alpha > 0 $.

% We prove this by contradiction. Suppose $ P_{\Xhat,X}^* $ is given as \eqref{eq:P*_sum_perm_mat}. Define a matrix $ P' $ as $ P' = V + \alpha \Iscr $. Clearly, $ P' \in \Wscr(\prbX,\Pcode,R) $. Further,
% \begin{align}
%   \Ebb_{P'}\ut(\Xhat,X) &= \sum_{i,j}V(i,j)\ut(i,j) \non \\
%   &=   \Ebb_{P_{\Xhat,X}^*}\ut(\Xhat,X) - \alpha\sum_{i,j}Q(i,j)\ut(i,j). \non
% \end{align}
% Since $ \Gamma(\ut) < 0 $, we have $  \Ebb_{P'}\ut(\Xhat,X) >  \Ebb_{P_{\Xhat,X}^*}\ut(\Xhat,X) $ which is a contradiction since $ P_{\Xhat,X}^* $ is optimal. Thus, $ P_{\Xhat,X}^* $ cannot be given as \eqref{eq:P*_sum_perm_mat}. %there does not exist

% A consequence of this fact is that for all distinct symbols $ i_0,i_1,\hdots,i_l \in \Xscr $, $ 1 \leq l \leq q-1 $
% \begin{align}
%   P_{\Xhat,X}^*(i_0,i_1) P_{\Xhat,X}^*(i_1,i_2) \hdots P_{\Xhat,X}^*(i_{l-1},i_l) P_{\Xhat,X}^*(i_l,i_0) = 0. \non
% \end{align}

%\subsection{Proof of Theorem~\ref{thm:achie-rate-lossy_gen_alph}}

We now proceed with the proof of Theorem~\ref{thm:achie-rate-lossy_gen_alph}. Let $ \delta \leq \bd/(q-1) $ and define $ \Pcode $ as
\begin{equation}
  \begin{aligned}
    \Pcode(0) &= \prbX(0) - \delta,  \\
    \Pcode(q-1) &= \prbX(q-1) + \delta, \\
    \Pcode(i) &= \prbX(i) \quad \forall \; i \neq 0,q-1.
  \end{aligned} \label{eq:defn_P0}
\end{equation}
  Here, the choice of the symbols $ 0 $ and $ q-1 $ is arbitrary and is chosen for ease of exposition.   Let $ R = H(\Pcode) $. Thus, we have
\begin{align}
    \Wscr(\prbX,\Pcode,H(\Pcode)) %\non \\
    &= \Big\{ \Ptilde_{\Xhat,X} \in \Pscr(\Xscr \times \Xscr) : \Ptilde_{\Xhat} = \Pcode, \Ptilde_X = \prbX \Big\}. \non% \non \\
%                           & \hspace{3.8cm} H(Y|X) \geq 0 \Big\}.
\end{align}
Thus, the set $ \Wscr(\prbX,\Pcode,H(\Pcode)) $ consists of joint distributions such that the $ X $ marginal is equal to the source distribution and the $ \Xhat $ marginal  differs from $ \prbX $ only in the symbols $ 0 $ and $ q-1 $.

\begin{proof}[Proof of Theorem~\ref{thm:achie-rate-lossy_gen_alph} a)]

As mentioned earlier, we use Theorem~\ref{thm:dist_rate_bound} to prove the claim. We show that $ \Dbar(H(\Pcode)) = \max_{P\in \Fscr} \Ebb_P d(\Xhat,X) \leq (q-1)\delta $ where \newline $ \Fscr = \argmax_{\Ptilde \in \Wscr(\prbX,P_0,H(\Pcode))} \Ebb_{\Ptilde}[\ut(\Xhat,X)] $.  Here $ \Wscr(\prbX,\Pcode,H(\Pcode)) $ is as above and $ P_0 $ is defined in  \eqref{eq:defn_P0}. Thus, from Theorem~\ref{thm:dist_rate_bound}, there exists an achievable sequence of strategies such that $ R_0 \leq H(\Pcode) $ and $ \bd_0 = \Dbar(H(\Pcode)) $.  The claim will then follow by varying $ \delta \in [0,\bd/(q-1)] $.

  We proceed by showing that any worst-case optimal joint distribution  \newline $  P_{\Xhat,X}^* \in  \argmax_{P\in \Fscr} \Ebb_P d(\Xhat,X) $ is such that  %\hltodo{}{How does showing this help? What is the logic? Where does this result lead?}\tododone{done} %for the sender from the set $  \Wscr(\prbX,\Pcode,H(\Pcode))
%\begin{align}
 $ \sum_{i \neq j} P_{\Xhat,X}^*(i,j) \leq (q-1)\delta$. % \non %\quad   \sum_{j \neq i} P_{\Xhat,X}^*(i,j) \leq \delta. \non
%\end{align}
  As in earlier proofs, we proceed in steps where we prove intermediate claims leading to the final result.  Let $ P_{\Xhat,X}^* \in \Fscr $ be a worst-case optimal joint distribution. First we show that we can construct a matrix $ P' $ by permuting the rows of $ P_{\Xhat,X}^*$ such that $ P' $ is lower triangular matrix. %We proceed in \hltodol{following steps}{this proof is an aimless series of arguments which no one can figure where they are headed. Extremely poorly written.}.
\begin{claim} \label{clm:perm_of_P*}
There exists a permutation of columns of $ P_{\Xhat,X}^* $ such that the resulting matrix is a lower triangular matrix.
\end{claim}
\begin{proof}
 We first show that $ P_{\Xhat,X}^*(0,j) = 0 $ for all $ j \neq 0 $ and $ P_{\Xhat,X}^*(i,q-1) = 0 $ for all $ i \neq q-1 $. %Thus, for the distribution $ P_{\Xhat,X}^* $,  the off-diagonal entries in the column corresponding to the row index $ 0 $ is zero  and the off-diagonal entries in the row corresponding to the column index $ q-1$ is zero.

  Let $ j_0 = 0 $ and suppose for some $ j_1 \neq j_0 $, we have $ P_{\Xhat,X}^*(j_0,j_1) > 0 $, \ie, an off-diagonal term in the $ j_0^{th}$ row is non-zero. In other words, an off-diagonal term in the $j_1^{th}$ column is non-zero.  From \eqref{eq:defn_P0}  and the definition of $ \Wscr(\prbX,\Pcode,H(\Pcode)) $, except for the index $ 0 $, the column marginal distribution is point-wise smaller than the row marginal distribution of $ P_{\Xhat,X}^*$. Thus, $ P_{X}^*(j_1) \leq P_{\Xhat}^*(j_1) $, and hence there must exist a non-zero off-diagonal term in the $j_1^{th}$ row. Thus, there exists a symbol $ j_2 \in \Xscr, j_2 \neq j_1 $ such that  $ P_{\Xhat,X}^*(j_1,j_2) > 0 $. It also follows that $ j_2 \neq j_0 $ since  $ j_2 = j_0 $ gives $ P_{\Xhat,X}^*(j_0,j_1)P_{\Xhat,X}^*(j_1,j_0) > 0 $ which is a contradiction due to Lemma~\ref{lem:no-permutation}.   %This proves that $ j_2 \neq 0 $.
  % Define $ \alpha = \min\{ P_{\Xhat,X}^*(0,i_1), P_{\Xhat,X}^*(i_1,0)\}$ and consider the permutation matrix $ Q $ where $ Q(0,i_1) = Q(i_1,0) = 1 $ and $ Q(i,i) = 1 $ for all $ i \neq 0,i_1 $.  Then, we can write $ P_{\Xhat,X}^* = V + \alpha Q $  for some $ V \in \Real_+^{q \times q } $.

  Again, from $  P_X^*(j_2) \leq P_{\Xhat}^*(j_2) $ it follows that $ P_{\Xhat,X}^*(j_2,j_3) > 0 $ for some $ j_3 \neq j_2,j_1,j_0 $. We can continue similarly till we reach the $ q^{th} $ symbol $ j_{q-1}$  where $  P_{\Xhat,X}^*(j_{q-2},j_{q-1}) > 0 $ and $ j_{q-1} \neq j_{q-2},\hdots,j_1,j_0 $. Again, $ P_{X}(j_{q-1}) \leq P_{\Xhat}^*(j_{q-1}) $ holds and there exists a symbol $ j_q $ such that $ P_{\Xhat,X}^*(j_{q-1},j_q) > 0 $. However, $ j_q = j_k $ for some $ k \in \{0,1,\hdots,q-1\} $. Thus, we have a sequence $ \{j_k,j_{k+1}, \hdots, j_{q-1},j_q = j_k\} $ such that
  %\begin{align}
   $ P_{\Xhat,X}^*(j_k,j_{k+1}) \hdots P_{\Xhat,X}^*(j_{q-1},j_k)  > 0.$ %\non%\label{eq:joint-type-ineq}
% \end{align}
  However, this is a contradiction due to Lemma~\ref{lem:no-permutation}. This proves that our assumption  of existence of a $ j_1 \neq j_0 = 0 $, such that $ P_{\Xhat,X}^*(0,j_1) > 0 $ is false and hence for all  $ j \neq 0 $, $ P_{\Xhat,X}^*(0,j) = 0 $.
  We can use a similar procedure as above to show that for all $ i \neq q-1 $, $ P_{\Xhat,X}^*(i,q-1) = 0 $.
                              %                               $ \sum_{i \neq q-1}  P_{\Xhat,X}^*(i,q-1) \geq \delta $

                              %                               \textbf{Step 2:} In this step, we show that there is
  We now permute the columns of $ P_{\Xhat,X}^* $ to construct a lower triangular matrix. %such that the resulting matrix is lower triangular.
  Consider the column corresponding to the index $ 0 $. We have that $ P_X^*(0) > P_{\Xhat}^*(0) $ and hence there exists $ i \neq 0 $ such that $ P_{\Xhat,X}^*(i,0) > 0 $. Moreover,
 % \begin{align}
   $ \Big| \supp( P_{\Xhat,X}^*(\cdot,0)) \Big| \leq q$. % \non
  %\end{align}
  Take some $ i_1 \in \supp( P_{\Xhat,X}^*(\cdot,0)) $. Due to Lemma~\ref{lem:no-permutation} it follows that $ P_{\Xhat,X}^*(0,i_1) = 0 $ and hence $ \Big| \supp( P_{\Xhat,X}^*(\cdot,i_1)) \Big| \leq q-1 $. We can continue similarly till we exhaust all the symbols. This would give a sequence $ i_0,i_1,\hdots,i_{q-1} $ with $ i_0 = 0 $ and $ i_{q-1} = q-1 $ such that
  \begin{align}
    \Big| \supp( P_{\Xhat,X}^*(\cdot,i_k)) \Big| \leq q-k \quad \forall \; k \leq q-1. \label{eq:supp_less_q-k+1}
  \end{align}
Thus, we have identified columns indexed with $ i_0, i_1, \hdots, i_{q-1}$ where the number of non-zero entries in the rows of consecutive columns is decreasing.  Let $ a_m $ denote the $ m^{th}$ column of the matrix $ P_{\Xhat,X}^* $. Define a matrix $ P' $ be rearranging the columns of  $ P_{\Xhat,X}^* $ as
 % \begin{align}
   $ P' = \left( a_{i_0} \; a_{i_1} \; \hdots \; a_{i_{q-1}}
    %  \begin{array}{c}
    %   a_{i_0} \\
    %   a_{i_1} \\
    %   \vdots \\
    %   a_{i_{q-1}}
    % \end{array}
    \right)$.  %\non
%  \end{align}
  % Clearly, $ P' $ is a permutation of $ P_{\Xhat,X}^* $. Moreover,
  From \eqref{eq:supp_less_q-k+1} it follows that $ P' $ is a lower triangular matrix.
\end{proof}

The previous step shows that there is a rearrangement of the columns of $ P_{\Xhat,X}^* $ that gives a lower triangular matrix. Without loss of generality assume that $ P_{\Xhat,X}^* $ is itself  an lower triangular matrix. Recall that we wish to bound  $\Ebb_{P_{\Xhat,X}^*} d(\Xhat,X) $ and  it suffices to show the following claim.%matrix, we now have the following claim.
\begin{claim} \label{clm:off_diag_delta}
  The sum of all the non-diagonal terms of the rows of $ P_{\Xhat,X}^n $ is at most $ (q-1)\delta $.
\end{claim}
\begin{proof}
  We write % $ P_{\Xhat,X}^* $ as
 % \begin{align}
   $ P_{\Xhat,X}^* = W_D + W_F$,  % \non
 % \end{align}
  where
  \begin{align}
    W_F(i,j) = \left\{
    \begin{array}{c l}
      P_{\Xhat,X}^*(i,j) & \mbox{ if } i \neq j \\
      0 & \mbox{ else }
    \end{array}
          \right.. \non
  \end{align}
  Thus, $ W_D $ comprises of the diagonal of $ P_{\Xhat,X}^* $ and $ W_F $ is the lower triangular part of $ P_{\Xhat,X}^* $.  %  except the diagonal.
  From the first part of the proof of the Claim~\ref{clm:perm_of_P*} and from the definition of $ P_0 $ in \eqref{eq:defn_P0}, it follows that $ \sum_{i \neq 0} W_F(i,0) = \delta = \sum_{j \neq q-1} W_F(q-1,j) $ and $ \sum_{j \neq 0} W_F(0,j) = 0 = \sum_{i \neq q-1} W_F(i,q-1)  $.

  Consider the $i^{th}$ row of $  W_F $ where $ 1 \leq i \leq q-2$. We denote the sum of the elements in the row as $ \Theta(i) $ which is given as
  \begin{align}
    \Theta(i) = W_F(i,0) +W_F(i,1) + \hdots +W_F(i,i-1). \label{eq:ith_row_sum}
  \end{align}
  Now using the fact that $ P_{\Xhat}^*(j) = P_X^*(j) $ for all $ j \neq 0,q-1 $ (cf defn. of $P_0$ in \eqref{eq:defn_P0}), we have
  %\begin{align}
   $ \sum_{k > i-1}W_F(k,i-1) = \sum_{j < i-1}W_F(i-1,j) $  % \non% +W_F(i-1,1) + \hdots +W_F(i-1,i-2)
 % \end{align}
  and hence
  %\begin{align}
  $W_F(i,i-1) \leq \sum_{j < i-1}W_F(i-1,j)$. %\non% +W_F(i-1,1) + \hdots +W_F(i-1,i-2)
  %\end{align}

  % we bound $W_F(i,i-1)  $ as
  % \begin{align}
  %   W_F(i,i-1) \leq W_F(i-1,0) +W_F(i-1,1) + \hdots +W_F(i-1,i-2). \non
  % \end{align}
  Using this, we write \eqref{eq:ith_row_sum} as
  \begin{align}
    \Theta(i)  &\leq
                 \begin{array}{l l l l}
                   \quad  W_F(i,0) &+W_F(i,1) &+ \hdots &+W_F(i,i-2)  \\
                   + \;W_F(i-1,0) &+W_F(i-1,1) &+ \hdots &+W_F(i-1,i-2).
                 \end{array} \non
  \end{align}
%  Again using $ P_{\Xhat}^*(j) = P_X^*(j) $ we bound $W_F(i,i-2) +W_F(i-1,i-2) $ to get
%  \begin{align}
%    \Theta(i)  &\leq
%                 \begin{array}{l l l l}
%                   \quad  W_F(i,0) &+W_F(i,1) &+ \hdots &+W_F(i,i-3)  \\
%                   + \;W_F(i-1,0) &+W_F(i-1,1) &+ \hdots &+W_F(i-1,i-3) \\
%                   + \;W_F(i-2,0) &+W_F(i-2,1) &+ \hdots &+W_F(i-2,i-3).
%                 \end{array} \non
%  \end{align}
  We can continue bounding the terms similarly till we get
  %\begin{align}
  $  \Theta(i) \leq \sum_{j \leq i} W_F(j,0) \leq \delta$,  % \non
  %\end{align}
  where the last inequality follows since the column sum corresponding to the index $ 0 $ is $ \delta $. This implies that all the rows of $ W_F$ sum to a value less than $ \delta $.
\end{proof}

Now we compute the distortion induced by the distribution $ P_{\Xhat,X}^* $. From Claim~\ref{clm:off_diag_delta}, we have
%\begin{align}
 $ \sum_{i\neq j} P_{\Xhat,X}^*(i,j)d(i,j) =  \sum_{i\neq j} W_{F}^*(i,j) d(i,j) \leq (q-1)\delta.$ %\non
%\end{align}
Since $ \delta \leq \bd/(q-1) $, we get that $ \sum_{i\neq j} P_{\Xhat,X}^*(i,j)d(i,j) \leq \bd $. Thus, from Theorem~\ref{thm:dist_rate_bound}, there exists an achievable sequence of strategies with rate $ R \leq H(\Pcode) $ and $ \bd_0 =  \Dbar(H(\Pcode)) $. Recall that $ \Pcode $ was defined as \eqref{eq:defn_P0} for the symbols $ 0 $ and $ q-1$ which was arbitrary. Thus, varying over all pairs of symbols and  choosing $ \delta \leq \bd/(q-1)$, we get  %\hltodo{}{what is $\Pcode$? And $X'$?} \tododone{done, and changed the entropy notation throughout}
%\begin{align}
 $ \Rscr_{\bd}^{\inf} \leq \min_{ P' \in \Pscr'} H(P')$, % \non %\\
   %\Rscr_{\bd}^{\sup} &\geq \max_{P' \in \Pscr'} H(P'). \non
%\end{align}
%\tododone{removed the result about $ \Rscr_{\bd}^{\sup} $, will need a result like the Claim~\ref{clm:H(P+d)_achievable} }
 where $  \Pscr' = \bigcup_{j,k \in \Xscr} \Pscr'_{jk}$. %This completes the proof.
\end{proof}

\begin{proof}[Proof of Theorem~\ref{thm:achie-rate-lossy_gen_alph} b)]
  We have
   %\begin{align}
     $\ut_n(\xvec,\yvec)  = \frac{1}{n}\sum_{i: x_i \neq y_i} \ut(x_i,y_i) = -d_n(\xvec,\yvec)c$,   % \non%, \non
   %\end{align}
   and hence, maximizing the utility leads to minimizing the distortion. Thus, the sender's objective is aligned with that of the receiver and hence the receiver can achieve a rate $ R(\bd) $. Moreover, since $ \ut_n(\xvec,\yvec) < 0 $ for all distinct sequences $ \xvec, \yvec \in \Xscr^n $, using the arguments from the proof of Theorem~\ref{thm:lossl_rate_charac} part c), we can show that the rate $ \log q $ is also achievable. % the graph $ \Gs^n $ contains only isolated vertices, thereby taking $ I^n = \Xscr^n $ we get that the rate $ \log q $ is also achievable.
 \end{proof}

\bibliographystyle{IEEEtran}
%\bibliography{ref.bib}
\bibliography{../../../../../Mylatexfiles/ref.bib}

\end{document}